\pdfoutput=1

\documentclass[%
    reprint,
    showpacs,preprintnumbers,
    showkeys,
    amsmath,amssymb,
    aps,
    prx,
    floatfix,
    longbibliography,
    notitlepage,
]{revtex4-2}

\usepackage[english]{babel}
\usepackage[utf8]{inputenc}
\usepackage{graphicx}
\usepackage{bm}
\usepackage{hyperref}
\usepackage{xcolor}
\usepackage{epsfig}
\usepackage{amsmath}
\usepackage{amsfonts}
\usepackage{amssymb}
\usepackage{amsthm}
\usepackage{gensymb}
\usepackage{enumerate}
\usepackage{tikz}
\usepackage{dsfont}
\usepackage[all]{nowidow}
\usepackage{pifont}
\usepackage{algorithm}
\usepackage{algpseudocode}
\usepackage[capitalise]{cleveref}
\usepackage[export]{adjustbox}
\usepackage{xspace}
\usepackage{nicefrac}

\let\doi\undefined
\usepackage{doi}

\usepackage[letterspace=130]{microtype}

\definecolor{refColor}{HTML}{0376E9}
\definecolor{figColor}{HTML}{E90303}
\definecolor{urlColor}{HTML}{0376E9}

\hypersetup{%
    unicode=false,                 
    pdftoolbar=true,               
    pdfmenubar=true,               
    pdffitwindow=false,            
    pdfstartview={FitH},           
    pdftitle={},                   
    pdfauthor={Nicolai Lang},      
    pdfsubject={Quantum physics},  
    pdfcreator={Nicolai Lang},     
    pdfproducer={Nicolai Lang},    
    pdfkeywords={},                
    pdfnewwindow=true,             
    colorlinks=true,               
    linkcolor=figColor,            
    citecolor=refColor,            
    filecolor=magenta,             
    urlcolor=urlColor              
}

\newcommand{\bra}[1]{\mathinner{\langle{#1}|}}
\newcommand{\ket}[1]{\mathinner{|{#1}\rangle}}

\renewcommand{\vec}[1]{\mathbf{#1}}

\renewcommand{\vec}[1]{\boldsymbol{#1}}

\newcommand{\com}[2]{\left[#1,#2\right]}

\renewcommand{\H}{\mathcal{H}}

\renewcommand{\L}{\mathcal{L}}
\newcommand{\C}{\mathcal{C}}

\renewcommand{\u}{\mathfrak{u}}

\newcommand{\GF}{\mathbb{F}_2}
\newcommand{\spn}[1]{\operatorname{span}\left\{\,#1\,\right\}}
\newcommand{\sub}[1]{{\text{\tiny\textnormal{#1}}}}
\newcommand{\vep}{\varepsilon}
\newcommand{\id}{\mathds{1}}

\newcommand{\NOR}{\mathtt{NOR}}
\newcommand{\eref}[1]{(\ref{#1})}

\newtheorem{theorem}{Theorem}

\newtheorem{lemma}{Lemma}

\definecolor{new}{HTML}{b533ff}

\tikzstyle{roundlabel}=[circle,inner sep=1.5pt,fill=tikzgrey,color=black,text=white,draw,font=\sffamily\footnotesize]
\tikzstyle{TlabelB}=[rectangle callout,rounded corners=0.03cm,inner sep=3.0pt,fill=tikzlight,color=tikzlight,text=tikzgrey,draw,font=\sffamily\footnotesize,callout absolute pointer={#1},at={#1},above=0.2cm,align=center]
\tikzstyle{smallbullet}=[circle,inner sep=0.5pt,fill=tikzorange,color=tikzorange,text=white,draw,font=\sffamily\small\bfseries]
\tikzstyle{labelline}=[rounded corners=0.5mm,black,line width=2pt,solid]
\tikzstyle{labeltext}=[color=white,fill=white,inner sep=0.5pt,text=tikzgrey,draw,font=\sffamily\footnotesize]
\tikzstyle{math}=[font=\Large,align=center]
\tikzstyle{figure}=[anchor=center]
\tikzstyle{label}=[circle,inner sep=0.07cm,fill=white,text=black,draw=white]
\tikzstyle{tlabel}=[inner sep=0.07cm,text=black]
\tikzstyle{wlabel}=[inner sep=0.07cm,text=black,anchor=west,font=\footnotesize]
\tikzstyle{figtext}=[anchor=west, rounded corners=0.0cm,inner sep=0.07cm,fill=white,color=white,text=black,draw,font=\sffamily\footnotesize]

\definecolor{sublatticecolor}{HTML}{CCCCCC}

\newcommand{\lbl}[1]{\bf\sffamily\small(#1)}

\newcommand{\gH}{\H_\sub{T}}
\definecolor{optimized}{HTML}{e4000d}
\definecolor{unoptimized}{HTML}{000000}
\newcommand{\Rb}{{r_\sub{B}}}
\newcommand{\fT}{{f_\sub{T}}}

\graphicspath{{fig/}}

\begin{document}

\title{Functional completeness of planar Rydberg blockade structures}

\author{Simon Stastny}
\author{Hans Peter Büchler}
\author{Nicolai Lang}
\email{nicolai.lang@itp3.uni-stuttgart.de}
\affiliation{%
    Institute for Theoretical Physics III 
    and Center for Integrated Quantum Science and Technology,\\
    University of Stuttgart, 70550 Stuttgart, Germany
}

\date{\today}


\begin{abstract}
    The construction of Hilbert spaces that are characterized by local
    constraints as the low-energy sectors of microscopic models is an
    important step towards the realization of a wide range of quantum phases
    with long-range entanglement and emergent gauge fields. Here we show
    that planar structures of trapped atoms in the Rydberg blockade regime
    are functionally complete: Their ground state manifold can realize any
    Hilbert space that can be characterized by local constraints in the
    product basis. We introduce a versatile framework, together with a set
    of provably minimal logic primitives as building blocks, to implement
    these constraints. As examples, we present lattice realizations of
    the string-net Hilbert spaces that underlie the surface code and the
    Fibonacci anyon model. We discuss possible optimizations of planar
    Rydberg structures to increase their geometrical robustness.
\end{abstract}

\maketitle

\section{Introduction}

Recent advances in the control of single atoms and their coherent
manipulation~\cite{Schlosser2001,Saffman_2010,Nogrette2014,Barredo_2016,Barredo2018}
are the technological foundation for applications such as quantum
simulation~\cite{Weimer2010,Georgescu2014,Gross2017,Altman2021},
high-precision metrology~\cite{Degen2017,Pezze2018} and, hopefully, future
quantum computers~\cite{Ladd2010,Henriet2020,Graham2022,Bluvstein2022}.
For any of these applications, suitable platforms must offer fine-grained
control over their degrees of freedom, dynamically tunable interactions,
and the possibility to decouple the environment.
Promising in this regard are arrays of individually trapped, neutral atoms
that can be manipulated by optical tweezers~\cite{Schlosser2001,Nogrette2014}
and excited into Rydberg states~\cite{Gallagher2006,Sibalic2018}. These
exhibit strong interactions which lead to the Rydberg blockade mechanism
where excited atoms prevent their neighbors within a tunable radius from
being excited~\cite{Jaksch2000,Tong2004,Singer2004,Gaetan2009,Urban2009}.
In this paper, we study on very general grounds the theoretical capabilities
of the Rydberg platform in the blockade regime and demonstrate its versatility
by constructing the gauge-invariant Hilbert spaces of two models with Abelian
and non-Abelian topological order.

Encouraged by the fast development and scalability of the Rydberg platform
(see e.g.~Refs.~\cite{Ebadi2021,Scholl2021,Schymik2022}), there has been
increased interest in identifying promising near-term applications for the
NISQ era~\footnote{%
NISQ = Noisy Intermediate-Scale Quantum technology, i.e., near-term
quantum technology without full-fledged quantum error correction, see
Ref.~\cite{Preskill2018}.}.
Among the many applications of two-dimensional arrays of Rydberg atoms,
the field of \emph{geometric programming} and the design of \emph{synthetic
quantum matter} have been identified as promising candidates to leverage
the capabilities of available and upcoming NISQ platforms.

The rationale of geometric programming is the solution of algorithmic problems
by encoding them into the geometry of the atomic array.
This direction of research is founded on the insight that due to the
Rydberg blockade, the ground states of these systems naturally map
to \emph{maximum independent sets} (MIS) on so called \emph{unit
disk graphs}~\cite{Pichler2018b}; finding MIS is a long-known
optimization problem in graph theory that has been shown to be
\textsf{NP}-hard~\cite{Clark1990}. This makes the computation of ground state
energies of Rydberg arrangements \textsf{NP}-hard as well~\cite{Pichler2018a},
but also opens the possibility to tackle a variety of other hard optimization
problems~\cite{Serret_2020,Wurtz2022,Dalyac2022,Nguyen2022,Lanthaler2023,Jeong2023}
by polynomial-time reductions to the MIS problem~\footnote{Of
course one should not expect an exponential speedup by these
mappings as it is widely believed~\cite{Bennett1997} that
$\text{\textsf{NP}}\nsubseteq\text{\textsf{BQP}}$.}.
First solutions of MIS instances on various graphs in two
and three dimensions have been demonstrated in experiments
recently~\cite{Byun2022,Kim2022,Ebadi2022}, and a quantitative comparison
of experimental solutions with classical algorithms suggest a superlinear
quantum speedup for some classes of graphs~\cite{Ebadi2022}.

A very different application of the Rydberg blockade mechanism is
the engineering of synthetic quantum matter on the single-atom level
\cite{Celi2020}.
The potential of this approach has been demonstrated recently by
Verresen~\textit{et~al.}~\cite{Verresen_2021} (related results were
reported by Samajdar~\textit{et~al.}~\cite{Samajdar_2021}), who proposed
the realization of topological spin liquids on delicately designed lattice
structures of atoms. In this scenario, the Rydberg blockade enforces
a dimer constraint (the local gauge constraint of an odd $\mathbb{Z}_2$
lattice gauge theory~\cite{Moessner2001}) which, in combination with quantum
fluctuations, can give rise to long-range entangled many-body states with
Abelian topological order. First experimental results were reported shortly
after~\cite{Semeghini2021}, accompanied by theoretical studies of the used
quasiadiabatic preparation schemes~\cite{Giudici2022,Sahay2022}.

This paper is written from and motivated by the synthetic quantum matter
perspective, but its results apply to geometric programming as well. Our
starting point is the question whether other local constraints (besides the
dimer constraint) can be realized on the Rydberg platform.
To find an answer, we first formalize the problem and then use this formulation
to derive our main result, namely that \emph{every} local constraint that
can be encoded by a Boolean function can be implemented in the ground state
manifold of a planar arrangement of atoms in the blockade regime. Crucial for
this result is the existence of a structure that implements the truth table
of a \texttt{NOR}-gate (``Not OR'') in its ground state manifold. While our
proof is constructive, it does typically not yield optimal (= small) solutions.
We therefore expand on our main result and compile a comprehensive list
of provably minimal structures that realize all important primitives of
Boolean logic. Together with a structure that facilitates the crossing of
two ``wires'' within the plane, these primitives provide a toolbox to build
structures that satisfy more complicated constraints.
As an example, we construct a system with a ground state manifold that
is locally isomorphic to the gauge-invariant Hilbert space of an even
$\mathbb{Z}_2$ lattice gauge theory, i.e., the charge-free sector of the toric
code~\cite{Kitaev2003}. With a similar construction, we tailor a pattern
of atoms with a ground state manifold isomorphic to the string-net Hilbert
space of the ``golden string-net model~\cite{Levin2005}''; a system that,
with added quantum fluctuations, could support non-Abelian Fibonacci anyons.
Having constructed all these structures, we briefly discuss possibilities
to numerically optimize their geometries to make them more robust against
geometric imperfections and the effects of long-range van der Waals
interactions.

\emph{Note added.} When finalizing this manuscript we became aware of related
results \cite{Nguyen2022,Lanthaler2023}. The authors of both publications
focus on optimization problems and find some of the primitives discussed
in this paper. (The ring-shaped \texttt{NOR}-gate and the crossing is
found by Nguyen \textit{et~al.} \cite{Nguyen2022} and the triangle shaped
\texttt{XNOR}-gate by Lanthaler \textit{et~al.} \cite{Lanthaler2023}). Both
papers follow the rationale of \emph{geometric programming}, so that their
motivation, approach and framework differ from ours.

\section{Rationale and Outline}
\label{sec:overview}

\begin{figure}[tb]
    \centering
    \includegraphics[width=1.0\linewidth]{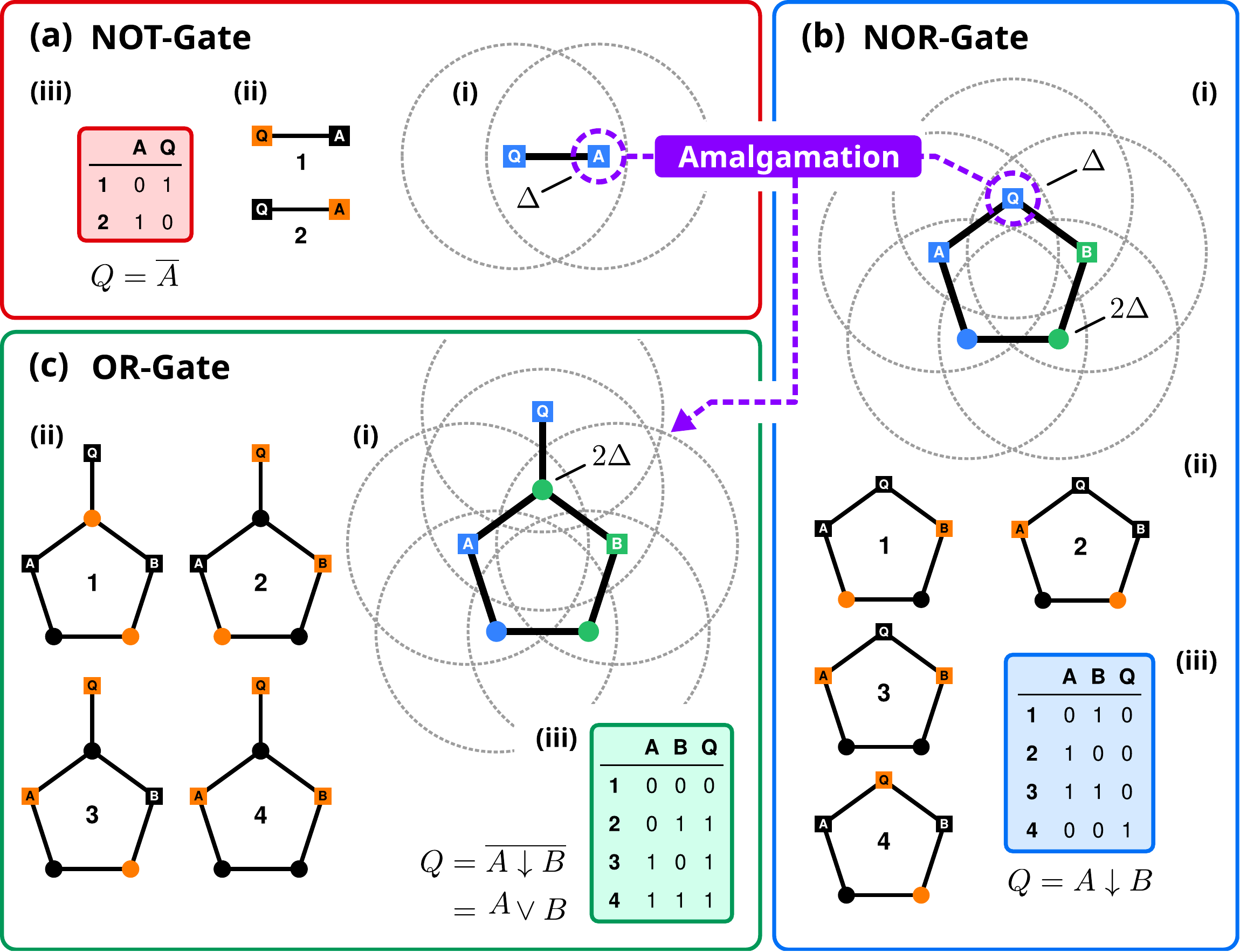}
    \caption{%
	\emph{Rationale.} 
    (a) Structure of two atoms (i) with local detunings $\Delta$ (blue
    vertices) that are in Rydberg blockade (gray circles); the blockade is
    indicated by a black edge connecting the atoms. The ground state manifold
    (ii) is given by patterns of excited atoms (orange) that minimize the
    energy; here it is two-fold degenerate. The two ground state configurations
    realize the truth table (iii) of a \texttt{NOT}-gate $Q=\overline{A}$.
    (b) Structure of five atoms (i) with local detunings $\Delta$ (blue)
    and $2\Delta$ (green) in a ring-like Rydberg blockade. The ground state
    manifold (ii) is four-fold degenerate. If one selects the three labeled
    atoms and identifies them with the columns of the table in (iii), the four
    ground state configurations realize the truth table of a \texttt{NOR}-gate
    $Q=A\downarrow B=\overline{A\vee B}$.
    (c) Joining the output atom of the \texttt{NOR}-gate with the input
    atom of the \texttt{NOT}-gate (and adding their detunings) yields a
    new structure that realizes the truth table of an \texttt{OR}-gate:
    $Q=\overline{A\downarrow B}=A\vee B$. This construction is called
    \emph{amalgamation}.
    }
    \label{fig:nutshell}
\end{figure}

Here we illustrate the rationale of the paper and provide a brief outline
of its main results without technical overhead. Readers interested in the
details can skip forward to \cref{sec:setting}. Readers only interested
in specific applications can read this section first and then skip to
\cref{sec:primitives} or \cref{sec:examples}.

In this paper, we consider two-dimensional arrangements of trapped atoms that
can either be in their electronic ground state or excited into a Rydberg state
(\emph{Rydberg structures}). We focus on systems without quantum fluctuations,
where the ground states are determined by local detunings and Rydberg blockade
interactions (\cref{sec:setting}). The detunings lower the energy for atoms
in the Rydberg state by an atom-specific amount, and the Rydberg blockade
interaction forbids atoms closer than a specific distance to be excited
simultaneously. The interplay of these two contributions singles out ground
states that are characterized by excitations patterns where no additional atom
can be excited without violating the Rydberg blockade, and where the sum of
the detunings of the excited atoms is maximal (so called \emph{maximum-weight
independent sets}). There can be different configurations that minimize the
energy, hence the ground state manifold is typically degenerate. In this
paper, we ask which ground state manifolds such structures can realize and,
conversely, how to tailor structures that realize a prescribed ground state
manifold (\cref{sec:definition}).

A simple example is given in \cref{fig:nutshell}a where the position of the
atoms is shown in (i); the two atoms are constrained by the Rydberg blockade
(gray circles) and cannot be excited simultaneously (indicated by the black
edge connecting them). The color of the atoms encodes their detuning; here
both atoms lower the energy of the system by $\Delta$ when excited into the
Rydberg state (blue nodes). In (ii) we show the two excitation patterns that
minimize the energy (orange nodes denote excited atoms). Note that the atoms
cannot be excited simultaneously due to the Rydberg blockade. If one lists
the ground state configurations in a table, where each column corresponds to
an atom and each row to a ground state configuration, we find the ``truth
table'' of a Boolean \texttt{NOT}-gate $Q=\overline{A}$. Here we interpret
one of the atoms as ``input'' (A) and the other as ``output'' (Q).

This concept generalizes to more complicated Boolean gates
(\cref{fig:nutshell}b): Consider the five atoms in a ring-like blockade
(i). Three of the atoms (blue) lower the energy by $\Delta$, two (green)
by $2\Delta$ when excited. By inspection one finds the four degenerate
ground state configurations in (ii). This is promising as truth tables of
Boolean gates that operate on two bits have four rows. However, they only
have three columns (two for the inputs of the gate and one for its output).
We therefore select three of the five atoms by assigning labels to them:
A and B play the role of the inputs and Q is the output. We call atomic
structures with designated input/output atoms \emph{Rydberg complexes}%
~\footnote{%
To prevent misconceptions, we stress that the term ``complex'' in ``Rydberg
complex'' refers to a \emph{spatial arrangement} of Rydberg atoms (with
additional data) and is \emph{not} related to the mathematical concept of an
\emph{independence complex}, i.e., the family of independent sets of a graph.
}
(\cref{subsec:complexes}). If we list the four ground state configurations
of these three atoms, we find the truth table of a \texttt{NOR}-gate
$Q=A\downarrow B=\overline{A\vee B}$ in (iii).  Note that the remaining two
atoms (we call them \emph{ancillas})---while not contributing independent
degrees of freedom--- are still necessary to realize this specific ground state
manifold. At this point things get interesting because it is a well-known fact
of Boolean algebra that the \texttt{NOR}-gate is \emph{functionally complete}
(just like the \texttt{NAND}-gate): Every Boolean function can be decomposed
into a circuit build from \texttt{NOR}-gates only.

To leverage this decomposition, we need a method to combine ``gate
complexes'' to form larger ``circuit complexes''; we call this procedure
\emph{amalgamation} (\cref{subsec:amalgamation}). A simple example is
shown in \cref{fig:nutshell}c where we attach the \texttt{NOT}-gate
from \cref{fig:nutshell}a to the output of the \texttt{NOR}-gate in
\cref{fig:nutshell}b (note that the detunings of the atoms that are
joined add up). Using the detunings and blockades in (i) yields the four
degenerate ground state configurations in (ii). When we label the inputs
of the \texttt{NOR}-gate again by A and B, and now focus on the output Q
of the attached \texttt{NOT}-gate, we find indeed the truth table of an
\texttt{OR}-gate $Q=\overline{A\downarrow B}=A\vee B$ in (iii). Thus we can
parallel the logical composition of gates by a geometrical combination of
atomic structures such that the relation between ground state configurations
and truth tables remains intact. In combination with the insight that every
Boolean circuit can be drawn in the plane without crossing lines (after
suitable augmentations), this allows us to show that the truth table of any
Boolean function can be realized as the ground state manifold of a suitably
designed atomic structure. This \emph{functional completeness} is our first
main result and motivates the title of the paper (\cref{sec:completeness}).

For instance, the \emph{existence} of a structure that realizes the truth table
of an \texttt{OR}-gate is a corollary of functional completeness.  However,
the \emph{specific} construction as the combination of a \texttt{NOR}-gate
and a \texttt{NOT}-gate in \cref{fig:nutshell}c raises the questions
whether this particular realization with six atoms is \emph{unique} and
whether it is \emph{minimal} (in the sense that the same truth table could
not be realized with fewer atoms). The answer to the first question is
negative: There are geometrically different structures that realize the
same truth table in their ground state manifold. The answer to the second
question is positive, though: We show that it is impossible to implement
this truth table with less than six atoms. Note that the functional
completeness implies the existences of structures for \emph{all} common
gates of Boolean logic (such as \texttt{AND}, \texttt{XOR}, etc.). We
take this as motivation to construct provably minimal structures for all
these primitives (\cref{sec:primitives,sec:crossing}). Together with the
procedure of amalgamation, these equip our versatile toolbox to engineer
more complicated structures.

Our second important contribution is an application of the
functional completeness as a tool to engineer synthetic quantum matter
(\cref{sec:examples}). Many interesting quantum phases in two dimension
are characterized by hidden patterns of long-range entanglement, known
as topological order. These patterns can give rise to anyonic excitations
which make such systems potential substrates for quantum memories and even
quantum computers. A large class of entanglement patterns can be understood
as condensates of extended objects (like strings). A crucial first step
for the realization of these phases is therefore the preparation of Hilbert
spaces spanned by states of such extended objects. However, in experiments,
we typically start from Hilbert spaces with a local tensor product structure
(for example, an array of two-level atoms). Our only hope is to make the
extended objects emerge due to interactions in the low-energy sector of a
suitably designed physical system. This often boils down to enforce local
\emph{gauge symmetries} which single out states that can be interpreted in
terms of extended objects. Such local constraints can be reformulated as
Boolean functions that must be satisfied by the states of the local degrees
of freedom of the underlying system. For any constraint of this form,
our functional completeness result ensures the existence of a structure of
atoms, interacting via the Rydberg blockade mechanism, that realizes this
constraint in its ground state space. It is then just a matter of copying
and joining these structures in a translational invariant way to tessellate
the plane. The ground state manifolds of such tessellations can therefore
implement a large class of non-trivial Hilbert spaces on which condensation
(driven by quantum fluctuations) might lead to topologically ordered many-body
quantum phases. Using our toolbox developed in the first part of the paper,
we demonstrate this construction explicitly for the Abelian toric code phase
(\cref{subsec:toric}) and the non-Abelian, computationally universal Fibonacci
anyon model (\cref{subsec:fibonacci}).

The truth tables realized by the ground states of all atomic structures
presented in this paper depend on the positions of the atoms. (Because these
positions define which pairs are in blockade and which atoms can be excited
simultaneously.) However, the \emph{exact} placement is often ambiguous.
For example, consider the structure in \cref{fig:nutshell}a (i) which realizes
the \texttt{NOT}-gate. It is clear that the blockade constraint (black edge)
does not change if the atoms are slightly shifted, as long as the blockade
radii (gray circles) encompass both atoms. We refer to the set of atom
positions as the \emph{geometry} of a structure and argue that ``robust''
geometries should avoid distances between atoms that are close to the critical
blockade distance. For the complexes in \cref{fig:nutshell}, this translates
into the geometric objective to maximize the distances between nodes and gray
circles. We formalize this notion by assigning a number to geometries that
quantifies their ``robustness'' (\cref{subsec:optimization_preliminaries})
and numerically construct optimized geometries that maximize this number
(\cref{subsec:optimization_results}).

We conclude the paper with an outline of open questions, directions for further
research (\cref{sec:outlook}), and a brief summary (\cref{sec:summary}).

\section{Physical setting}
\label{sec:setting}

We consider planar arrangements of trapped atoms with
repulsive van der Waals interactions when excited into the Rydberg
state~\cite{Lukin_2001,Saffman_2010}. Every atom is assigned an index $i\in
V=\{1\dots N\}$, placed at position $\vec r_i\in\mathbb{R}^2$, and described
by a two-level system $\ket{n}_i$ where $n=0$ corresponds to the electronic
ground state and $n=1$ the excited Rydberg state.

The quantum dynamics of such systems is achieved by coupling the
electronic ground state to the Rydberg state by external laser fields
with Rabi frequency $\Omega_i$ and detuning $\Delta_{i}$ for each
atom~\cite{Schauss2015,Labuhn2016,Bernien_2017}. Here we are mainly interested
in the regime $\Omega_i\rightarrow 0$ where the Hamiltonian reduces to
\begin{align}
    H[\C]=-\sum_i\,\Delta_i n_i
    +\sum_{i<j}\,U(|\vec r_i-\vec r_j|)\,n_in_j\,.
    \label{eq:H}
\end{align}
Note that we assume the detunings $\Delta_i$ to be site
dependent~\cite{Labuhn_2014,Omran2019}.
This Hamiltonian acts on the full Hilbert space $\H=(\mathbb{C}^2)^{\otimes
N}$ with the representation $n_i=\ket{1}\bra{1}_i$. The configuration of
the system is completely specified by $\C\equiv(\vec r_i,\Delta_i)_{i\in
V}$ to which we refer as \emph{(Rydberg) structure}; the position data
$G_\C\equiv(\vec r_i)_{i\in V}$ alone is the \emph{geometry} of the structure
$\C$ (\cref{fig:rationale}).
For atoms in the Rydberg state, the interaction potential in \cref{eq:H}
is $U_\sub{vdW}(r)=C_6\,r^{-6}$ with $C_6>0$ the coupling strength of the
van der Waals interaction; we refer to $H[\C]$ with $U=U_\sub{vdW}$ as the
\emph{van der Waals (vdW) model}. However, in many situations a simplified
model $U=U_\infty$ with $U_\infty(r\geq \Rb)=0$ and $U_\infty(r<\Rb)=\infty$
with \emph{blockade radius} $\Rb$ is a reasonable approximation for the
low-energy physics of \cref{eq:H}; we refer to $H[\C]$ with $U=U_\infty$
as the \emph{PXP model}~\cite{Lesanovsky2011,Verresen_2021}. In this paper,
we use the PXP model unless stated otherwise. We discuss valid choices for
the blockade radius $\Rb$ in \cref{subsec:optimization_preliminaries} where
we optimize the geometry of structures to limit the effects of residual van
der Waals interactions.

In the PXP model, the effect of the van der Waals interactions is approximated
by a kinematic constraint that is completely encoded by a \emph{blockade
graph} $B=(V,E)$, where an edge $e=(i,j)\in E$ between atoms $i,j\in V$
indicates that they are in blockade, i.e., their distance is smaller than
the blockade radius $\Rb$. An abstract graph that can be realized in this
way is called a \emph{unit disk graph} (not every graph has this property);
conversely, a geometry $G_\C$ that realizes a prescribed graph as its blockade
graph is a \emph{unit disk embedding} of this graph (the ``unit'' here is
the blockade radius $\Rb$). Throughout the paper, the blockade graph of a
structure will be drawn by black edges connecting atoms that are in blockade.

\begin{figure}[tb]
    \centering
    \includegraphics[width=0.9\linewidth]{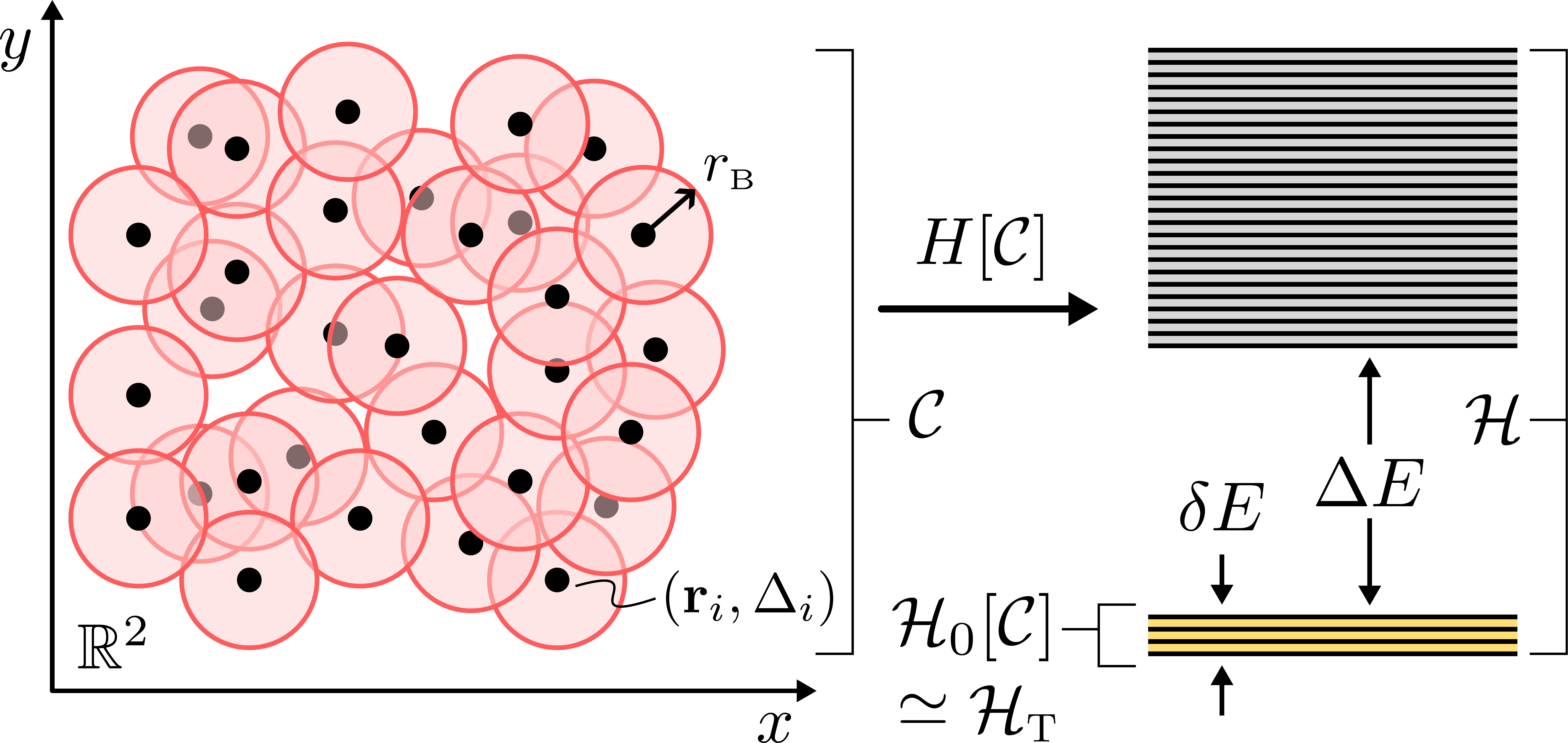}
    \caption{%
	\emph{Setting \& Objective.} 
    A two-dimensional structure $\C=(\vec
	r_i,\Delta_i)_{i\in V}$ of atoms $i\in V$ with position $\vec r_i$
	and detuning $\Delta_i$ is governed by the Hamiltonian $H[\C]$
	that describes the Rydberg blockade interaction with blockade
	radius $\Rb$. The Hamiltonian gives rise to a low-energy eigenspace
	$\H_0[\C]<\H$ of width $\delta E$, separated from the excited states
	by a gap $\Delta E$. The objective of this paper is the construction
	of a structure $\C$ from a given target Hilbert space $\gH$ such
	that $\H_0[\C]\simeq\gH$.
    }
    \label{fig:rationale}
\end{figure}

\section{Definition of the Problem}
\label{sec:definition}

The primary goal of this paper is to find structures $\C$ such that there
is a well-separated low-energy eigenspace $\H_0[\C]$ of $H[\C]$ where
$\H_0[\C]$ satisfies certain prescribed properties that we describe in
detail below. We quantify the separation of $\H_0[\C]$ by its spectral
width $\delta E$ and its gap $\Delta E$ to the rest of the spectrum
(\cref{fig:rationale}). Note that the experimental prerequisites
for the construction of arbitrary structures $\C$ are already in
place~\cite{Labuhn_2014,Barredo_2016,Omran2019,Browaeys_2020}.
If one would switch on a weak drive $\delta E<\Omega_i\ll\Delta E$, this
would induce quantum fluctuations between the states of the Hilbert space
$\H_0[\C]$, potentially giving rise to many-body states with interesting
properties. In this paper, we do not study such quantum effects but focus
on the implementation of the subspace $\H_0[\C]$.
We specify the eigenspace to construct in terms of a \emph{target Hilbert
space} $\gH$:
\begin{align}
    \H_0[\C]\stackrel{!}{\simeq}\H_\sub{T}\,.
\end{align}
Informally speaking, our goal is to ``solve'' this equation for structures
$\C$ for given $\H_\sub{T}$. To make this possible, the target Hilbert space
$\H_\sub{T}$ must be specifiable in a form that we define in the remainder
of this section.

\subparagraph{Formal languages.}
Throughout the paper we make use of the notion of (formal)
languages~\cite{Davis1994} on the binary alphabet $\GF=\{0,1\}$. A \emph{word}
$\vec x\equiv (x_1x_2\dots x_n)\equiv x_1x_2\dots x_n \in \GF^*$ is a
finite string of \emph{letters} $x_i\in\GF$ (the set of all such finite
strings is denoted $\GF^*$). A \emph{(formal) language} $L$ is then
simply a collection of words: $L\subseteq \GF^*$. Here we only consider
\emph{uniform} languages with words that have all the same length. For
example, $L_\texttt{CPY}:=\{000,111\}\subset\GF^*$ is a uniform language of
words with length $n=3$, $\vec x=(111)$ is a word in $L_\texttt{CPY}$ and
$x_1=1$ is the first letter of $\vec x$. The words $\vec y=(011)\in\GF^*$
and $\vec z=(0000)\in\GF^*$ are not in this language: $\vec y,\vec z\notin
L_\texttt{CPY}$. The subscript ``\texttt{CPY}'' stands for ``copy'' and
hints at the role this language will play later.

Other examples are the class of languages generated by the truth tables of
Boolean functions. Let $w:\GF^{n-1}\to\GF$ be an arbitrary Boolean function
of $n-1$ variables; then
\begin{align}
    L[w]:=\left\{x_1\dots x_{n-1}y\,|\,y=w(x_1,\dots,x_{n-1})\right\}\subset\GF^*
    \label{eq:Lw}
\end{align}
is the language generated from the rows of the truth table of $w$, where
the first $n-1$ letters of each word correspond to the input $\vec x$ and
the last letter encodes the output $w(\vec x)$. A language of this class
always has $2^{n-1}$ words of uniform length $n$. Note that the ``copy''
language $L_\texttt{CPY}$ is \emph{not} of the form \cref{eq:Lw}.

\begin{figure}[tb]
    \centering
    \includegraphics[width=0.95\linewidth]{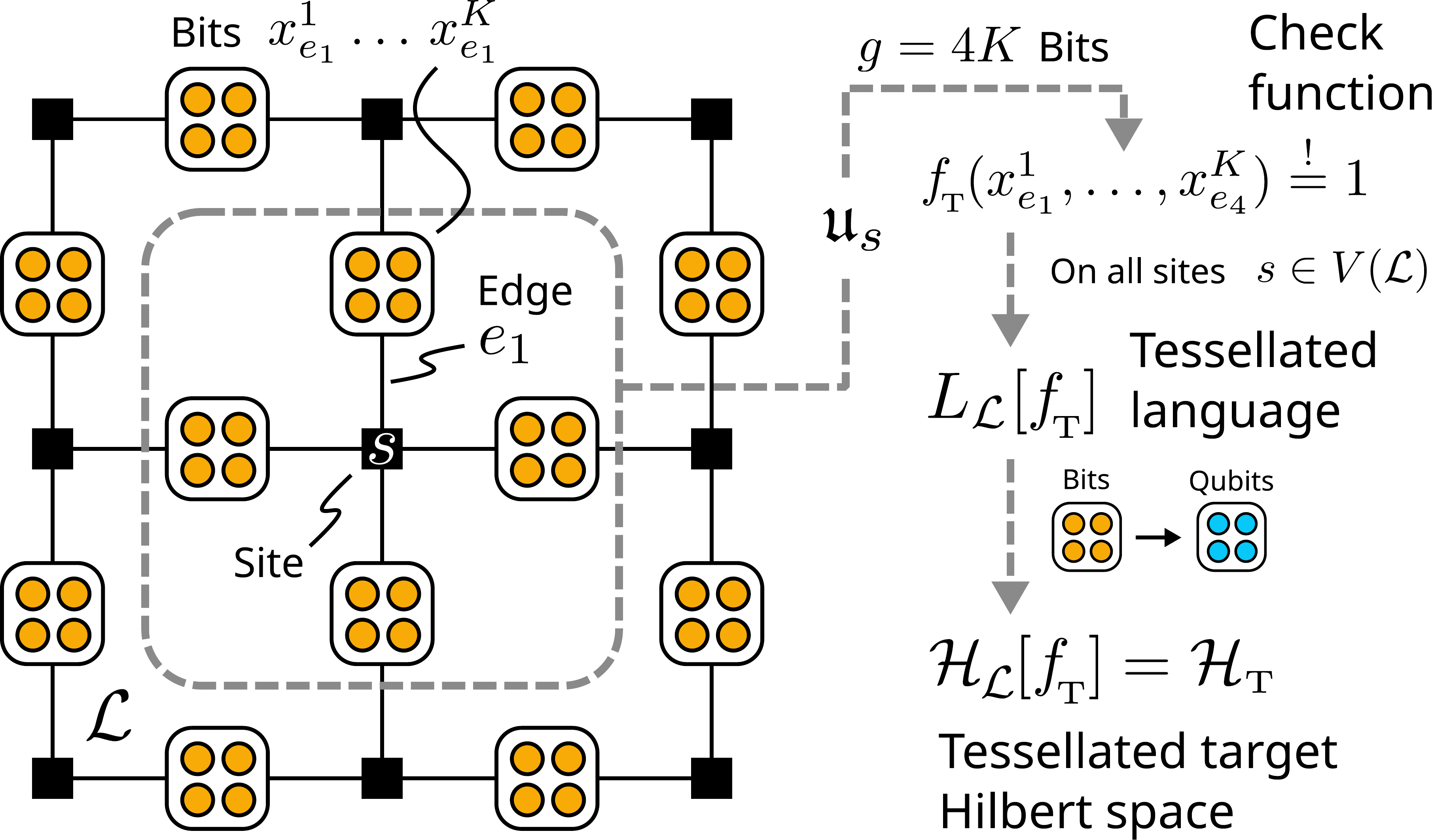}
    \caption{%
    \emph{Tessellated language \& target Hilbert space.} A tessellated
    target Hilbert space $\H_\sub{T}$ is a subspace of the full Hilbert
    space of $K$ qubits placed on each edge of a square lattice $\L$; it
    is spanned by product states $\ket{\vec x}$ of bit patterns $\vec x \in
    L_\L[\fT]$. The tessellated language $L_\L[\fT]$ comprises
    all bit patterns $\vec x\in \GF^*$ that locally satisfy the Boolean check
    function $\fT\,:\,\GF^g\to\GF$. The $g=4K$ arguments of the check
    function on each site $s$ are singled out by the bit-projector $\u_s$.
    }
    \label{fig:bits}
\end{figure}

Another special class is given by \emph{tessellated languages} on lattices. In
the following, we introduce the concept exemplarily for a finite square
lattice $\L$ with periodic boundaries; the generalization to other lattices
and boundary conditions is straightforward.
Start by associating $K$ classical bits to every edge $e\in E(\L)$ of the
lattice (\cref{fig:bits}). A bit configuration of the system $\vec x\in
\mathbb{X}_\L=\GF^{K |E(\L)|}\subset\GF^*$ assigns every bit a Boolean
value $x_e^i$ ($i=1\dots K$). We focus on the family of uniform languages
$L\subseteq\mathbb X_\L$ that can be characterized by a Boolean function that
is local in the following sense: For a site $s\in V(\L)$ of the square lattice,
let the \emph{bit-projector} $\u_s(\vec x)=(x_{e_1}^1,\dots,x_{e_4}^K)$
single out the (ordered) set of $g=4K$ bits on the four edges $e_i$ emanating
from $s$. Let $f\,:\,\GF^g\to\GF$ be an arbitrary Boolean function of $g$
arguments, henceforth referred to as \emph{check function}. The tessellated
language of bit patterns on $\L$ generated by $f$ is then defined as
\begin{align}
    L_\L[f]:=\left\{\vec x\in\mathbb{X}_\L\,|\,\forall s\in V(\L)\,:\,f(\u_s(\vec x))=1\right\}\,.
\end{align}
In words: $L_\L[f]$ is the set of bit patterns on the lattice $\L$ that
locally satisfy the constraints imposed by $f$.

\subparagraph{Target Hilbert spaces.} 
To any uniform language $L\subseteq\GF^n$ we can naturally associate the
linear subspace of states on $n$ qubits (or spin-$\nicefrac{1}{2}$)
\begin{align}
    \H(L):=\spn{\ket{\vec x}\,|\,\vec x\in L}\subseteq (\mathbb{C}^2)^{\otimes n}\,.
    \label{eq:HL}
\end{align}
For example, $\H(L_\texttt{CPY})=\spn{\ket{000},\ket{111}}$ is the
two-dimensional subspace on three qubits spanned by product states with
configurations in $L_\texttt{CPY}=\{000,111\}$. By contrast, the Hilbert space
$\H'=\spn{(\ket{000}+\ket{111})/\sqrt{2}}$ is not of the form \eref{eq:HL}.

We require the target Hilbert space $\H_\sub{T}$, that we aim to realize as
ground state manifold $\H_0[\C]$, to be specified by a language $L_\sub{T}$
according to \cref{eq:HL}:
\begin{align}
    \H_\sub{T}=\H(L_\sub{T})\,.
\end{align}
We are particularly interested in the special class of \emph{tessellated}
target Hilbert spaces given in terms of tessellated languages that are
generated by a check function (\cref{fig:bits}):
\begin{align}
    \H_\sub{T}=\H_\L[\fT]:=\H(L_\L[\fT])\,.
\end{align}
Recall that these languages come equipped with a spatial structure (in the
sense that the \emph{bits} are located on the edges of a lattice $\L$). This
spatial structure is inherited by the Hilbert space $\H_\L[\fT]$ viewed as
state space of a system where $K$ \emph{qubits} are placed on every edge
of $\L$.

For example, the Hilbert space $\H_{\mathbb{Z}_2}$ of the even
$\mathbb{Z}_2$ lattice gauge theory is a particular subspace of a
Hilbert space that describes a system of qubits on the edges of a
square lattice (i.e.,~$K=1$ and $g=4$). $\H_{\mathbb{Z}_2}$ is spanned
by the product states of patterns of qubits in the state $\ket{1}$
that form closed loops~\cite{Kogut1979}. $\H_{\mathbb{Z}_2}$ is an
admissible tessellated target Hilbert space because we can realize
$\H_{\mathbb{Z}_2}=\H_\L[f_{\mathbb{Z}_2}]$ with the check function
\begin{align}
    f_{\mathbb{Z}_2}(x_1,x_2,x_3,x_4)=1\oplus x_1\oplus x_2\oplus x_3\oplus
    x_4 \label{eq:f1}
\end{align}
where $\oplus$ denotes modulo-2 addition (Exclusive-OR or \texttt{XOR});
the bit-projector $\mathfrak{u}_s(\vec x)$ simply singles out the four bits
on edges emanating from site $s$:
\begin{align}
    \mathfrak{u}_s\left(
    \includegraphics[width=1.7cm,valign=c]{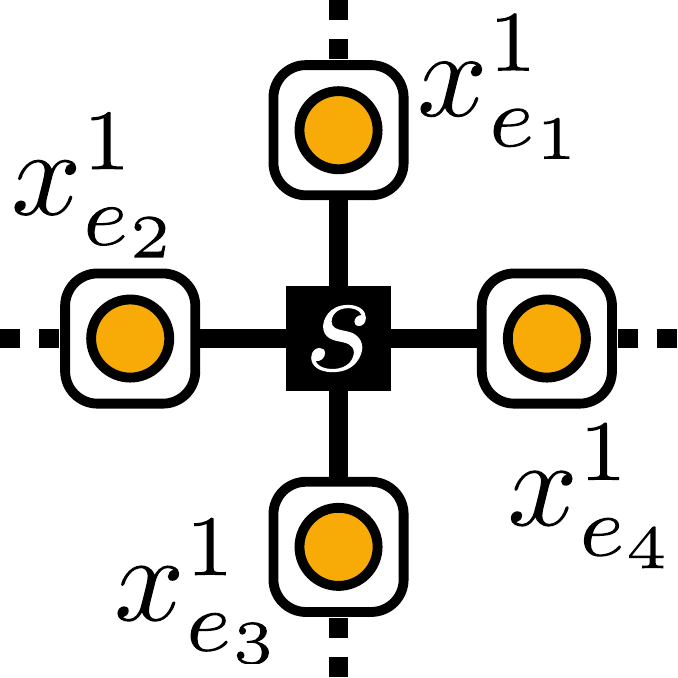}
    \right)
    =(x_{e_1}^1,x_{e_2}^1,x_{e_3}^1,x_{e_4}^1)\,.
\end{align}
Physically, \cref{eq:f1} enforces Gauss's law on a charge-free background
by forbidding strings of qubits in state $\ket{1}$ to end on a site.
Further examples for tessellated target Hilbert spaces are the more general
``string-net'' Hilbert spaces that can describe a large variety of topological
orders and deconfined gauge theories~\cite{Levin2005}.

\section{Rydberg complexes}
\label{sec:complexes}

Before we can tackle our main goal, namely the construction of tessellated
Rydberg structures $\C$ with $\H_0[\C]\simeq\gH=\H_\L[\fT]$ for a given
check function $\fT$, we first need to specify the notion of a finite Rydberg
\emph{complex} as a preliminary step. Specific examples for Rydberg complexes
can be found throughout the remainder of the paper.

\subsection{From structures to complexes}
\label{subsec:complexes}

Consider the language $L_\texttt{CPY}=\{000,111\}$ and let
$\H_\texttt{CPY}=\H(L_\texttt{CPY})=\spn{\ket{000},\ket{111}}$ be our target
Hilbert space. Our goal is to realize $\H_\texttt{CPY}$ as the ground state
manifold $\H_0[\C_\texttt{CPY}]$ of a structure $\C_\texttt{CPY}$ of $n=3$
atoms. This, however, is impossible: Since $\ket{111}\in \H_\texttt{CPY}$,
none of the three atoms can be in blockade with each other. Consequently,
$\H_0[\C_\texttt{CPY}]$ cannot contain only the states $\ket{000}$ and
$\ket{111}$ (\cref{app:cpy}). This problem is not specific to the language
$L_\texttt{CPY}$ but shared by many (though not all) languages. The solution
is to consider \emph{larger} structures of $N\geq n$ atoms and to identify
the letters of words with a \emph{subset} of $n$ distinguished atoms (we
call them \emph{ports}); the remaining $N-n$ atoms play then the role of
\emph{ancillas}. A structure together with a distinguished set of ports will
be referred to as a \emph{(Rydberg) complex}.

Let us formalize this notion. Consider a structure $\C$ of $N$ atoms and
a language $L\subseteq\GF^n$ of words of uniform length $n\leq N$. Let
$\mathfrak{L}=\{\text{A},\text{B},\dots\}$ denote a set of $n$ labels where
each label is associated with a fixed letter position of words in $L$. (If one
prints all words of $L$ as rows of a table, the labels correspond to the column
headers.) Let $\ell:\mathfrak{L}\rightarrow V$ be an injective label function
that assigns a label to a subset of $n$ atoms (the \emph{ports}); the $N-n$
atoms without labels are the \emph{ancillas}. We refer to the structure $\C$
together with the labeling $\ell$ as a \emph{(Rydberg) $L$-complex} $\C_L$
if the states that span $\H_0[\C]$ can be identified by the configurations
of the ports alone:
\begin{align}
    \H_0[\C_L]\equiv\H_0[\C]
    =\spn{\ket{\vec x,a(\vec x)}\in\H\,|\,\vec x\in L}\,.
    \label{eq:lcomplex}
\end{align}
In $\ket{\vec x,a(\vec x)}$, the state of ports is given by the first $n$
bits $\vec x$ (in some fixed order) and the state of ancillas by a $N-n$
bit-valued function $a:L\to\GF^{N-n}$. The ground state space $\H_0[\C_L]$
will be referred to as an \emph{$L$-manifold}. An important aspect of
this definition is that the ancillas do \emph{not} introduce additional
low-energy degrees of freedom; they are only needed to unleash the full
potential of the blockade interactions. In this sense, we say that a complex
$\C_\sub{T}\equiv \C_{L_\sub{T}}$ \emph{realizes} a target Hilbert space
$\H_\sub{T}=\H(L_\sub{T})$ and write
\begin{align}
    (\mathbb{C}^2)^{\otimes N}\supseteq
    \H_0[\C_\sub{T}]\simeq \H_\sub{T}=\H(L_\sub{T})
    \subseteq(\mathbb{C}^2)^{\otimes n}
    \label{eq:iso}
\end{align}
with the isomorphism $\simeq$ given by $\ket{\vec x,a(\vec
x)}\leftrightarrow\ket{\vec x}$. If we say that a complex realizes a
\emph{language} $L$, we mean that it realizes the target Hilbert space
$\H_\sub{T}=\H(L)$ defined by this language.

As an example, consider again the ``copy'' language
$L_\text{\texttt{CPY}}=\{000,111\}$ with $n=3$; the ground state manifold
of a $L_\texttt{CPY}$-complex $\C_\texttt{CPY}\equiv\C_{L_\texttt{CPY}}$
must be two-dimensional (since $|L_\texttt{CPY}|=2$) and characterized by the
property that three distinguished atoms (the ones assigned labels by $\ell$)
are always forced to be in the same state:
\begin{align}
    \H_0[\C_\texttt{CPY}]=\spn{\ket{000,a(000)},\ket{111,a(111)}}\,.
\end{align}
Such a complex will be one of our primitives to implement check functions
for tessellated target Hilbert spaces. We will discuss a specific realization
$\C_\texttt{CPY}$ that requires a single ancilla in \cref{sec:completeness};
that is, with $N=n=3$ atoms the target Hilbert space $\H_\texttt{CPY}$
cannot be realized, whereas with $N=4$ it can.

As another example, consider the logical \texttt{XOR}-gate
$w_\texttt{XOR}(x_1,x_2)=x_1\oplus x_2$ which may be needed as
a primitive for a check function like \cref{eq:f1}. We can ask
for a complex $\C_\texttt{XOR}$ that realizes the target Hilbert
space $\H_\texttt{XOR}=\H(L_\texttt{XOR})$ given by the language
$L_\texttt{XOR}\equiv L[w_\texttt{XOR}]=\{000,011,101,110\}$ that is generated
by this Boolean gate. The ground state manifold of such a complex must be
spanned by four states,
\begin{align}
    \H_0[\C_\texttt{XOR}]=\spn{
    \begin{matrix}
    \ket{000,a(000)},
    &\ket{011,a(011)},\\
    \ket{101,a(101)},
    &\ket{110,a(110)}\phantom{,}
    \end{matrix}
    }
\end{align}
where the configurations of potential ancillas are determined by the
configurations of the three ports. We will introduce a specific realization
$\C_\texttt{XOR}$ in \cref{sec:primitives}; it requires $N=7$ atoms of which
four are ancillas, and we show that this is indeed the smallest complex that
can realize the language of a \texttt{XOR}-gate.

Since $L_\texttt{XOR}=L[w_\texttt{XOR}]$ is generated from a Boolean gate,
we refer to complexes that realize a language of this form as \emph{gates},
too. Furthermore, we denote the atoms that map to the input bits of the gate
as \emph{input ports}, and the atom that corresponds to the output bit as the
\emph{output port}. We also extend this nomenclature to Boolean functions $w$
on more than two inputs. Let us stress that these terms are only inspired
from the usual role played by such functions as parts of Boolean circuits. In
the present context, there is \emph{no} time evolution or dynamics involved
(there is no information ``flowing into'' the input ports, although it might
be sometimes helpful to use this picture).

The construction of an $L$-complex for a given language $L$ with word length
$n$ can be split into two steps: First, one has to find a \emph{structure}
$\C$ on at least $n$ atoms with an $|L|$-fold degenerate ground state
manifold. Then, one has to identify a labeling $\ell$ of $n$ atoms such that
their states in the ground state manifold map one-to-one to words in $L$. The
structure $\C$ together with the labeling then yields an \emph{$L$-complex}.
Note that the same structure can be interpreted as different complexes for
different languages by choosing different label functions. Furthermore,
not every structure with $|L|$-fold degenerate ground state manifold allows
for a valid labeling that realizes $L$. Hence the construction is a quite
non-trivial task in general. This makes a reductionist approach seem most
promising, where one starts with a finite set of small ``primitive'' complexes
and constructs larger complexes by ``gluing'' them together.

\subsection{Amalgamation}
\label{subsec:amalgamation}

The process of combining two complexes by joining (some of) their ports is
referred to \emph{amalgamation}. To define the process formally, we first
need a new concept to combine two languages.

Consider two uniform languages $L_1$ and $L_2$ of words of length $n_1$
and $n_2$, respectively.
Let $\gamma\subseteq \{(p_1,p_2)\,|\,p_i\in\{1,\dots,n_i\}\}$ be a set of disjoint~\footnote{%
Here, \emph{disjoint} means that $(x,y)\neq(x',y')$ implies $x\neq x'$
\emph{and} $y\neq y'$; thus $\gamma$ can be interpreted as a \emph{partial
bijection} between character positions of the two languages $L_1$ and $L_2$.
}
pairs of letter positions and set $\gamma_i:=\{p_i\,|\,p\in\gamma\}$.
For a word $\vec x\in L_i$, let $\vec x^{\gamma_i}$ denote the
word with all letters at positions in $\gamma_i$ deleted. Then, the
\emph{$\gamma$-intersection of $L_1$ and $L_2$} is defined as
\begin{align}
    L_1\underline{\stackrel{\gamma}{\cap}} L_2:=
    \left\{
        \vec x\,\vec y^{\gamma_2}\,|\,
        \vec x\in L_1,\vec y\in L_2,
        \forall_{(a,b)\in\gamma}\,x_a=y_b
    \right\}
    \nonumber
\end{align}
which is a language of words of length
$n_1+n_2-|\gamma|$. $L_1\underline{\stackrel{\gamma}{\cap}}L_2$ is the set
of concatenations of words from $L_1$ and $L_2$ where the letters at the
positions indicated by pairs in $\gamma$ coincide, and where the second copy
of these letters has been deleted. Analogously, we define the \emph{reduced
$\gamma$-intersection} as
\begin{align}
    L_1\stackrel{\gamma}{\cap} L_2:=
    \left\{
        \vec x^{\gamma_1}\,\vec y^{\gamma_2}\,|\,
        \vec x\in L_1,\vec y\in L_2,
        \forall_{(a,b)\in\gamma}\,x_a=y_b
    \right\},
    \nonumber
\end{align}
only that now \emph{both} copies of identified letters are deleted; hence
this is a language of words with length $n_1+n_2-2|\gamma|$.

As an example, consider again the \texttt{XOR}-language
$L_\texttt{XOR}=\{000,011,101,110\}$ and the \texttt{CPY}-language
$L_\texttt{CPY}=\{000,111\}$. We would like to copy the output of the
\texttt{XOR}-gate. To do this, we intersect the output bit (letter 3)
of the \texttt{XOR}-language with one of the bits (say letter 1) of the
\texttt{CPY}-language: $\gamma=\{(3,1)\}$. The $\gamma$-intersection is the
new language
\begin{align}
    L_\texttt{XOR}\underline{\stackrel{\gamma}{\cap}} L_\texttt{CPY}
    &=
    \left\{
        00\underline{0}00,01\underline{1}11,10\underline{1}11,11\underline{0}00
    \right\}
    \label{eq:L1}
\end{align}
with words of length $3+3-1=5$. The underscores indicate the letters that
derive from words of both languages. If one drops these letters as well
(by using the reduced $\gamma$-intersection), the language describes a
\texttt{XOR}-gate with \emph{fan-out} of two:
\begin{align}
    L_\texttt{XOR}\stackrel{\gamma}{\cap} L_\texttt{CPY}
    &=
    \left\{
        0000,0111,1011,1100
    \right\}\,.
    \label{eq:L1b}
\end{align}

The above definitions on the level of languages are useful because they
are paralleled by a combination of complexes called \emph{amalgamation}:
Consider two complexes $\C_{L_1}$ and $\C_{L_2}$ that realize the languages
$L_1$ and $L_2$ with $N_1$ and $N_2$ atoms, respectively. Fix a set of pairs
of ports $\gamma$ such that $L'=L_1\underline{\stackrel{\gamma}{\cap}}
L_2\neq\emptyset$, and then combine the two complexes by identifying the
atoms in $\gamma$:
\begin{align}
    \C_{L'}=\C_{L_1}\stackrel{\gamma}{\otimes}\C_{L_2}
    &:=
    \includegraphics[width=4cm,valign=c]{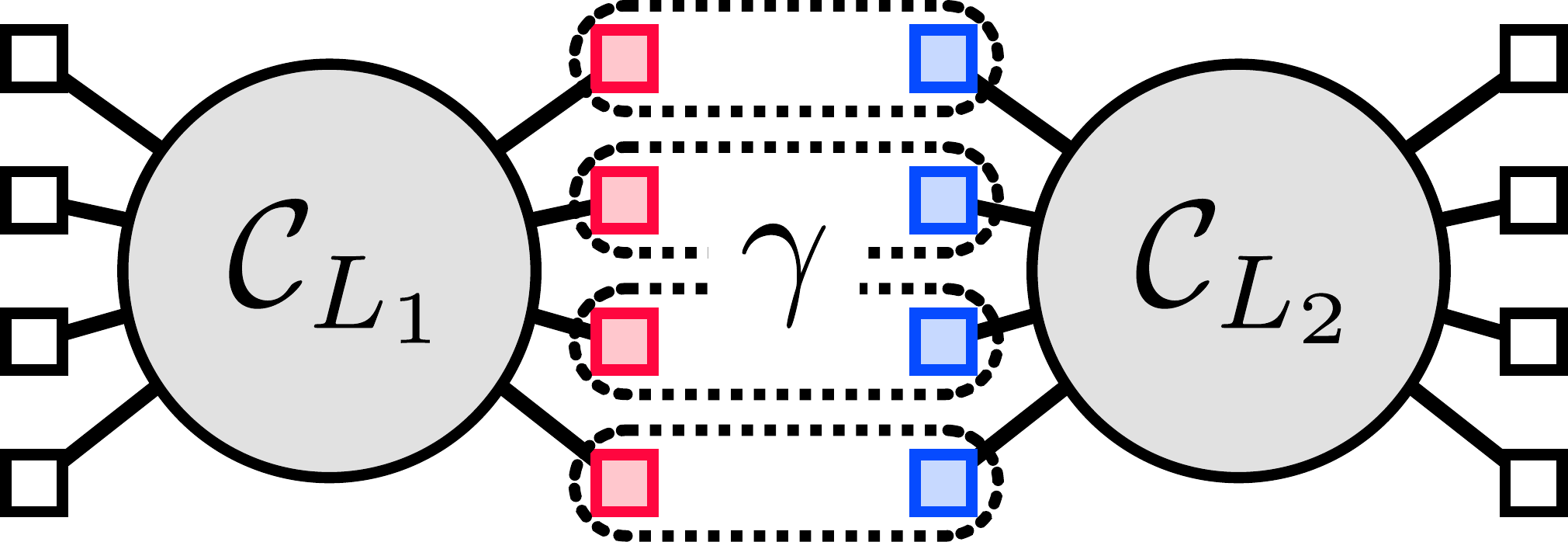}
    \nonumber\\
    &=
    \includegraphics[width=3.1cm,valign=c]{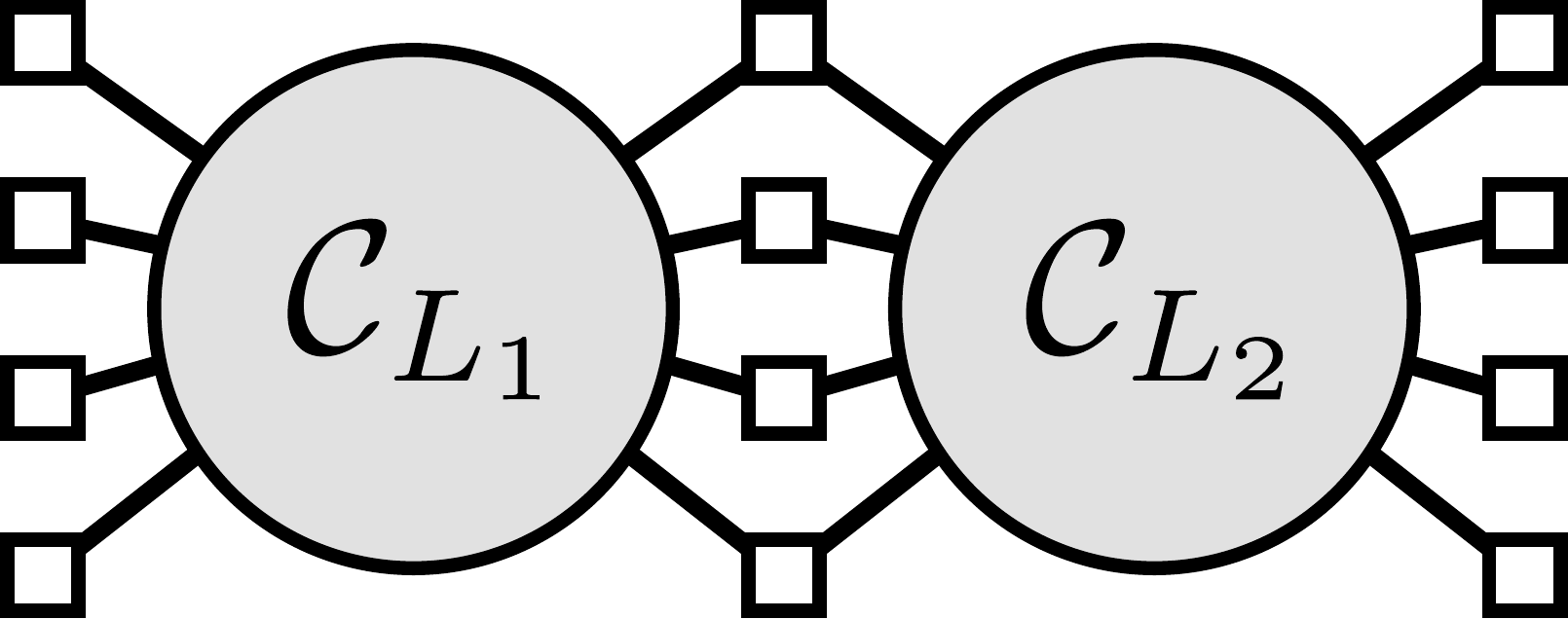}\,.
\end{align}
The new complex $\C_{L'}$ has $N_1+N_2-|\gamma|$ atoms. For this construction,
we assume that the ports that belong to pairs in $\gamma$ are located on the
boundary of their complex (we will show in \cref{sec:completeness} why this
is possible).
The Hamiltonian of the new complex is
\begin{align}
    H[\C_{L'}]=\left(H[\C_{L_1}]+H[\C_{L_2}]+\delta H\right)/\gamma
    \label{eq:amalgamation}
\end{align}
where the formal quotient $\bullet/\gamma$ indicates that pairs of atoms in
$\gamma$ are identified; $\delta H$ denotes additional interactions between the
two subcomplexes $\C_{L_i}$ that vanish in the PXP model (in the vdW model they
are finite but strongly suppressed due to the quick decay of $U_\sub{vdW}$).

In a nutshell: $H[\C_{L'}]$ is the sum of the Hamiltonians of the original
two complexes were the detunings of the ports that are identified by
$\gamma$ add up. For example, let $n^{(1)}$ and $n^{(2)}$ describe ports of
$\C_{L_1}$ and $\C_{L_2}$, respectively, and let $\gamma$ identify these
two ports. Then $H[\C_{L_1}]$ contains a term $-\Delta^{(1)} n^{(1)}$
and $H[\C_{L_2}]$ contains a term $-\Delta^{(2)} n^{(2)}$. The Hamiltonian
\eref{eq:amalgamation} of the amalgamation contains the term $(-\Delta^{(1)}
n^{(1)}-\Delta^{(2)} n^{(2)})/\gamma=-(\Delta^{(1)}+\Delta^{(2)})n'$ where
$n'=n^{(1)}/\gamma=n^{(2)}/\gamma$ describes the atom that corresponds to
the identification of the two ports.

With $\delta H=0$, it is straightforward to verify that the amalgamation
$\C_{L'}$ realizes the language $L'=L_1\underline{\stackrel{\gamma}{\cap}}
L_2$. This is so because the ground state energy of $H[\C_{L'}]$ is
lower-bounded by the sum of the ground state energies of the summands
$H[\C_{L_i}]$; but this lower bound is realized by configurations in
$L'\neq\emptyset$.
The ports identified by $\gamma$ can be interpreted as ancillas
of the new complex if $|L_1\underline{\stackrel{\gamma}{\cap}}L_2| =
|L_1\stackrel{\gamma}{\cap}L_2|$, i.e., if the states of these atoms provide
redundant information about the ground state manifold; in this case, one
would define $L'=L_1\stackrel{\gamma}{\cap}L_2$ instead.

An important special case of the above construction is the amalgamation
of \emph{gates} where the input ports of one gate are identified with the
output ports of others. For example, let $w(x_1,x_2)$ and $w'(x_1',x_2')$
be two Boolean gates that are concatenated into the circuit on three inputs
$\tilde w(x_1',x_1,x_2):=w'(x_1',w(x_1,x_2))$. It is easy to see that
$L[\tilde w]=L[w]\stackrel{\gamma}{\cap} L[w']$ with $\gamma=\{(3,2)\}$
where 3 labels the third letter of words in $L[w]$, which encodes the
output $y=w(x_1,x_2)$, and 2 labels the second letter of words in $L[w']$,
which encodes the input $x_2'$. Note that for Boolean circuits without
redundancies it is always $|L[w]\underline{\stackrel{\gamma}{\cap}}L[w']|
= |L[w]\stackrel{\gamma}{\cap}L[w']|$ because all words are identified by
the input bits.
This example demonstrates that the amalgamation of gates is a crucial
ingredient for the decomposition of complex Boolean circuits into a small
set of simple gates.

\section{Functional completeness}
\label{sec:completeness}

\begin{figure*}[t]
    \centering
    \includegraphics[width=1.0\linewidth]{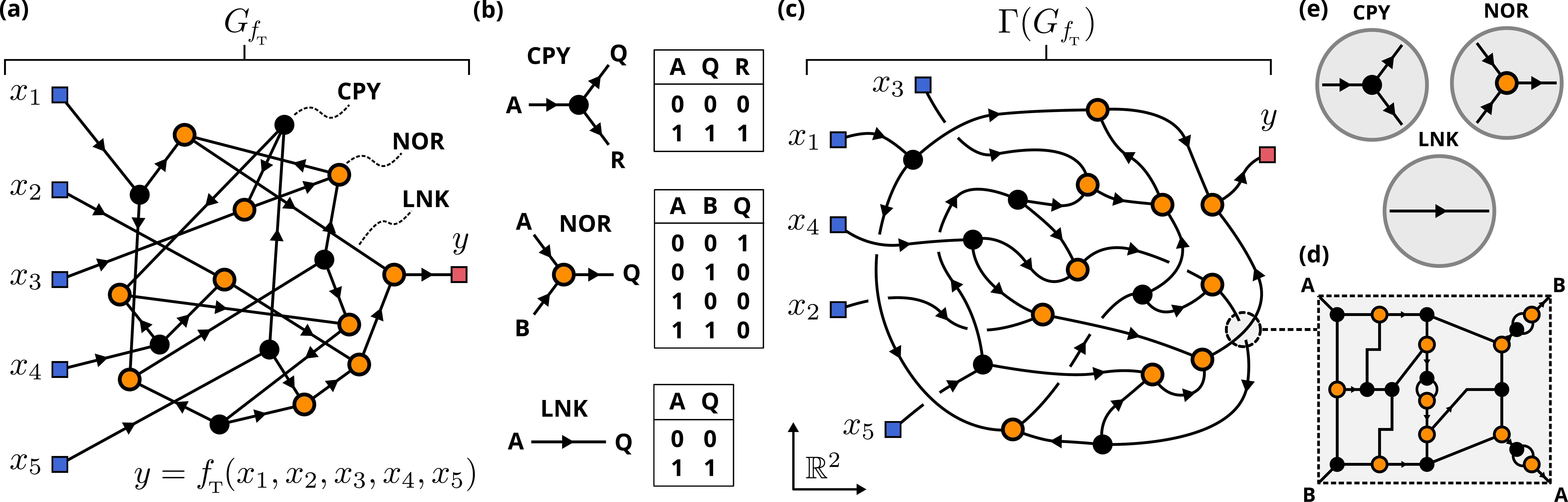}
    \caption{%
	\emph{Decomposition of Boolean functions.} 
    (a) Any Boolean function $\fT$ can be represented by a graph $G_\fT$ (a
    ``Boolean circuit'') with dedicated input vertices (blue squares), one
    output vertex (red square), and trivalent vertices (circles) of two types
    (b):
	\texttt{NOR}-gates with two incoming and one outgoing edge (orange
	circles) and \texttt{CPY}-vertices with one incoming and two outgoing
	edges (black circles); the edges themselves can be interpreted as
	trivial single-bit gates, here referred to as \texttt{LNK}-gates
	(black edges).
    If the inputs (A,B) and outputs (Q,R) of all three primitives are assigned
    Boolean values that satisfy the truth tables in (b), the value at the
    output vertex is $y=\fT(x_1,\dots)$ by construction.
    (c) The embedding $\Gamma(G_\fT)$ (``drawing'') of the abstract graph $G_\fT$
    in the plane $\mathbb{R}^2$ typically involves crossings (whenever $G_\fT$
    is \emph{non-planar}); furthermore, input and output vertices may lie in
    the interior of the graph. Since a crossing of wires can be implemented
    with the available vertices (d), the graph can always be enhanced such
    that it becomes planar \emph{and} input/output vertices lie on the
    perimeter of the embedding. (e) Locally, the embedding $\Gamma(G_\fT)$
    decomposes into three primitives, namely the structures referred to as
    \texttt{NOR}, \texttt{CPY}, and \texttt{LNK} that are functionally defined
    by the truth tables in (b) and geometrically by the sketches in (e).
    }
    \label{fig:decomposition}
\end{figure*}

We have now all concepts and tools in place to formulate the main result of
this paper:
\begin{theorem}[Functional completeness]
    \label{thm:1} 
    For every tessellated target Hilbert space $\gH=\H_\L[\fT]$ on
    some lattice $\L$ that is generated by a check function $\fT$,
    there exists a structure $\C_\sub{T}$ in the PXP model such that
    \begin{align}
        \gH\stackrel{\text{loc}}{\simeq} \H_0[\C_\sub{T}]\,,
        \label{eq:equiv}
    \end{align}
    with finite gap $\Delta E>0$ and perfect degeneracy $\delta E=0$.
\end{theorem}
In \cref{eq:equiv}, $\stackrel{\text{loc}}{\simeq}$ denotes an isomorphism
of Hilbert spaces like \cref{eq:iso} that, in addition, preservers the
locality structure: it maps local unitaries on $\gH$ to local unitaries
on $\H_0[\C_\sub{T}]$ and vice versa. Here the locality structure of
$\H_0[\C_\sub{T}]$ is induced by the locality structure of $\mathcal{H}$
which reflects the physical realization of the system. The locality
structure of $\gH=\H_\L[\fT]$ derives from the lattice $\mathcal{L}$ and
the bit-projector $\mathfrak{u}_s$ that was used to define the tessellated
language $L_\L[\fT]$; it is therefore part of the defining properties of
the Hilbert space $\gH$. This local isomorphism will be explicit for the
examples in \cref{sec:examples}.

\begin{proof}

The proof of \cref{thm:1} is constructive in principle and best split into
several steps: Steps 1 to 4 deal with the construction of a Rydberg complex
$\C_{\fT=1}$ that implements the constraint of the check function on a single
site of the lattice. In the final Step 5, the structure $\C_\sub{T}$ is then
constructed as the amalgamation of copies of $\C_{\fT=1}$ on the full lattice.

\paragraph*{Step 1: Decomposition of $\fT$.} 
The first goal is to convert the check function $\fT\,:\,\GF^g\to\GF$ on $g$
binary inputs into a finite set of Boolean gates as ``building blocks.'' There
are many universal gate sets to choose from~\cite{Wernick1942} but the one
that is most natural to the Rydberg platform is the singleton $\{\NOR\}$
that contains only the $\NOR$-gate~\cite{Sheffer1913}
\begin{align}
    A\downarrow B:=\overline{A\vee B}\,.
    \label{eq:nor}
\end{align}
The idea behind this choice is simple: Placing three atoms $A,C,B$ in a
row such that the pairs $(A,C)$ and $(C,B)$ are in blockade but the pair
$(A,B)$ is not naturally gives rise to a constraint akin to $C=A\downarrow B$
(we discuss the details below).
The functional completeness of $\{\NOR\}$ allows us to write
\begin{align}
    \fT(x_1,\dots,x_g)=(\dots (x_i\downarrow x_j) \dots (x_k\downarrow x_l)\dots)
    \label{eq:f}
\end{align}
where the expression on the right can be any (recursive) combination of
expressions built from the input variables paired by $\NOR$-gates. On an
abstract level, this is a neat result; however, in reality one has to be more
careful because variables can be used multiple times at different locations
in the $\NOR$-expansion of $\fT$.

To identify the true physical building blocks needed to cast \cref{eq:f}
into a structure of atoms, it is advisable to translate the $\NOR$-expansion
into a graph $G_\fT$ that represents the underlying Boolean circuit and
uses the inputs $x_i$ only \emph{once} at dedicated ``input vertices'' and
outputs the result $\fT(x_1,\dots,x_g)$ at a dedicated ``output vertex''
(\cref{fig:decomposition}a). Otherwise, $G_\fT$ is a trivalent graph
with two types of vertices, corresponding to $\texttt{CPY}$-operations
that copy a bit and $\texttt{NOR}$-gates that combine two bits according
to \cref{eq:nor}. If we assign arrows to the edges to highlight the
information flow, the two vertices are distinguished by the number of
in- and outgoing edges (\texttt{CPY}: 1 in and 2 out, \texttt{NOR}: 2 in
and 1 out). Furthermore, we can interpret the edges themselves as trivial
single-bit gates (``\texttt{LNK}-gates''). If we assign Boolean values to the
inputs and outputs of these three primitives according to the truth tables
in \cref{fig:decomposition}b, the value of the output vertex is given by
$y=\fT(x_1,\dots,x_g)$. Without loss of generality, we consider only circuits
without redundancy, i.e., for a given input $\{x_1,\dots,x_g\}$ the state
of the inputs and outputs of all its primitives is uniquely determined. This
implies that there are exactly $2^g$ such assignments that are parametrized
by the $g$ inputs $\{x_1,\dots,x_g\}$ (this can be seen as a \emph{boundary}
condition; in a dynamical circuit, one would call it an \emph{initial}
condition).

\begin{figure*}[tb]
    \centering
    \includegraphics[width=1.0\linewidth]{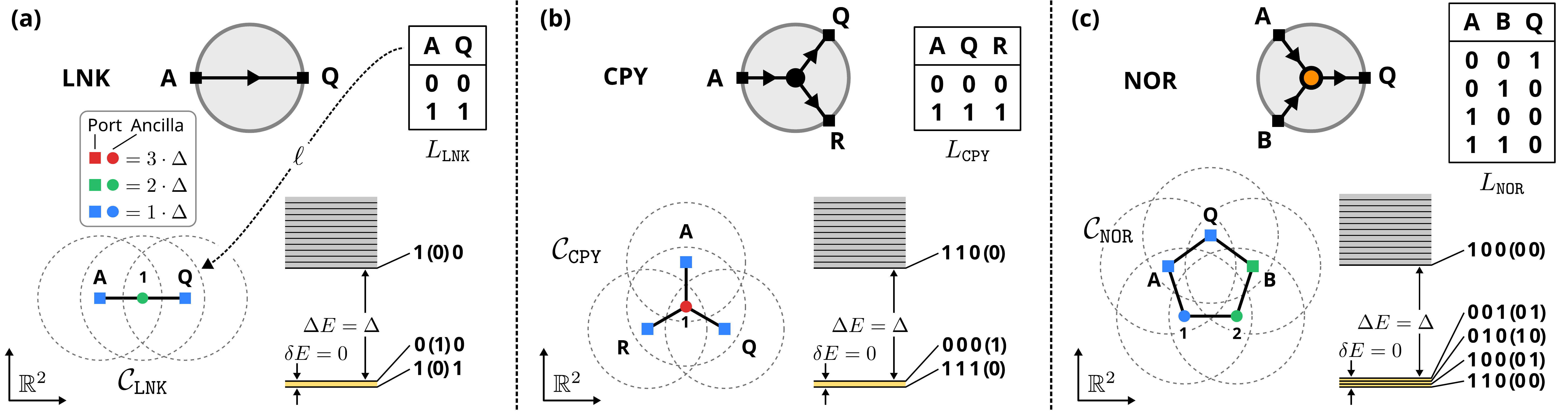}
    \caption{%
	\emph{Complete set of logic primitives.}
    (a)~The (elementary) \texttt{LNK}-complex $\C_\texttt{LNK}$ can
    be realized by a chain of three atoms where adjacent atoms are in
    blockade (black edges). The detuning of the ports $\Delta$ (blue
    squares, labeled by $\ell$) is half that of the ancilla $2\Delta$
    (green circle) in the bulk. The width $\delta E$ and gap $\Delta E$ are
    shown together with a schematic spectrum that highlights the logical
    manifold $\H_0[\C_\texttt{LNK}]$ and one of the states orthogonal to
    $\H_0[\C_\texttt{LNK}]$ that define the gap. The state of ancillas is
    shown in parentheses.
    (b)~The \texttt{CPY}-complex $\C_\texttt{CPY}$ can be realized with
    a central ancilla (red circle) that is in blockade with the three
    surrounding atoms (blue squares). To make the two logical states
    degenerate, the ancilla has a detuning of $3\Delta$ if the other atoms
    are detuned by~$\Delta$.
    (c)~The \texttt{NOR}-complex $\C_\texttt{NOR}$ can be realized with two
    ancillas (blue and green circles) that form a ring-like blockade with
    the three ports (blue and green squares). To make the four logical states
    unique and degenerate, the detunings cannot be chosen uniformly but must
    break the reflection symmetry about the axis through the output port~$Q$.
    }
    \label{fig:pxp_primitives}
\end{figure*}

\paragraph*{Step 2: Embedding of $G_\fT$.} 
The graph $G_\fT$ represents the Boolean circuit of $\fT$ on an abstract
level (only the connectivity of $G_\fT$ is relevant). Our final goal is to
translate this graph into a functionally equivalent structure of atoms \emph{in
the plane}. Thus we have to find an embedding $\Gamma(G_\fT)$ of $G_\fT$
in $\mathbb{R}^2$; this embedding should be planar, i.e., without crossing
edges to avoid unwanted interactions. Here we skip a formal definition of
$\Gamma(G_\fT)$ and appeal to the intuition of the reader: $\Gamma(G_\fT)$
describes a drawing of $G_\fT$ in the plane without crossing edges and with
well-separated vertices (\cref{fig:decomposition}c).
Of course not every graph $G_\fT$ is planar, i.e., can be drawn without
crossing edges in the plane. However, it has been shown long ago that
every Boolean circuit can be made planar by augmenting it with ``crossover
sub-circuits'' whenever two lines cross~\cite{Dewdney_1979}. This crossover can
be constructed with various gate sets, including the \texttt{NOR}-singleton
(\cref{fig:decomposition}d). The embedding of the crossover then uses only
the three available primitives in \cref{fig:decomposition}b so that we can,
without loss of generality, assume $\Gamma(G_\fT)$ to be planar.
Note that the existence of a crossover also implies that we can assume the
input and output vertices to be located on the perimeter of the embedding
(as realized in \cref{fig:decomposition}c). Translated into complexes, this
will prove our claim in \cref{sec:complexes} that we can assume the ports
to sit on the perimeter of a complex.

While $\Gamma(G_\fT)$ may look very convoluted on a larger scale,
locally it decomposes into the three simple primitives depicted in
\cref{fig:decomposition}e, namely \texttt{CPY}, \texttt{NOR}, and
\texttt{LNK}. The next step is then to implement these three primitives
as complexes both geometrically (i.e., following the geometry in
\cref{fig:decomposition}e) and functionally (i.e., following the truth
tables in \cref{fig:decomposition}b). An $\fT$-complex can then be obtained
by amalgamation of these primitives according to the geometric blueprint
provided by $\Gamma(G_\fT)$.

\paragraph*{Step 3a: Implementing the \texttt{LNK}-complex.} 
The \texttt{LNK}-complex is the physical counterpart of the ``wires'' in
the drawing of the circuit $\Gamma(G_\fT)$. Logically, it corresponds to the
trivial gate $w(x)=x$ with language $L_\text{\texttt{LNK}} =\{00,11\}$. On
the level of pure Boolean logic, wires are not entities of their own but
on the physical level, sending a bit from one location to another requires
dedicated machinery.

Before we discuss its construction, it is useful to introduce a more
fundamental complex that can be used to construct two of the three primitives:
the \texttt{NOT}-gate with defining language $L_\neg = \{01,10\}$; it realizes
the single-bit gate $w(x)=\overline{x}$ and formalizes the core concept of
the Rydberg blockade. In the PXP model, it can be realized naturally without
ancillas by the Hamiltonian
\begin{align}
    H_\neg=-\Delta (n_A+n_Q)
\end{align}
with a complex $\C_\neg$ where $|\vec r_A-\vec r_Q|<\Rb$. The subscripts
denote the labels of the ports assigned by $\ell$ (we reserve A, B, \dots~for
input ports and Q, R, \dots~for output ports).
The ground state manifold is $\H_0[\C_\neg]=\spn{\ket{01},\ket{10}}$ with
degeneracy $\delta E_\neg=0$ and gap $\Delta E_\neg=\Delta>0$.

The elementary \texttt{LNK}-complex that translates a bit in space
can then be constructed as the amalgamation of two \texttt{NOT}-gates
(\cref{fig:pxp_primitives}a) with Hamiltonian
\begin{align}
    H_\text{\texttt{LNK}} =-\Delta n_A-2\Delta \tilde n_1-\Delta n_Q\,,
\end{align}
where adjacent atoms are in blockade but next-nearest neighbors are not. Above
and in the following we label ancillas with a tilde and assign them numerical
indices. As for the \texttt{NOT}-gate, it is $\delta E_\text{\texttt{LNK}}
=0$ and $\Delta E_\text{\texttt{LNK}}=\Delta$ with the \texttt{LNK}-manifold
\begin{align}
    \H_0[\C_\text{\texttt{LNK}}]=\spn{\ket{0(1)0},\ket{1(0)1}}\,.
\end{align}
Here and in the following we mark the states of ancillas by parentheses.
Repeated amalgamation of elementary \texttt{LNK}-complexes results in
\texttt{LNK}-complexes of arbitrary length (always composed of an odd number
of atoms and with halved detuning at the endpoints). The two states in
$\H_0[\C_\text{\texttt{LNK}}]$ of such chains correspond to the two ground
states of an antiferromagnetic Ising chain.

\paragraph*{Step 3b: Implementing the \texttt{CPY}-complex.}
The purpose of the \texttt{CPY}-complex is to copy classical bits; it is
defined by the ``copy'' language $L_\text{\texttt{CPY}}=\{000,111\}$. The
\texttt{CPY}-complex is necessary because expansions in universal gates
can reuse inputs multiple times. Furthermore, circuits can be simplified
dramatically if intermediate results can be reused. In conventional drawings
of Boolean circuits, the possibility to copy bits is silently assumed whenever
one splits up wires. Again, in a physical implementation one has to provide
the means to do so.

The implementation of the \texttt{CPY}-complex is detailed in
\cref{fig:pxp_primitives}b. It is easy to see (\cref{app:cpy}) that there
cannot be a \texttt{CPY}-complex without ancillas because the configuration
$111$ excludes a Rydberg blockade between any of the three ports (which would
automatically render them completely uncorrelated). Adding a single ancilla
does the trick because the amalgamation of three \texttt{NOT}-complexes on
a single atom yields the desired complex by construction. The four atoms
are described by the Hamiltonian
\begin{align}
    H_\text{\texttt{CPY}}=-\Delta(n_A+n_Q+n_R)-3\Delta\,\tilde n_1\,,
    \label{eq:Hcpy}
\end{align}
and the geometry of the complex $\mathcal{C}_\texttt{CPY}$ is chosen so that
the ancilla is in blockade with the three ports, but these are not within
blockade of each other. In combination with \cref{eq:Hcpy}, this implements
the \texttt{CPY}-manifold
\begin{align}
    \H_0[\C_\text{\texttt{CPY}}]=\spn{\ket{000(1)},\ket{111(0)}}
\end{align}
with $\delta E_\text{\texttt{CPY}}=0$ and $\Delta
E_\text{\texttt{CPY}}=\Delta>0$.

\paragraph*{Step 3c: Implementing the \texttt{NOR}-complex.}
The \texttt{NOR}-complex is crucial as it realizes a
functionally complete two-bit gate; it is specified by the language
$L_\texttt{NOR}=\{001,010,100,110\}$. In contrast to the \texttt{LNK}- and
\texttt{CPY}-complexes, the \texttt{NOR}-complex cannot be bootstrapped from
the \texttt{NOT}-complex but must be constructed from scratch.

In \cref{app:nor} we show that a \texttt{NOR}-complex cannot be realized
with less than two ancillas in the PXP model. One implementation of a
\texttt{NOR}-complex is detailed in \cref{fig:pxp_primitives}c. The five
atoms are governed by the Hamiltonian
\begin{align}
    H_\text{\texttt{NOR}}=
    -\Delta(n_A+n_Q+\tilde n_1)-2\Delta (n_B+\tilde n_2)
\end{align}
which gives rise to the \texttt{NOR}-manifold
\begin{align}
    \H_0[\C_\text{\texttt{NOR}}]=\spn{%
        \begin{aligned}
            &\ket{001(01)},\ket{010(10)},\\
            &\ket{100(01)},\ket{110(00)}
        \end{aligned}
    }
\end{align}
with $\delta E_\text{\texttt{NOR}} =0$ and $\Delta
E_\text{\texttt{NOR}}=\Delta$; this requires that the atoms are arranged in
a ring-like blockade, as depicted in \cref{fig:pxp_primitives}c. Note that
the two ancillas are only necessary to enforce the degeneracy of the logical
states $010$ and $100$ with $110$. All remaining constraints come for free
with the Rydberg blockade. As we will show in \cref{sec:primitives}, the
\texttt{NOR}-complex in \cref{fig:pxp_primitives}c is not unique. We will
also see that the only fundamental Boolean gate that can be realized with
as few as five atoms is the \texttt{NOR}-gate, confirming our intuition in
Step~1 that the \texttt{NOR}-gate is the most natural on the Rydberg platform.

\paragraph*{Step 4: Constructing the $\fT$-complex.} 
To construct a complex $\C_\fT$ that implements the check function $\fT$
(more precisely: the language $L[\fT]$), one combines the three primitives
above according to an embedding $\Gamma(G_\fT)$. Since all vertices are
(at most) trivalent, it is easy to check that an amalgamation in the PXP
model is possible without geometrical obstructions, and that this procedure
yields an $\fT$-complex with $\delta E_\fT=0$ and $\Delta E_\fT\geq\Delta>0$.
At this point, we have a complex with $g=4K$ input ports on its boundary
that outputs $y=\fT(x_{e_1}^1,\dots)$ on a dedicated output port (also on
its boundary, but this is not important in the following):

\begin{align}
    \includegraphics[width=2.5cm,valign=c]{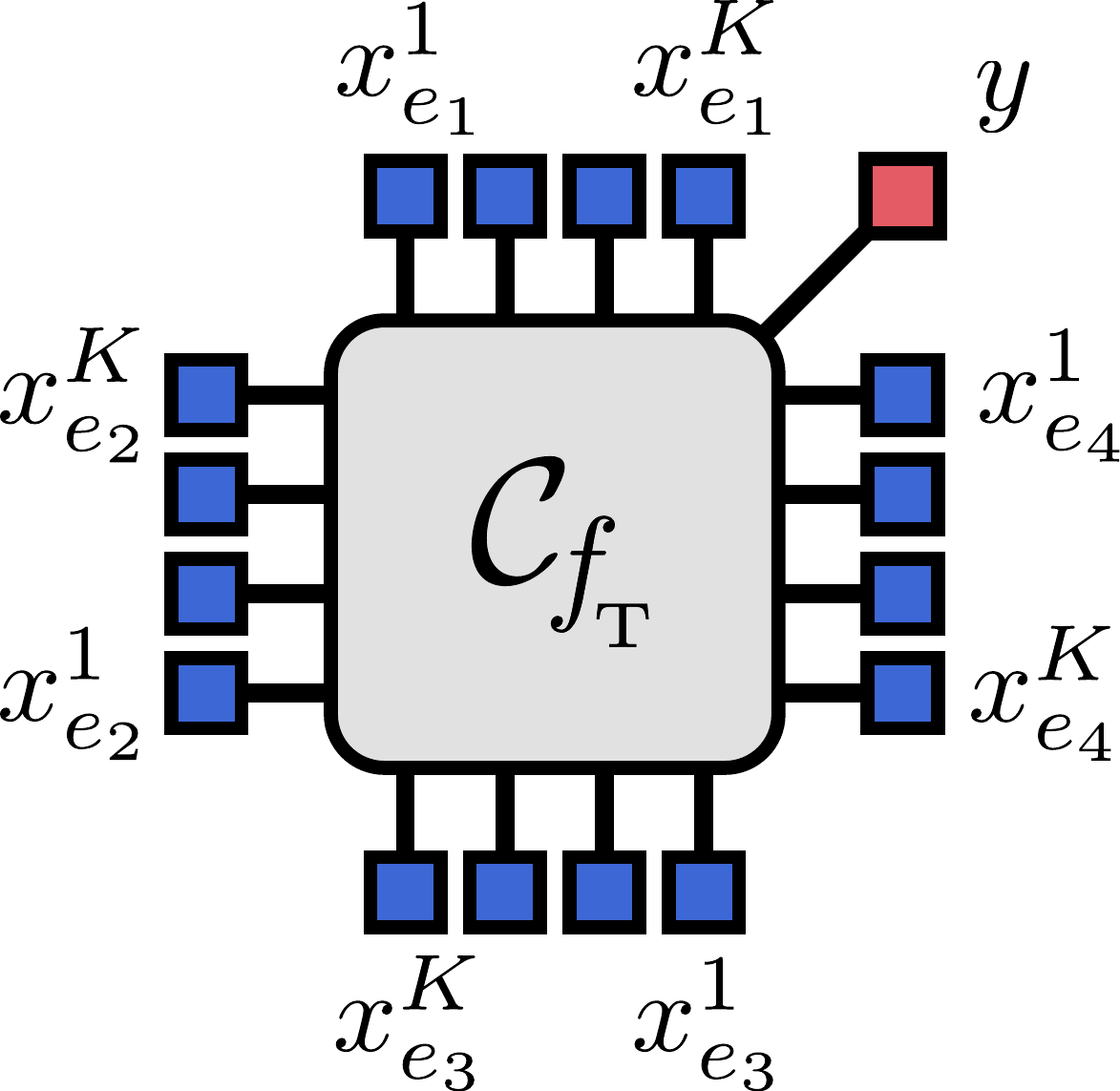}
\end{align}
To enforce the constraint $\fT(x_{e_1}^1,\dots)\stackrel{!}{=}1$, we only
have to add a local detuning on the output port to lower the energy of
valid configurations and gap out invalid ones. This boils down to a simple
modification of the check function complex,
\begin{align}
    \includegraphics[width=2.5cm,valign=c]{Cf.pdf}
    \quad\rightarrow\quad
    \includegraphics[width=2.5cm,valign=c]{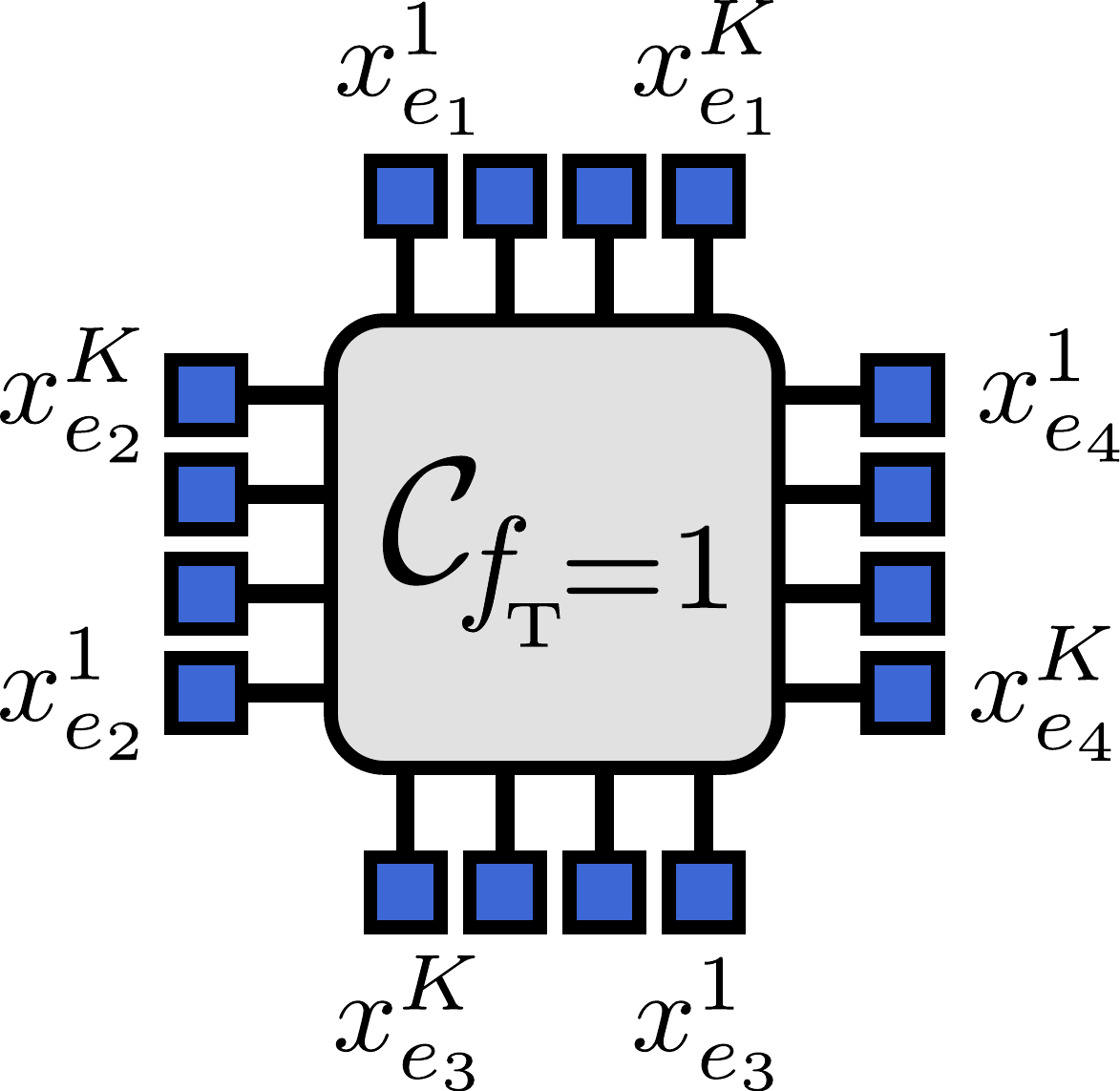}
    \label{eq:detuning}
\end{align}
where the output port is detuned and downgraded to an ancilla. The ground
state manifold of the modified complex $\C_{\fT=1}$ consists of all input
configurations for which $\fT(x_{e_1}^1,\dots)=1$.

\paragraph*{Step 5: Constructing $\C_\sub{T}$.} 
The complex $\C_{\fT=1}$ enforces the local constraint of the check function
on a single site of the lattice on which the tessellated target Hilbert space
$\H_\sub{T}=\H_\L[\fT]$ is defined. To construct $\C_\sub{T}$ for the full
system, place a copy $\C_{\fT=1}\mapsto \C_{\fT=1}^{s}$ on every site $s\in
V(\L)$ of the lattice, and amalgamate adjacent complexes at the corresponding
ports (possibly using \texttt{LNK}-complexes to avoid unwanted interactions):

\begin{align}
    \C_\sub{T}:=\quad\includegraphics[width=0.6\linewidth,valign=c]{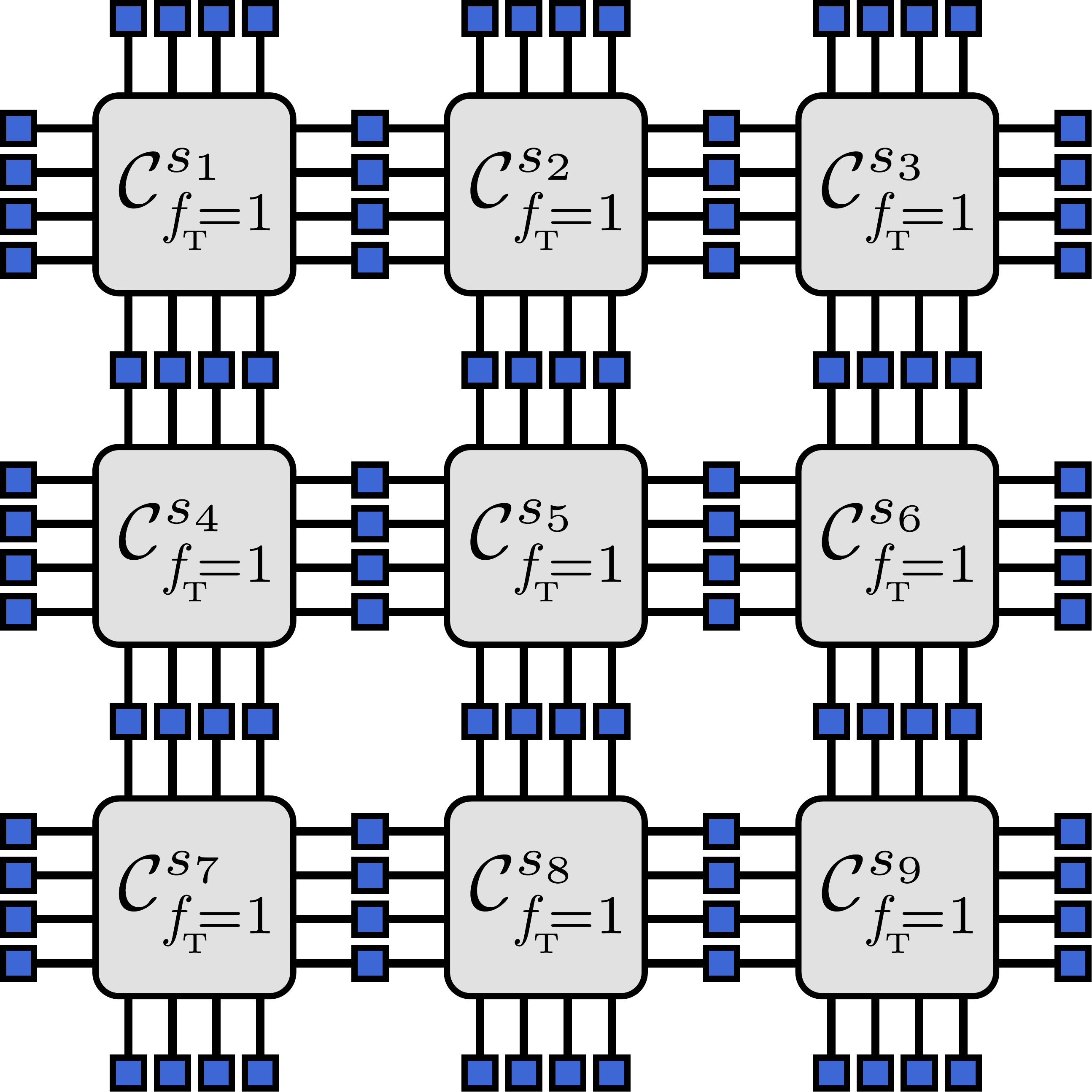}
\end{align}
By construction, the ground states of this complex are in one-to-one
correspondence with words $\vec x\in L_\L[\fT]$ (using the ports on the
edges denoted by blue squares). Note that here we show the construction for
a square lattice $\L$; the generalization to other lattices is straightforward.

This concludes the construction of $\C_\sub{T}$ such that
$\gH\stackrel{\text{loc}}{\simeq} \H_0[\C_\sub{T}]$ in the PXP
approximation. Note that the ancillas do not introduce additional degrees of
freedom in this subspace and local unitaries on $\gH$ map to local unitaries
on $\H_0[\C_\sub{T}]$ (the latter involve the ancillas of the $\C_{\fT=1}$
complexes and can therefore be very complicated---but they remain local
on $\H$).
\end{proof}

We conclude this section with a few remarks.

First, while the proof above \emph{is} constructive, one should not expect
the resulting structures to be useful in real-world applications, except for
simple special cases. In particular, we established no claims about optimality
(in any sense) of the constructed $\fT$-complexes; on this we focus in the
next \cref{sec:primitives}.

Second, the modification in \cref{eq:detuning} to construct $\C_{\fT=1}$ from
$\C_\fT$ is often straightforward to implement and can simplify the complex
considerably: When there are no blockades between the output port and some
of the input ports, one simply \emph{deletes} the output port along with all
ancillas that are in blockade with it. This removes all configurations of
input ports from the ground state manifold where the output was not excited
(see \cref{app:subcomplex}).
The removal of the output port may not be necessary at all if
the constraint $\fT(x_{e_1}^1,\dots)\stackrel{!}{=}1$ can be rewritten as
an equality of the form
\begin{align}
    f_1(x_{e_1}^1,\dots,x_{e_2}^1,\dots)\stackrel{!}{=}f_2(x_{e_3}^1,\dots,x_{e_4}^1,\dots)\,,
\end{align}
with Boolean functions $f_{1,2}$ that take only $2K$ inputs each. Then
$\C_{\fT=1}=\C_{f_1}\stackrel{\gamma}{\otimes}\C_{f_2}$ where the two
complexes are amalgamated at their output ports:

\begin{align}
    \includegraphics[width=2.5cm,valign=c]{Cf2.pdf}
    \quad=\quad
    \includegraphics[width=3.0cm,valign=c]{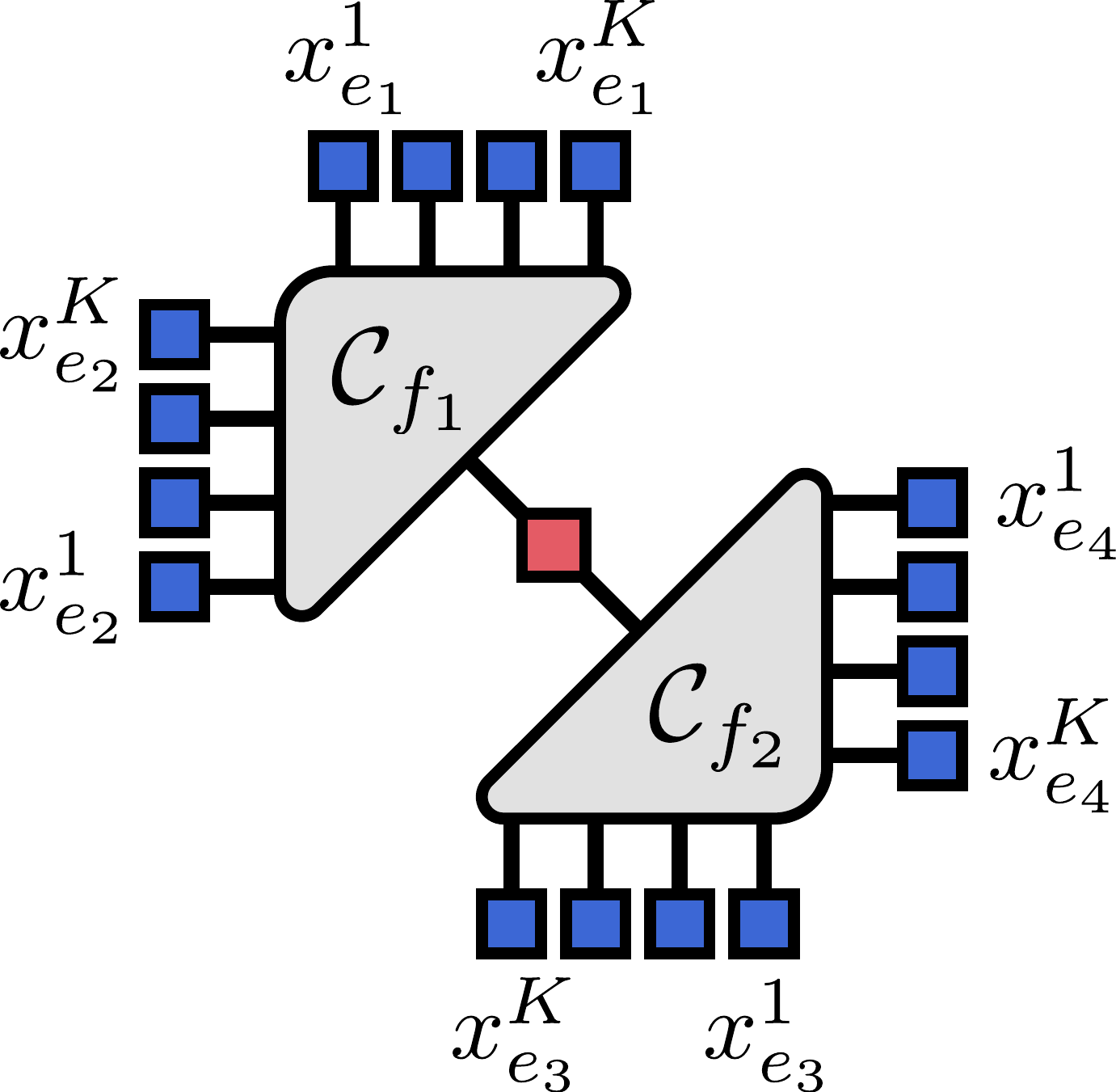}
\end{align}
An example for this construction can be found in \cref{subsec:toric}.

Lastly, the constructive proof implies that all Boolean functions
can be realized by a complex with \emph{bounded} detuning range
$\{1\Delta,2\Delta,3\Delta\}$, i.e., detunings do not grow with the
size (or depth) of the Boolean circuit. Indeed, since the ports of
the \texttt{LNK}-complex have detuning $1\Delta$, and the ports of the
\texttt{CPY}- and \texttt{NOR}-complexes at most $2\Delta$, amalgamations
of the latter two primitives via \texttt{LNK}-complexes produce atoms
with maximum detuning $3\Delta$. This result is particularly important for
experimental realizations that always operate within a bounded range of
applicable detunings.

\section{Logic Primitives}
\label{sec:primitives}

A crucial step of the proof in the previous section is to show that every
Boolean function $f$ can be realized by a Rydberg complex $\mathcal{C}_f$
in the sense that the language $L[f]$ of its truth table can be realized as
ground state manifold.
As mentioned above, the complexes that arise from the decomposition of $f$ into
\texttt{LNK}-, \texttt{CPY}- and \texttt{NOR}-primitives are typically large
and convoluted. For example, the decomposition of a simple \texttt{AND}-gate
($\wedge$) into \texttt{NOR}-gates reads
\begin{align}
    A\wedge B = (A \downarrow A) \downarrow (B \downarrow B)\,,
\end{align}
which would require two \texttt{CPY}- and three \texttt{NOR}-complexes, wired
together by a bunch of \texttt{LNK}-complexes so that the resulting complex
requires more than 20 atoms. As this is way too much overhead for a simple
gate, the question arises whether important primitives of Boolean logic can
be realized by complexes that are much smaller than the ones described by
the \texttt{NOR}-decomposition in \cref{sec:completeness}.

The answer is positive: In the following, we discuss provably minimal
complexes for the most important gates of Boolean logic, all of which improve
significantly over the na\"ive \texttt{NOR}-decomposition. Besides the usual
gates of Boolean algebra, \texttt{NOT} ($\neg$ or $\overline{\bullet}$),
\texttt{AND} ($\wedge$), and \texttt{OR} ($\vee$), we search for minimal
complexes that realize the following common logic gates (given in disjunctive
normal form):
\begin{subequations}
    \begin{align}
        \texttt{NOR:}\quad A\downarrow B &\phantom{:}=\overline A\wedge \overline B\\
        \texttt{NAND:}\quad A\uparrow B &:=\overline A\vee \overline B\\
        \texttt{XOR:}\quad A\oplus B &:=(A\wedge\overline B)\vee(\overline A\wedge B)\label{eq:XOR}\\
        \texttt{XNOR:}\quad A\odot B &:= (A\wedge B) \vee (\overline A \wedge \overline B)\,.
        \label{eq:XNOR}
    \end{align}
\end{subequations}

\clearpage

\makeatletter\onecolumngrid@push\makeatother
\begin{figure*}
    \centering
    \scalebox{0.96}{%
    \begin{tikzpicture}
        \node[figure,anchor=north west] at (-2,19.7) {\includegraphics[width=1.5cm]{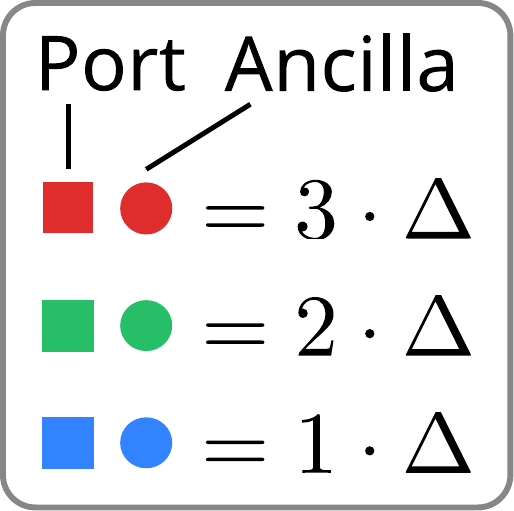}};
        \node[figure] at (11.3,18.3) {\includegraphics[width=3.5cm]{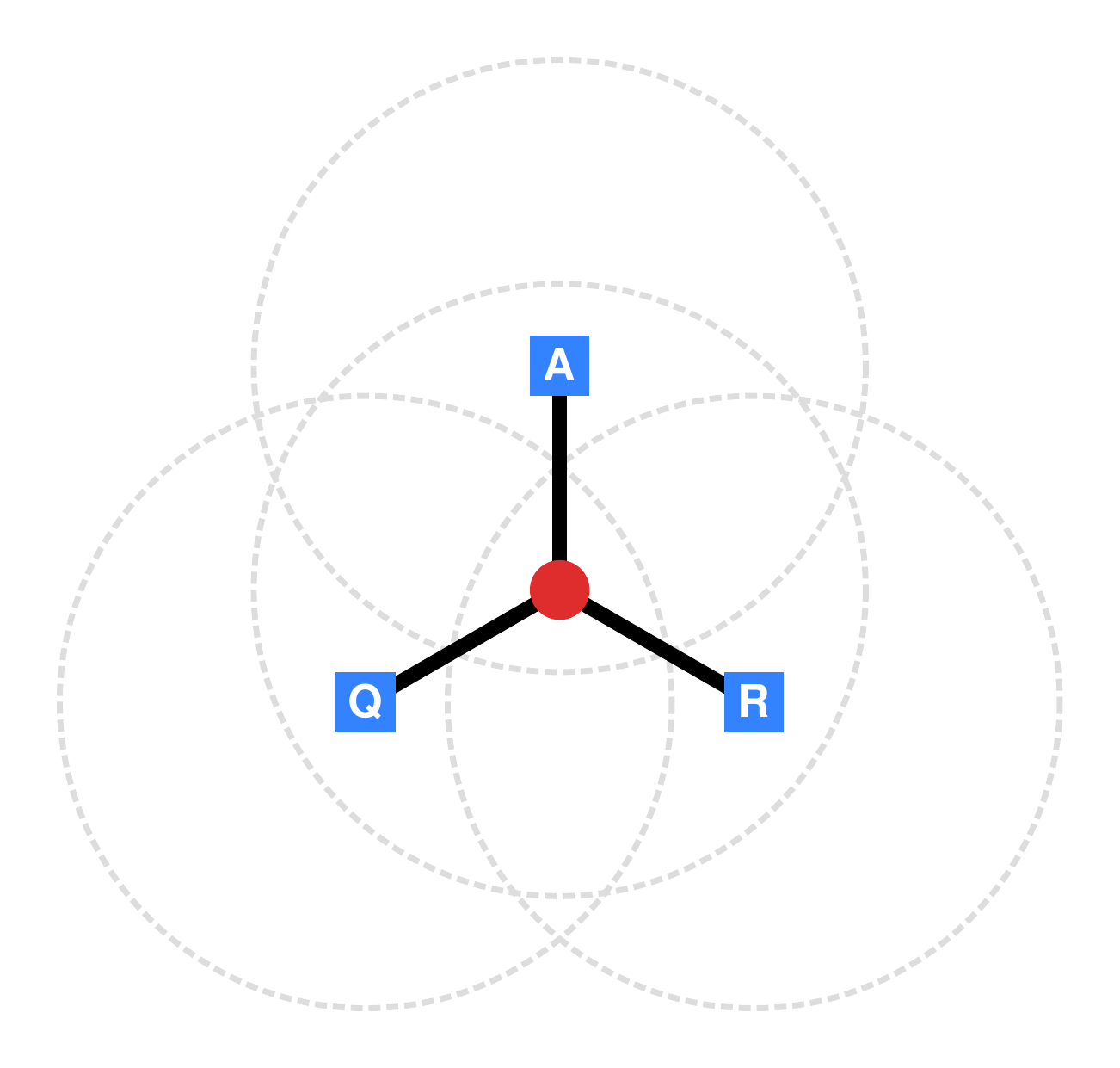}};
        \node[figure] at (14.0,19.3) {\includegraphics[width=3.2cm]{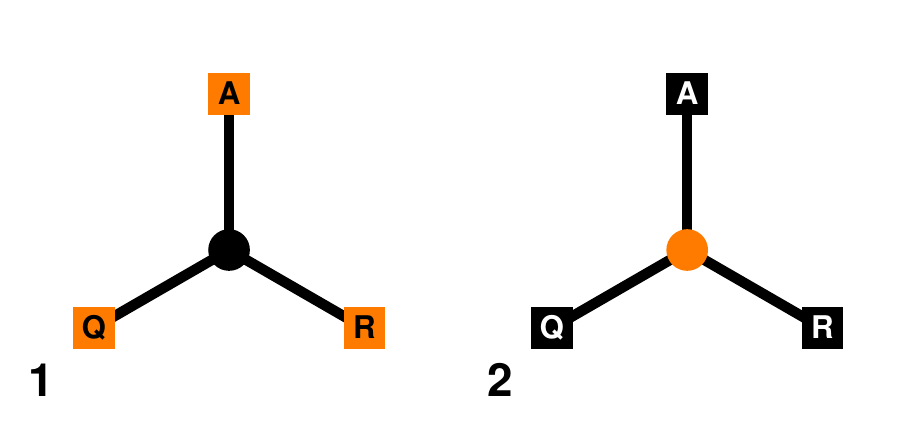}};
        \node[figure] at (14,17.8) {\includegraphics[width=1.2cm]{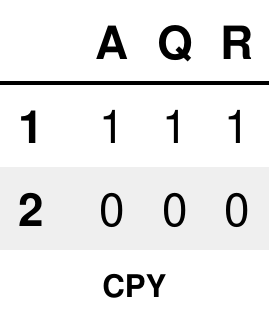}};
        \node[figure] at (12.2,20) {\includegraphics[width=1.5cm]{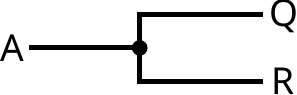}};
        \node[figtext,anchor=center] at (10.9,20.0) {\texttt{CPY}};
        \node[figure] at (6.5,18.7) {\includegraphics[width=3.5cm]{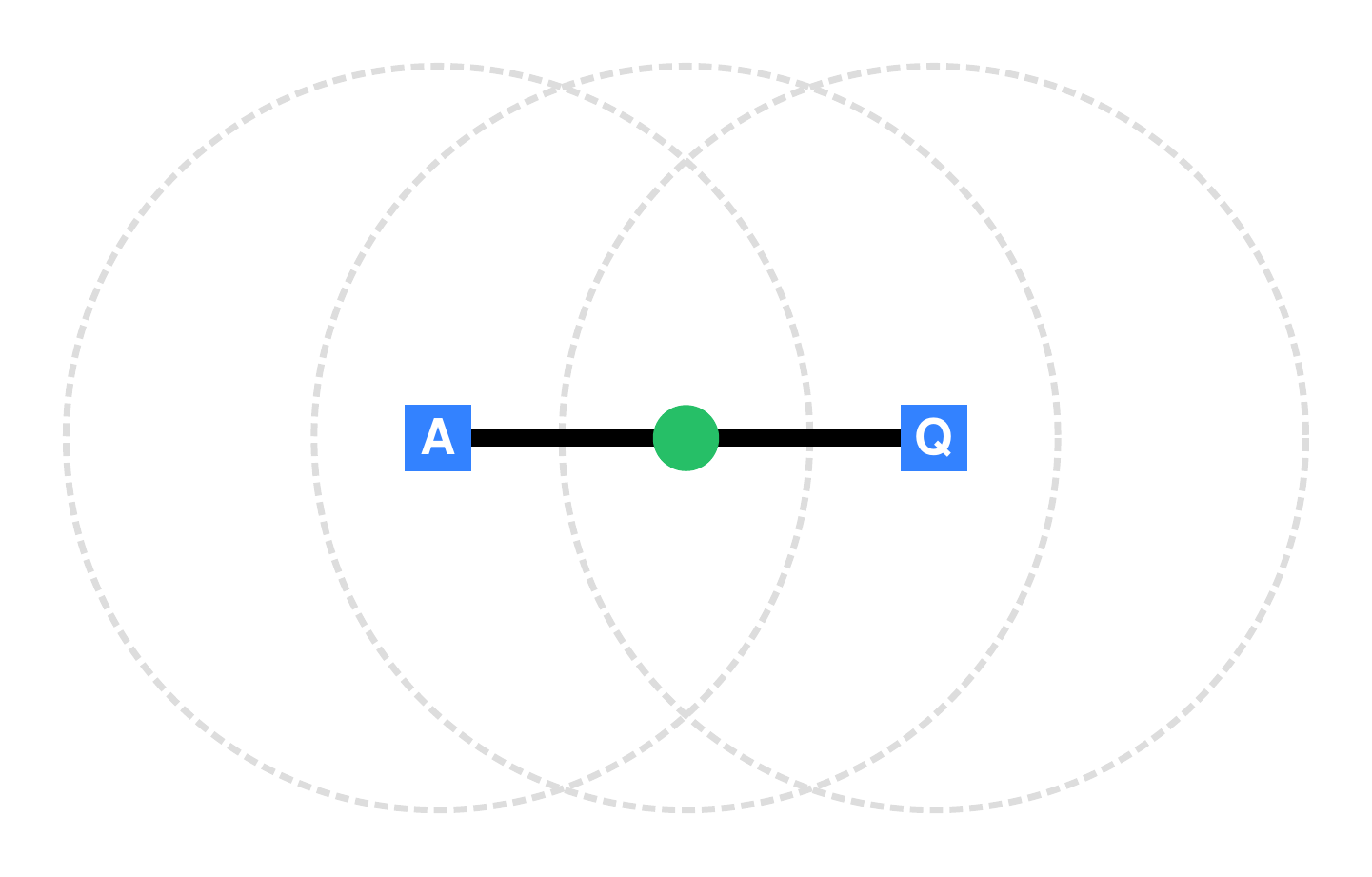}};
        \node[figure] at (6.5,17.6) {\includegraphics[width=3.5cm]{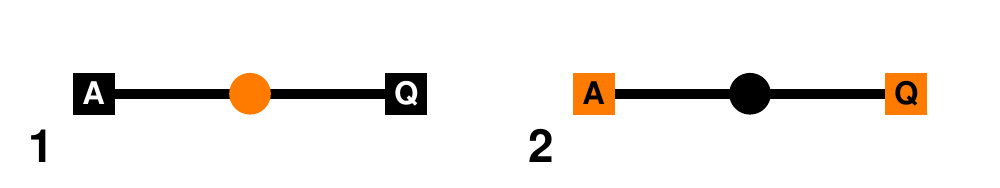}};
        \node[figure] at (8.7,18.7) {\includegraphics[width=0.9cm]{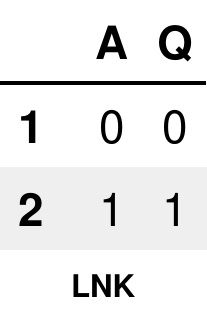}};
        \node[figure] at (7.4,20) {\includegraphics[width=1.5cm]{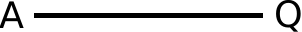}};
        \node[figtext,anchor=center] at (6.1,20.0) {\texttt{LNK}};
        \node[figure] at (1.5,18.7) {\includegraphics[width=3cm]{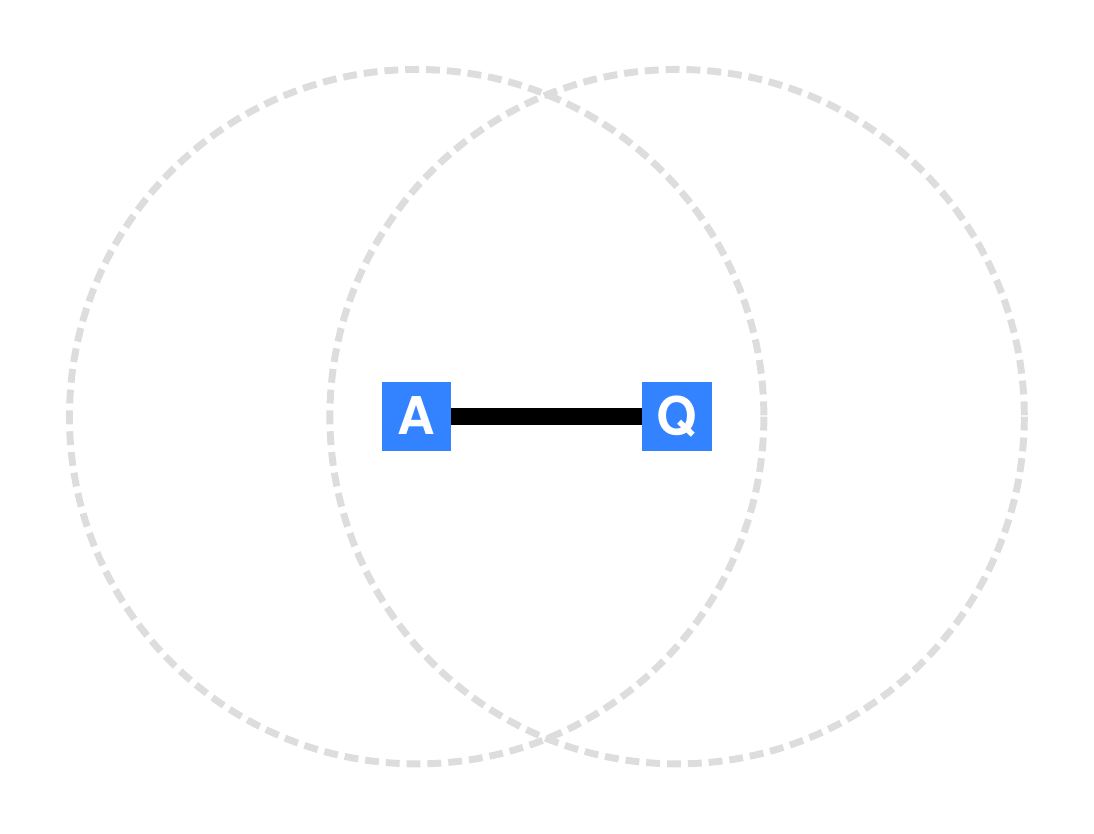}};
        \node[figure] at (1.5,17.6) {\includegraphics[width=2.5cm]{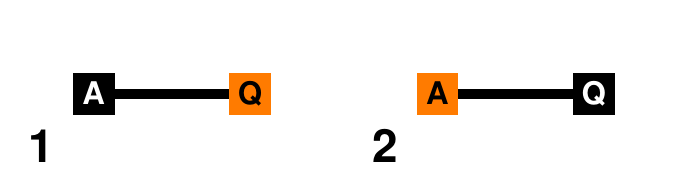}};
        \node[figure] at (3.5,18.7) {\includegraphics[width=0.9cm]{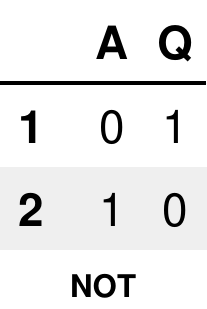}};
        \node[figure] at (2.4,20) {\includegraphics[width=1.5cm]{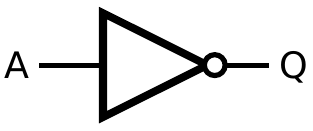}};
        \node[figtext,anchor=center] at (0.9,20.0) {\texttt{NOT ($\neg$)}};
        \node[figure] at (3,15.6) {\includegraphics[width=4cm]{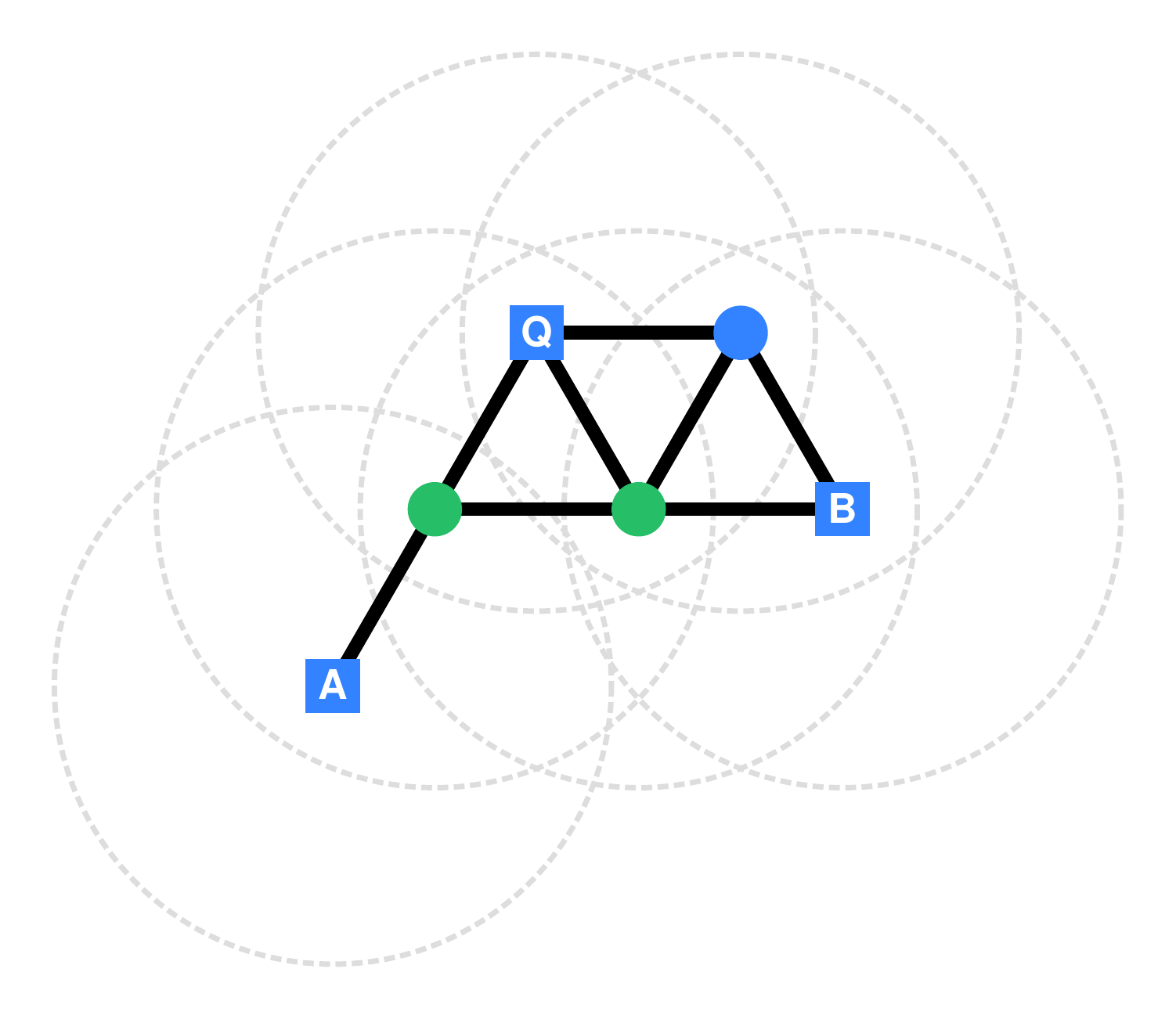}};
        \node[figure] at (9,15.6) {\includegraphics[width=7cm]{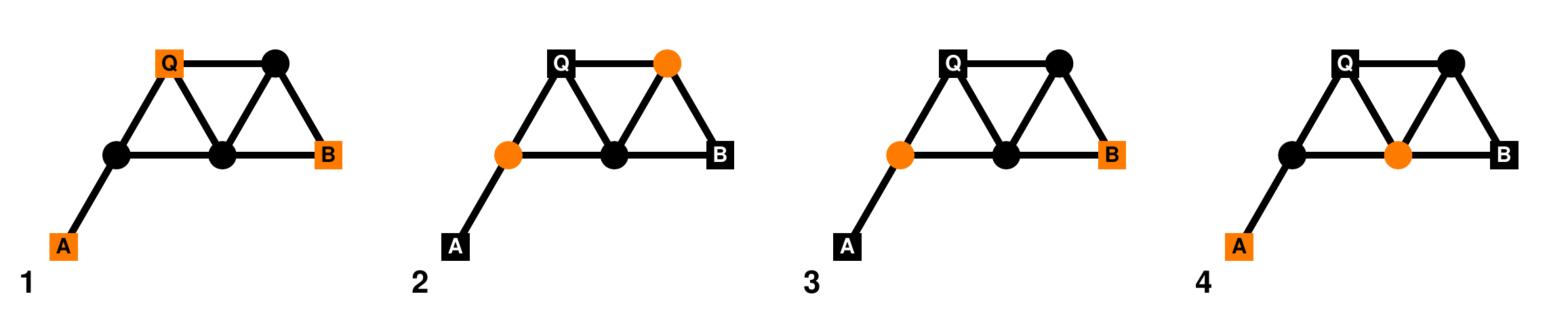}};
        \node[figure] at (14,15.6) {\includegraphics[width=1.2cm]{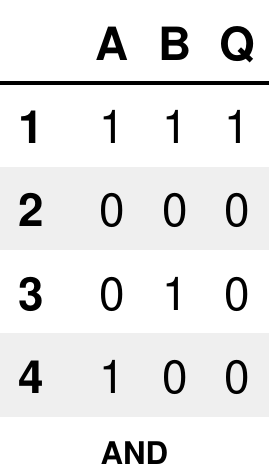}};
        \node[figure] at (0,15.3) {\includegraphics[width=1.5cm]{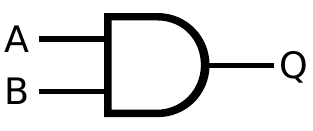}};
        \node[figtext,anchor=center] at (0,15.9) {\texttt{AND ($\wedge$)}};
        \node[figure] at (3,12.8) {\includegraphics[width=3.7cm]{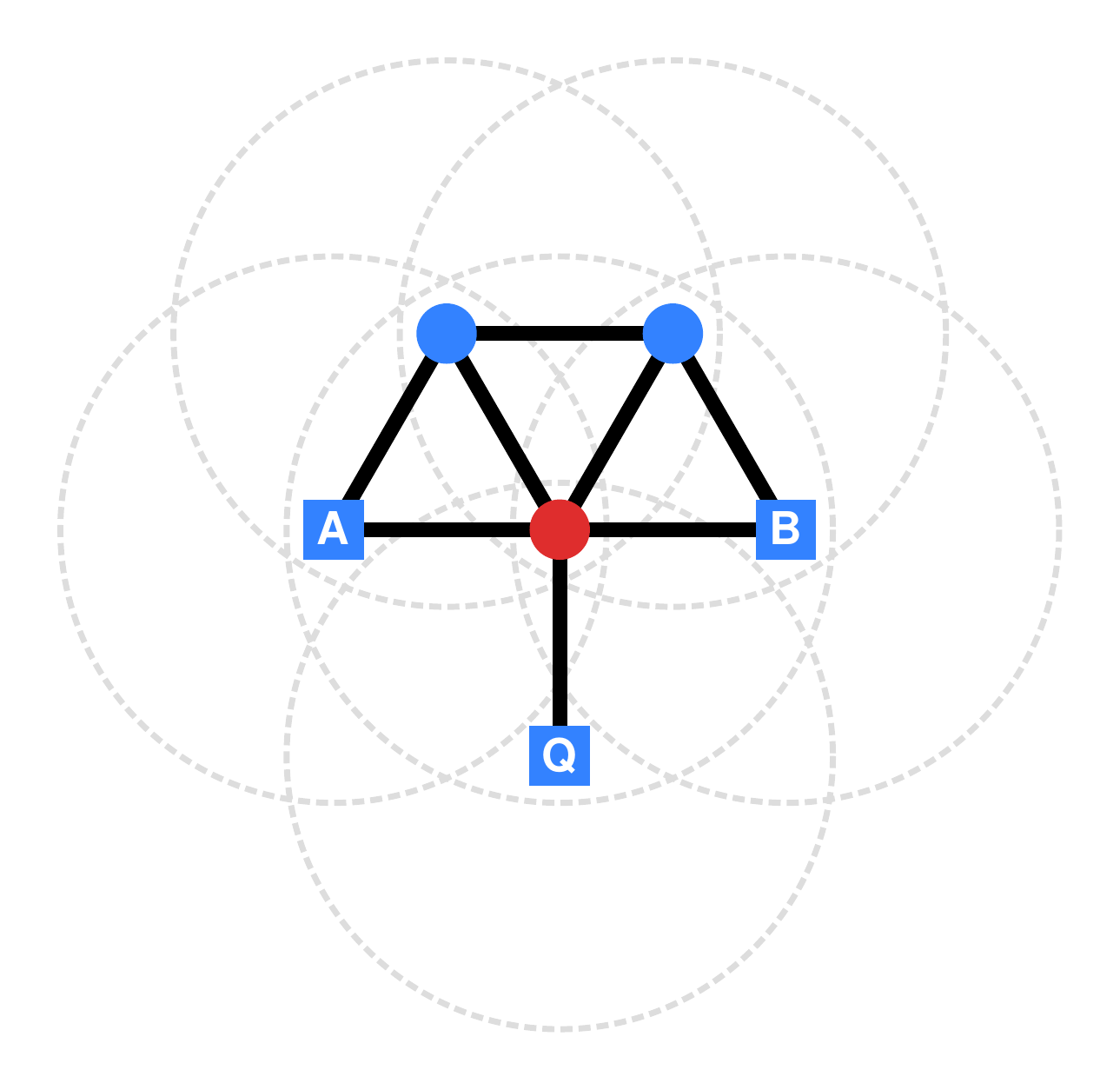}};
        \node[figure] at (9,12.8) {\includegraphics[width=7cm]{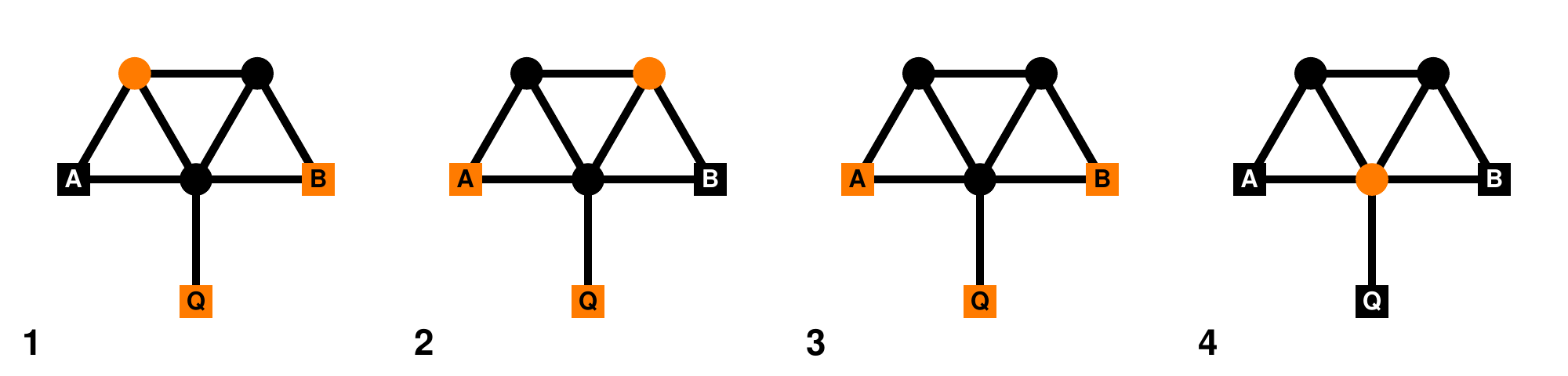}};
        \node[figure] at (14,12.8) {\includegraphics[width=1.2cm]{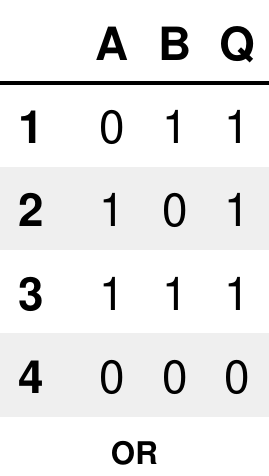}};
        \node[figure] at (0,12.5) {\includegraphics[width=1.5cm]{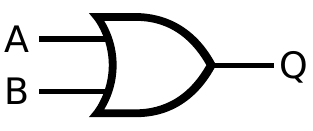}};
        \node[figtext,anchor=center] at (0,13.1) {\texttt{OR ($\vee$)}};
        \node[figure] at (3,10.1) {\includegraphics[width=3.8cm]{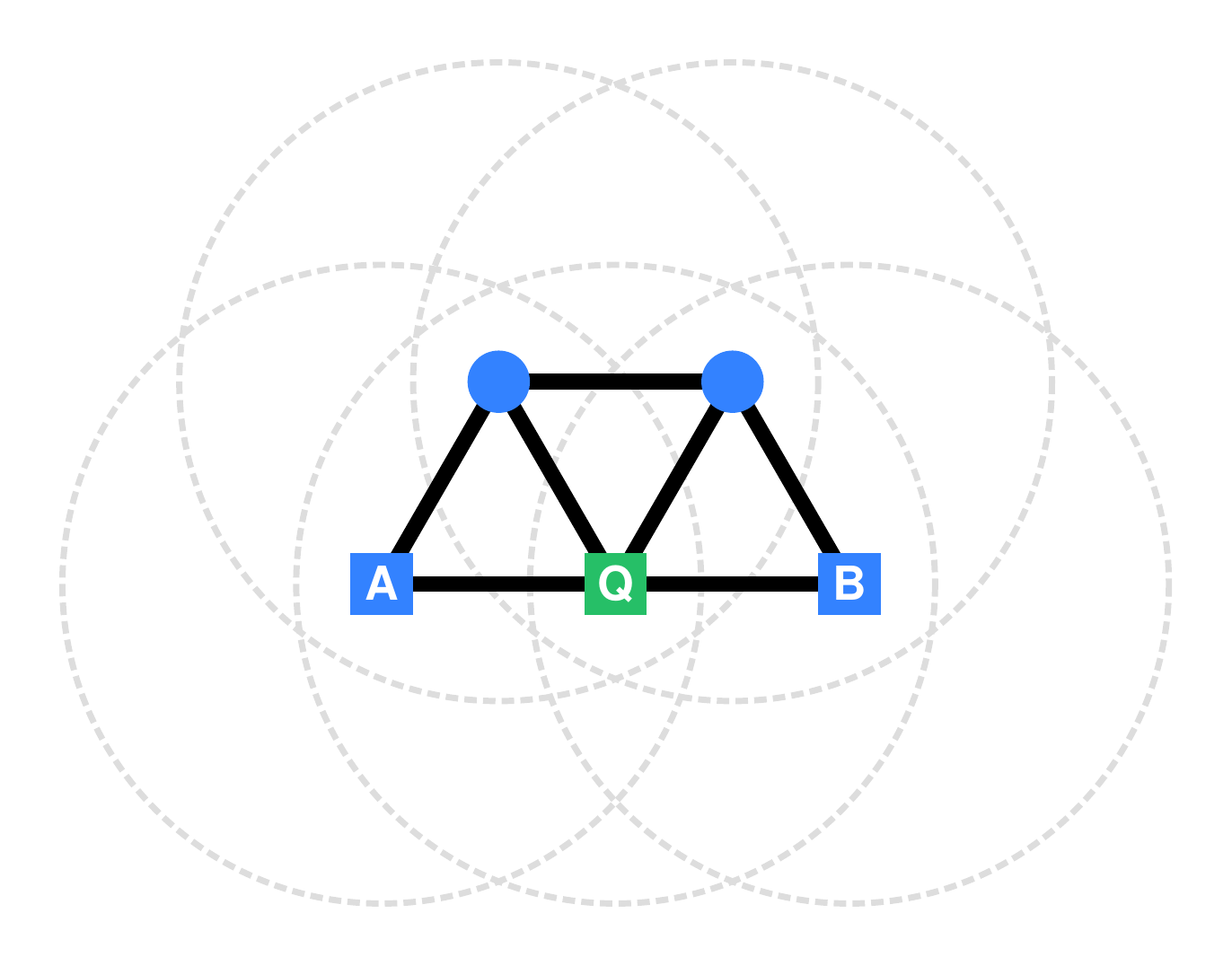}};
        \node[figure] at (9,10.1) {\includegraphics[width=7cm]{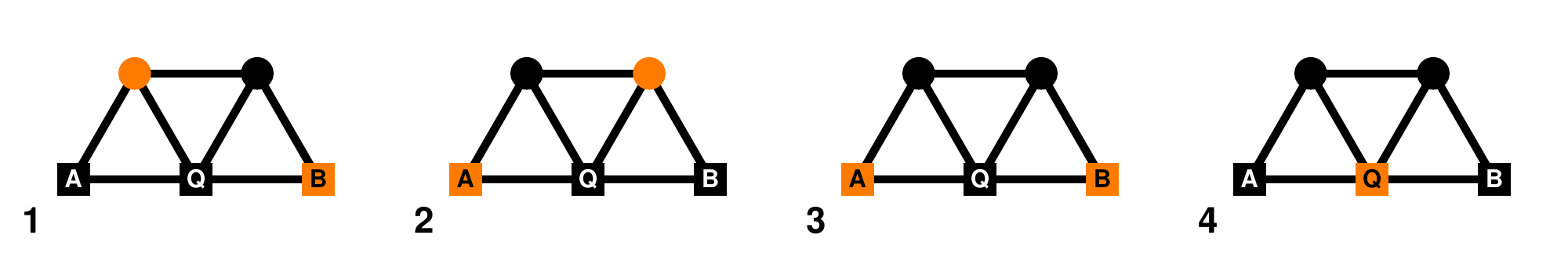}};
        \node[figure] at (14,10.1) {\includegraphics[width=1.2cm]{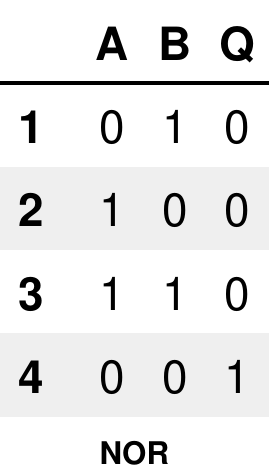}};
        \node[figure] at (0,9.8) {\includegraphics[width=1.5cm]{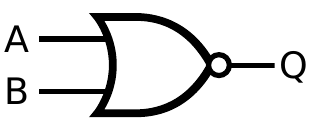}};
        \node[figtext,anchor=center] at (0,10.4) {\texttt{NOR ($\downarrow$)}};
        \node[figure] at (3,7.2) {\includegraphics[width=3.8cm]{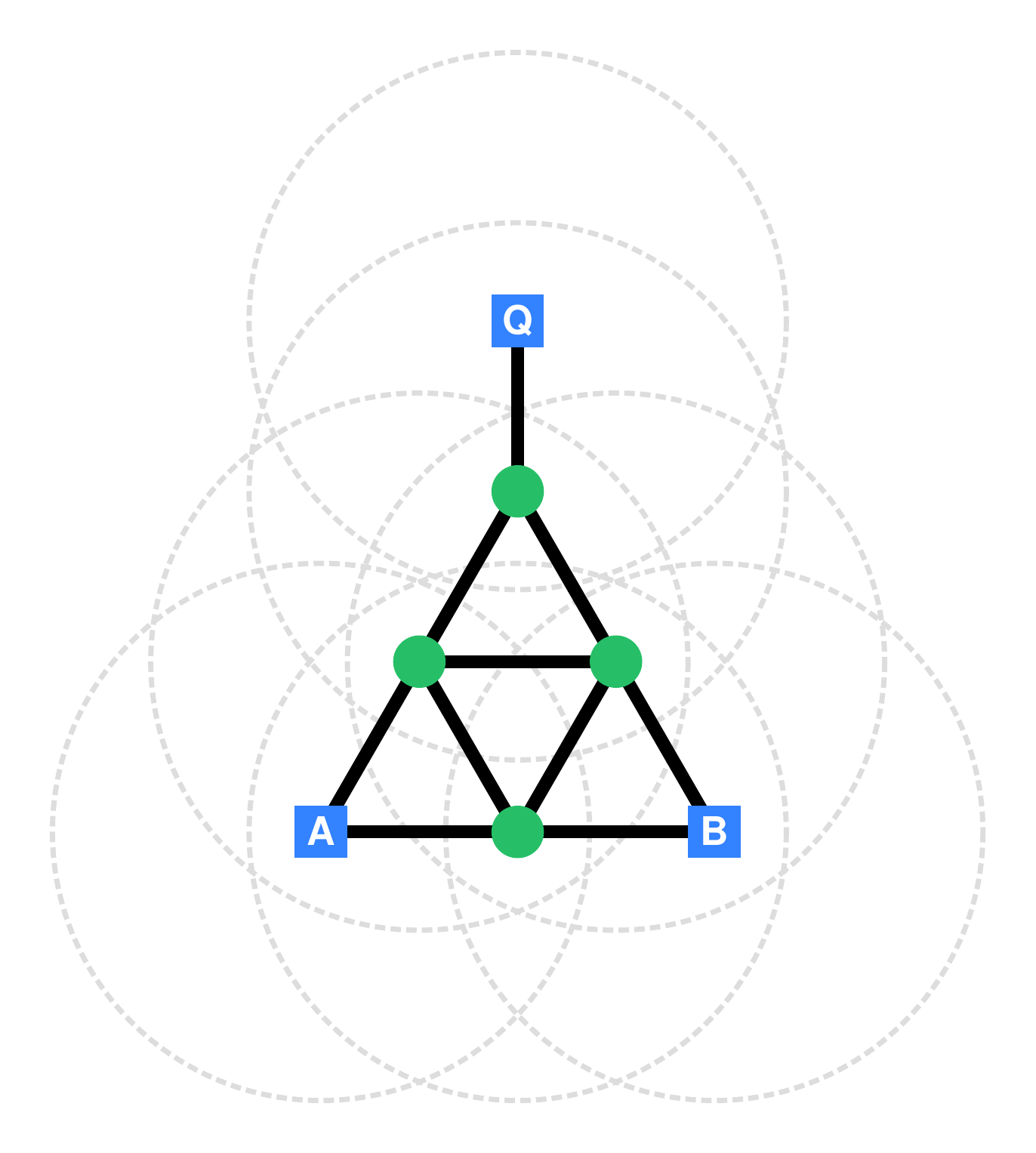}};
        \node[figure] at (9,7.2) {\includegraphics[width=7cm]{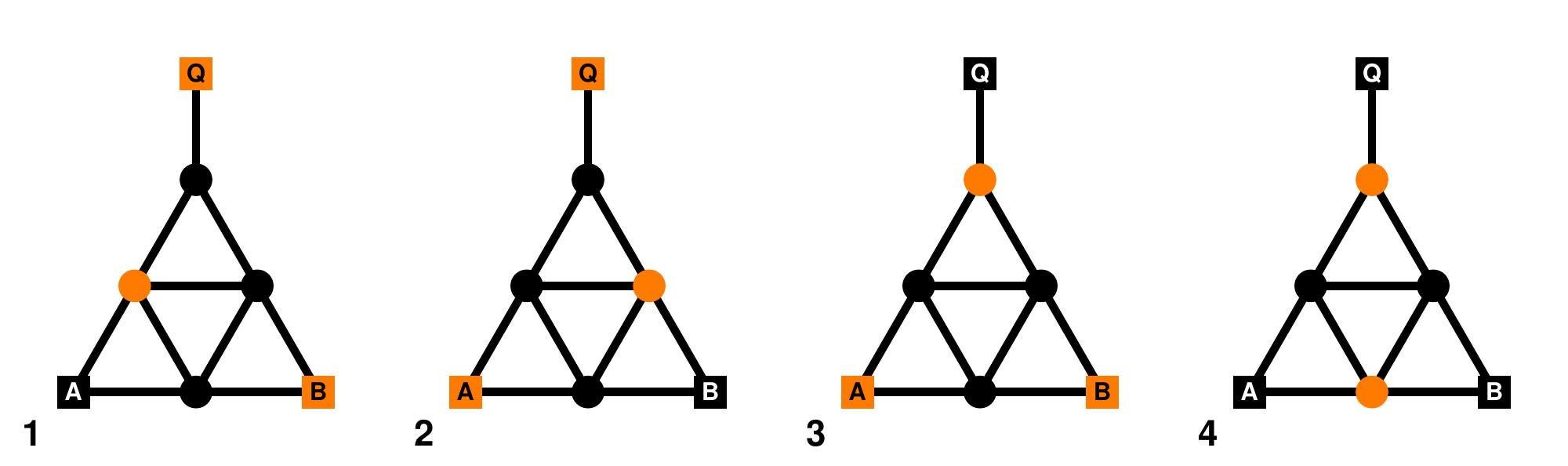}};
        \node[figure] at (14,7.2) {\includegraphics[width=1.2cm]{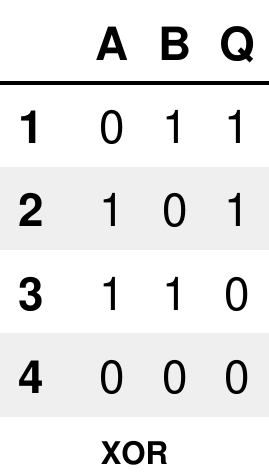}};
        \node[figure] at (0,6.9) {\includegraphics[width=1.5cm]{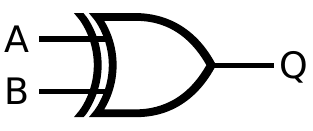}};
        \node[figtext,anchor=center] at (0,7.5) {\texttt{XOR ($\oplus$)}};
        \node[figure] at (3,3.9) {\includegraphics[width=4cm]{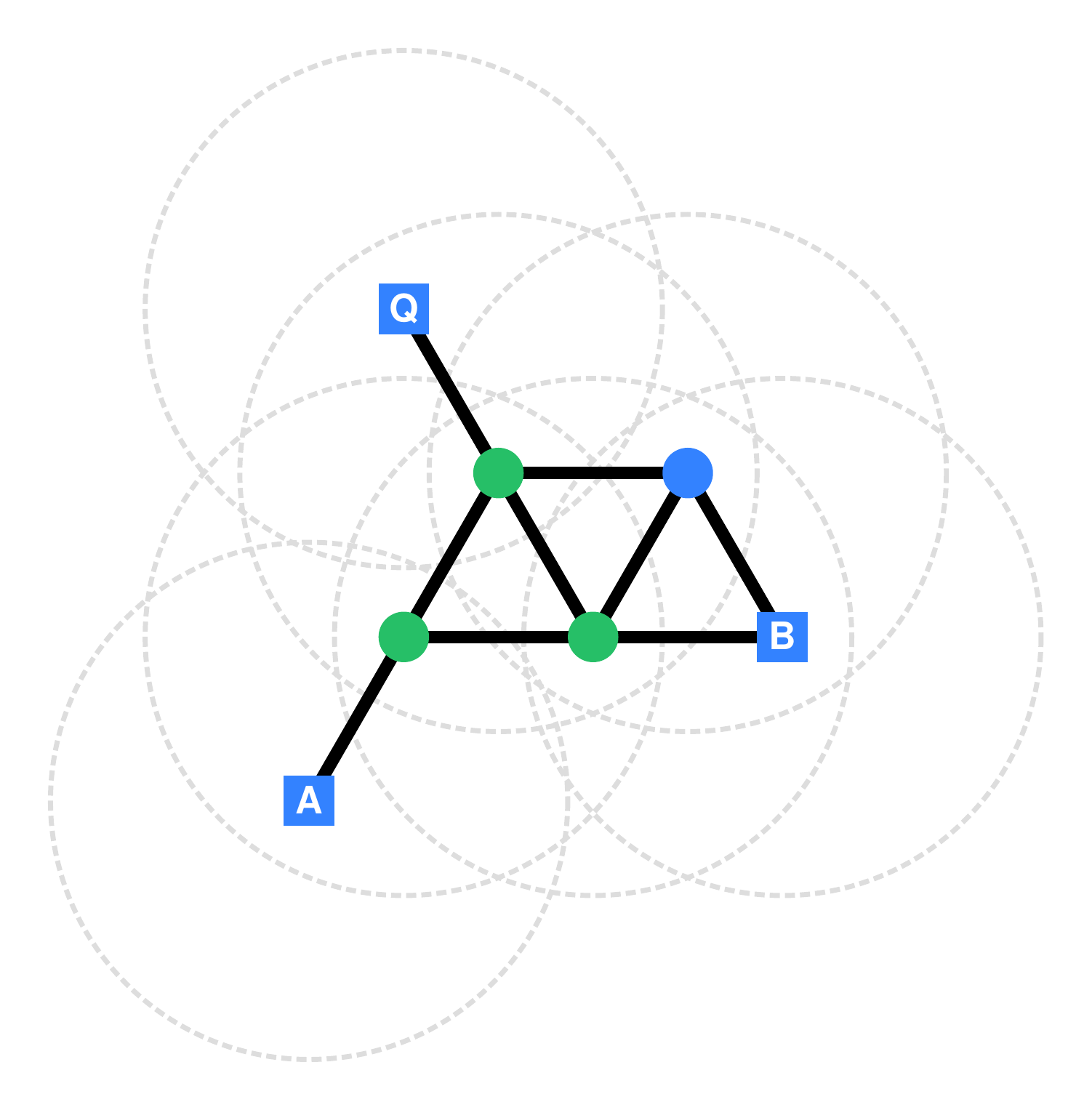}};
        \node[figure] at (9,3.9) {\includegraphics[width=7cm]{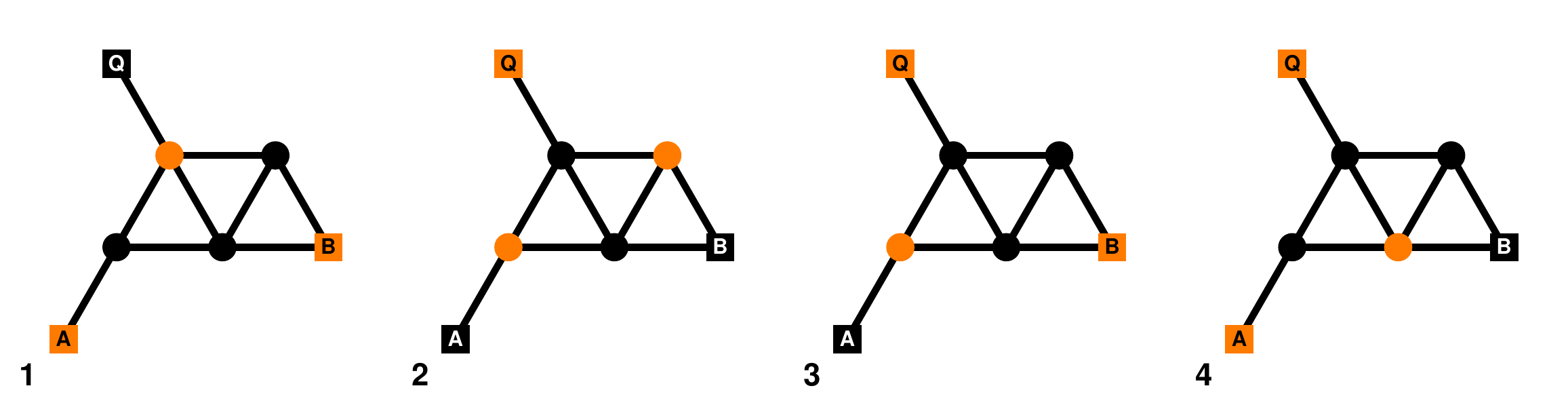}};
        \node[figure] at (14,3.9) {\includegraphics[width=1.2cm]{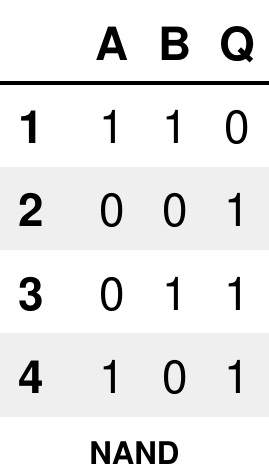}};
        \node[figure] at (0,3.6) {\includegraphics[width=1.5cm]{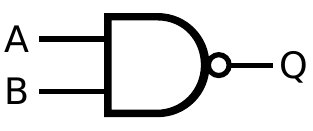}};
        \node[figtext,anchor=center] at (0,4.2) {\texttt{NAND ($\uparrow$)}};
        \node[figure] at (3,1) {\includegraphics[width=3.8cm]{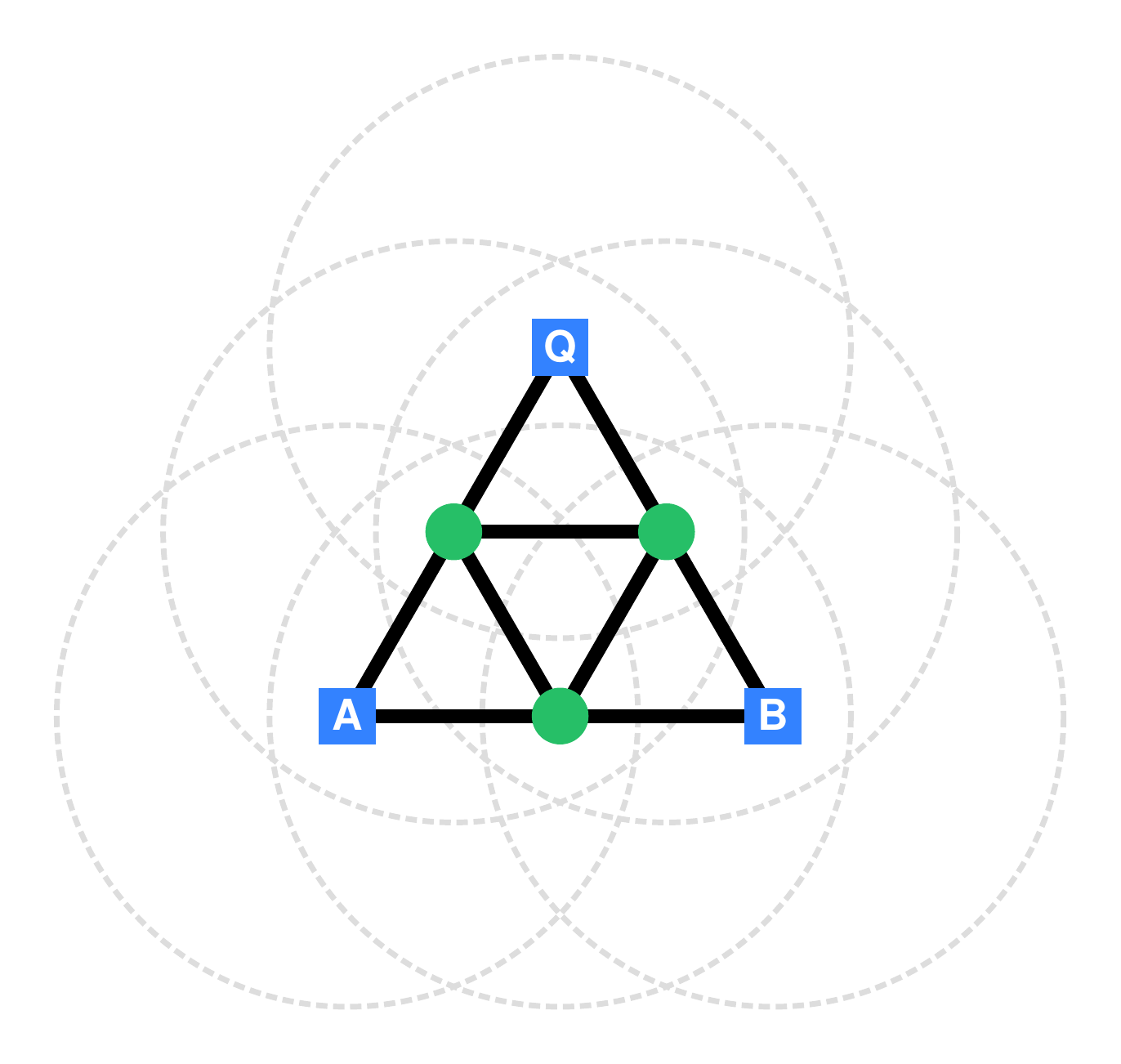}};
        \node[figure] at (9,1) {\includegraphics[width=7cm]{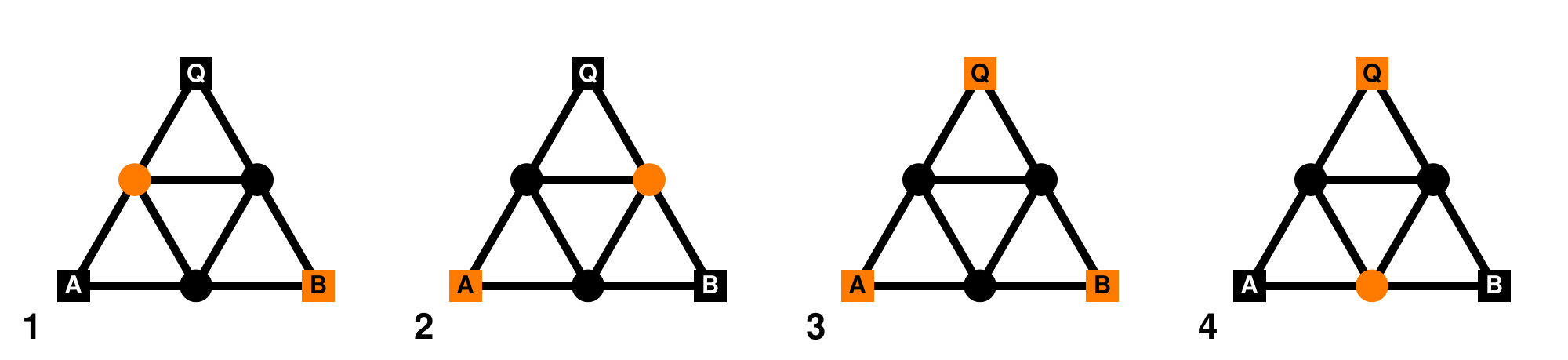}};
        \node[figure] at (14,1) {\includegraphics[width=1.2cm]{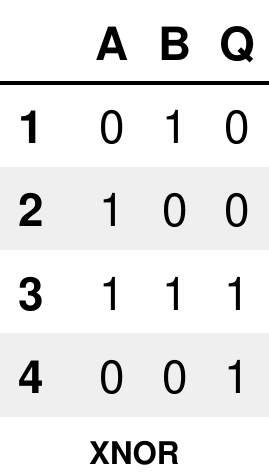}};
        \node[figure] at (0,0.7) {\includegraphics[width=1.5cm]{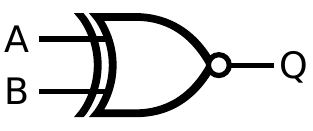}};
        \node[figtext,anchor=center] at (0,1.3) {\texttt{XNOR ($\odot$)}};
        %
    \end{tikzpicture}%
    }
    \caption{%
	\emph{Common logic primitives.}
    Rydberg complexes for the most common primitives of Boolean
    circuits. All complexes are provably minimal, see \cref{app:pxp}. Note
    that \emph{minimal} complexes are not necessarily \emph{unique};
    e.g.\ the shown \texttt{NOR}-gate is an alternative to the one in
    \cref{fig:pxp_primitives}c, both of which are minimal. For each complex
    we show (1) the geometry with blockade radii (gray dashed circles),
    (2) the complete ground state manifold (orange: $\ket{1}_i$, black:
    $\ket{0}_i$), and (3) the truth table of the ports (labeled atoms) in
    the ground state manifold. The rows of the truth tables correspond to
    the numbered ground state configurations. Colors of ancillas and ports
    in the geometry encode the detuning (see key). Atoms in blockade are
    connected by black solid lines.
    }
    \label{fig:logic_primitives}
\end{figure*}
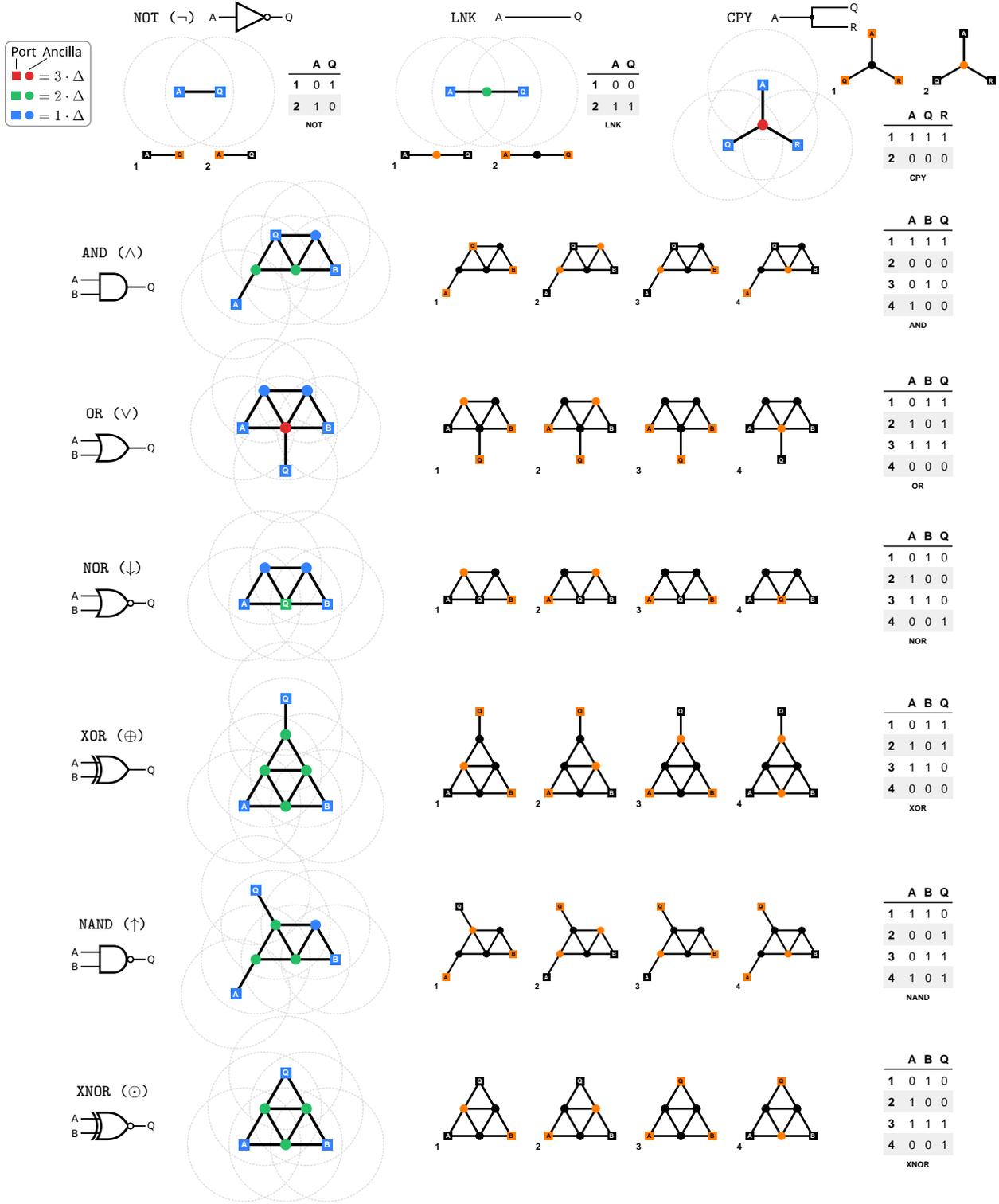
\clearpage
\makeatletter\onecolumngrid@pop\makeatother



%
\noindent Of these gates, only \texttt{NOR} and \texttt{NAND} are universal
on their own. The following identities show that some of these gates are
simply inverted versions of others (we will use this below):
\begin{subequations}
    \label{eq:relations}
    \begin{align}
        A\wedge B&=\overline{A\uparrow B}\\
        A\vee B&=\overline{A\downarrow B}\label{eq:or}\\
        A\oplus B&=\overline{A\odot B}\,.
    \end{align}
\end{subequations}

Of the gates $\{\neg,\vee,\wedge,\uparrow,\downarrow,\oplus,\odot\}$, we
already know minimal complexes for \texttt{NOT} (2 atoms) and \texttt{NOR}
(5 atoms), recall \cref{sec:completeness}.
Using \cref{eq:or}, we can immediately construct an \texttt{OR}-complex with
six atoms by amalgamation of a \texttt{NOT}-complex to the output port of a
\texttt{NOR}-complex (remember \cref{fig:nutshell}). However, it is unclear
whether this complex is \emph{minimal}, i.e., cannot be realized with fewer
atoms. Therefore we systematically devised proofs that a given truth table
\emph{cannot} be realized with a given number $N$ of atoms, starting at $N=3$
for each gate, and increasing the number incrementally until the proof fails,
i.e., realizations can no longer be excluded.
These arguments are quite technical and can be found in \cref{app:pxp}.
However, this approach has two benefits: First, it provides rigorous lower
bounds on how many atoms are needed to realize a given gate, and second, it
often provides a blueprint for the construction of a minimal complex that
saturates this bound by carefully observing \emph{why} one cannot exclude
realizations with a given number of atoms.

To complement this rigorous approach, we conducted a brute force search on a
computer that exhaustively scans for (small) complexes that realize a given
truth table. In accordance with our proofs, we found solutions with the minimal
atom number for a given truth table (in addition, we also found non-minimal
complexes). Interestingly, there were alternative minimal solutions that we
missed in our manual approach; so minimal complexes are not necessarily unique.

A selection of provably minimal complexes for all important Boolean
primitives is shown in \cref{fig:logic_primitives} (for the sake
of completeness, we include the \texttt{NOT}-, \texttt{LNK}- and
\texttt{CPY}-complexes discussed in \cref{sec:completeness}). There
are a few comments in order. First, an example of non-unique minimal
complexes is the depicted \texttt{NOR}-complex built from five atoms
arranged in a triangular structure (cf.\ the ring-shaped structure in
\cref{fig:pxp_primitives}c). Second, the six-atom \texttt{OR}-complex
we proposed above indeed is minimal, though not unique either. Third,
the selection of minimal complexes in \cref{fig:logic_primitives} for
$\{\vee,\wedge,\uparrow,\downarrow,\oplus,\odot\}$ all build around the
triangle-based core of the \texttt{NOR}-complex, once again emphasizing
its central role in the context of Rydberg complexes. Finally, it turns
out that the relations \eref{eq:relations} are all reflected in the
minimal complexes, e.g., the amalgamation of a \texttt{NOT}-complex and
a \texttt{XNOR}-complex yields a minimal \texttt{XOR}-complex; similar
constructions hold for \texttt{NAND} and \texttt{AND} as well as \texttt{NOR}
and \texttt{OR}. If we recall the relation between \texttt{NOT} and the minimal
\texttt{LNK}-complex, the general picture emerges that inverting complexes
are simpler (by one atom) than non-inverting ones. This is understandable
in so far as inversion is the most basic operation the Rydberg blockade is
capable of, thus leading to the simplest complexes. This is in contrast to
the notation for Boolean circuits known from electrical engineering where
inverting gates are represented by more complicated symbols than their
non-inverting counterparts (\cref{fig:logic_primitives}).

\section{Crossing}
\label{sec:crossing}

\begin{figure*}[tb]
    \centering
    \begin{tikzpicture}
        \node[label] at (0,4.2) {\lbl{a}};
        \node[figure] at (2.0,2.2) {\includegraphics[width=4.5cm]{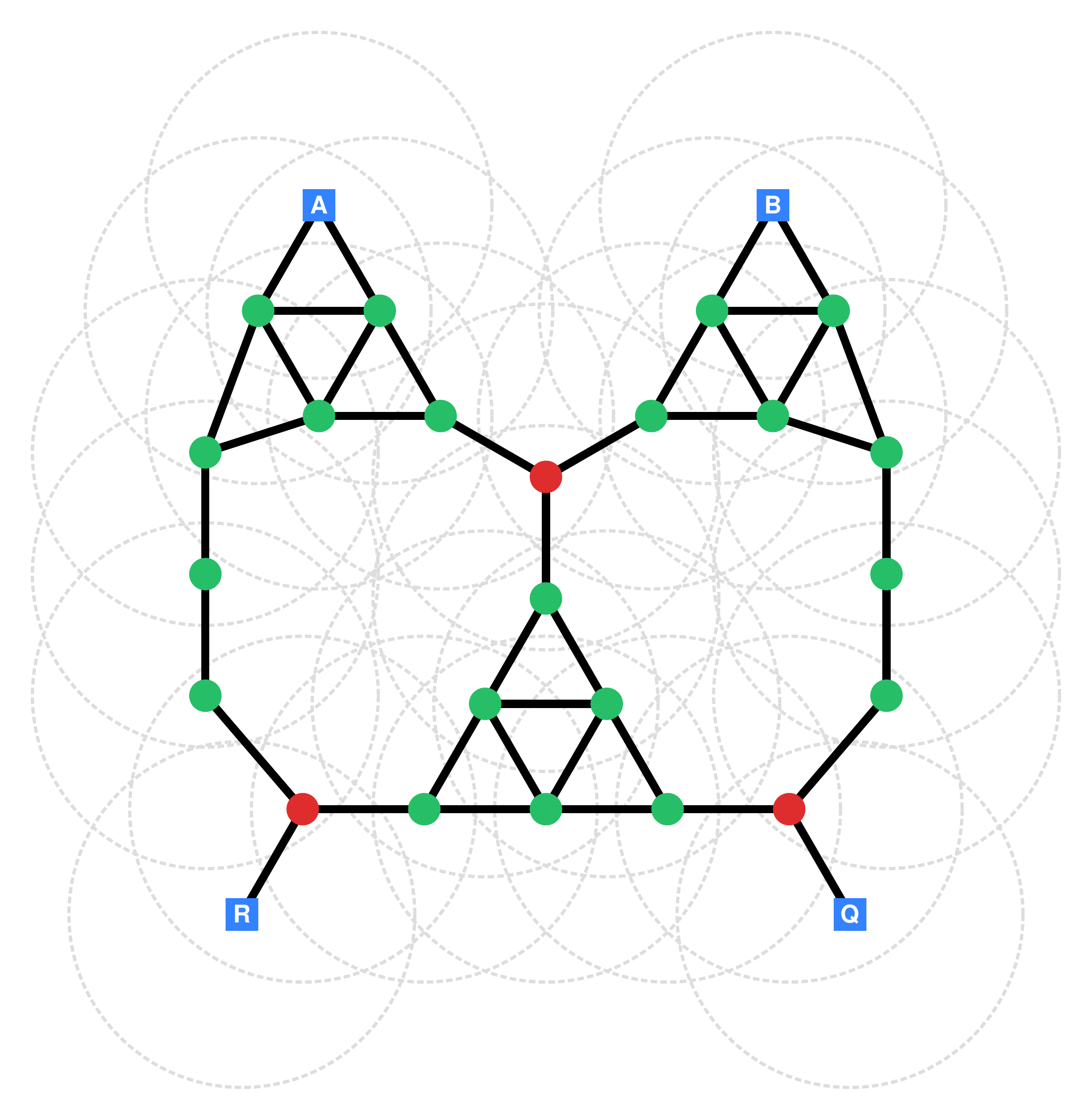}};
        \node[label] at (4.8,4.2) {\lbl{b}};
        \node[figtext,anchor=center] at (7.7,4) {\texttt{CRS}};
        \node[figure] at (9.0,4) {\includegraphics[width=1.5cm]{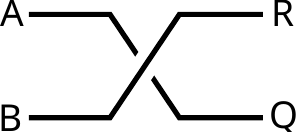}};
        \node[figure] at (6.3,2.6) {\includegraphics[width=4cm]{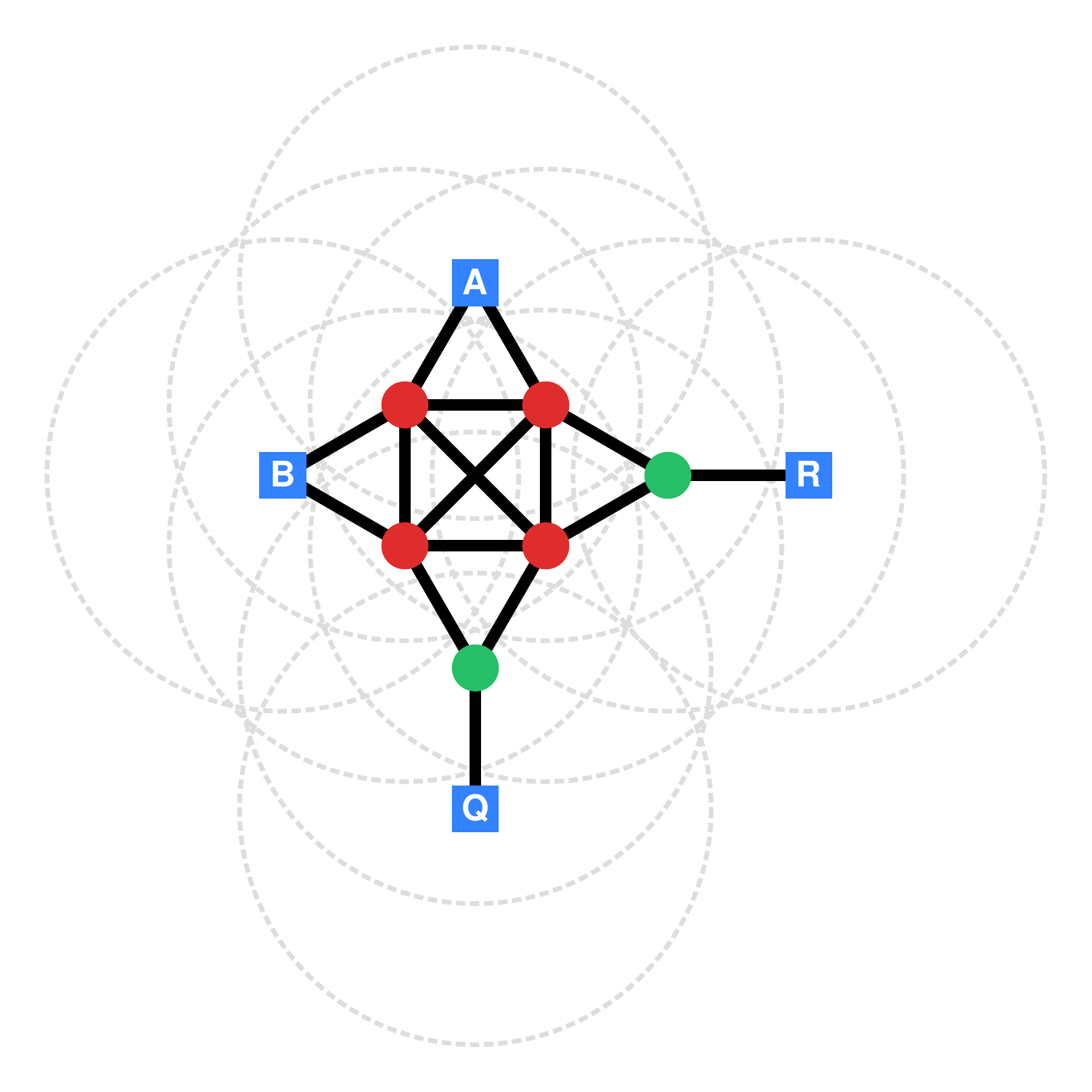}};
        \node[figure] at (7.1,0.2) {\includegraphics[width=5.2cm]{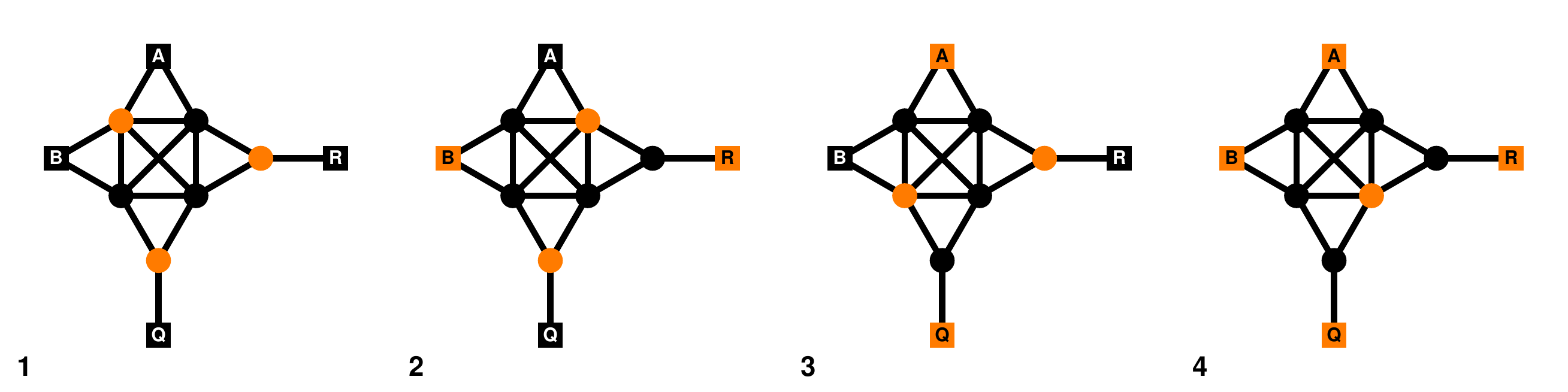}};
        \node[figure] at (9,2.3) {\includegraphics[width=1.2cm]{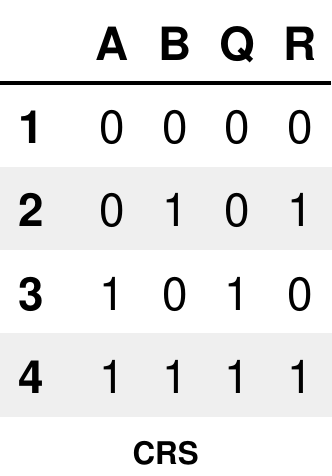}};
        \node[label] at (10.8,4.2) {\lbl{c}};
        \node[figtext,anchor=center] at (13.7,4) {\texttt{ICRS}};
        \node[figure] at (15.0,4) {\includegraphics[width=1.5cm]{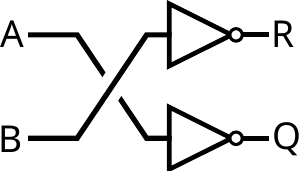}};
        \node[figure] at (12.3,2.4) {\includegraphics[width=3.8cm]{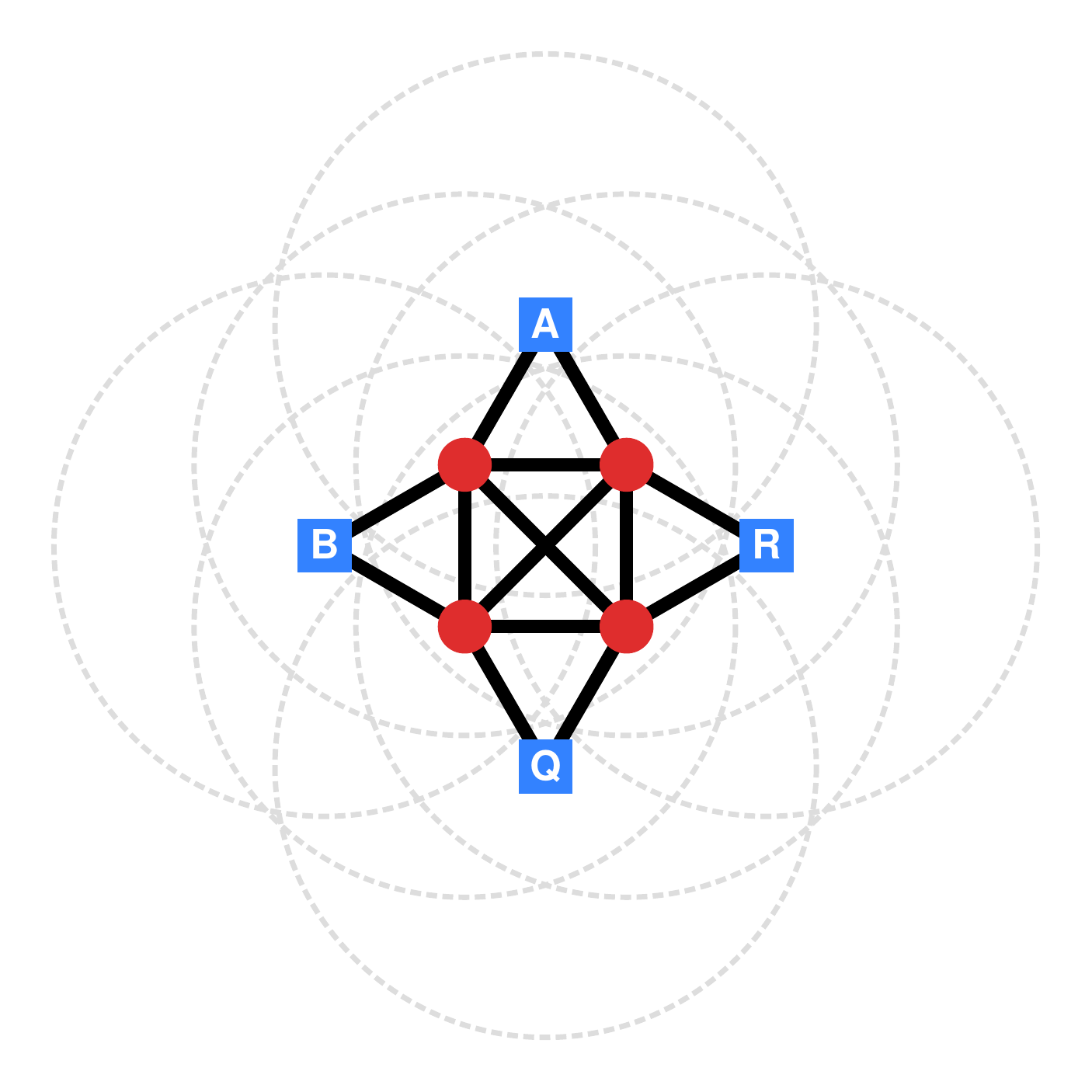}};
        \node[figure] at (13.3,0.2) {\includegraphics[width=5cm]{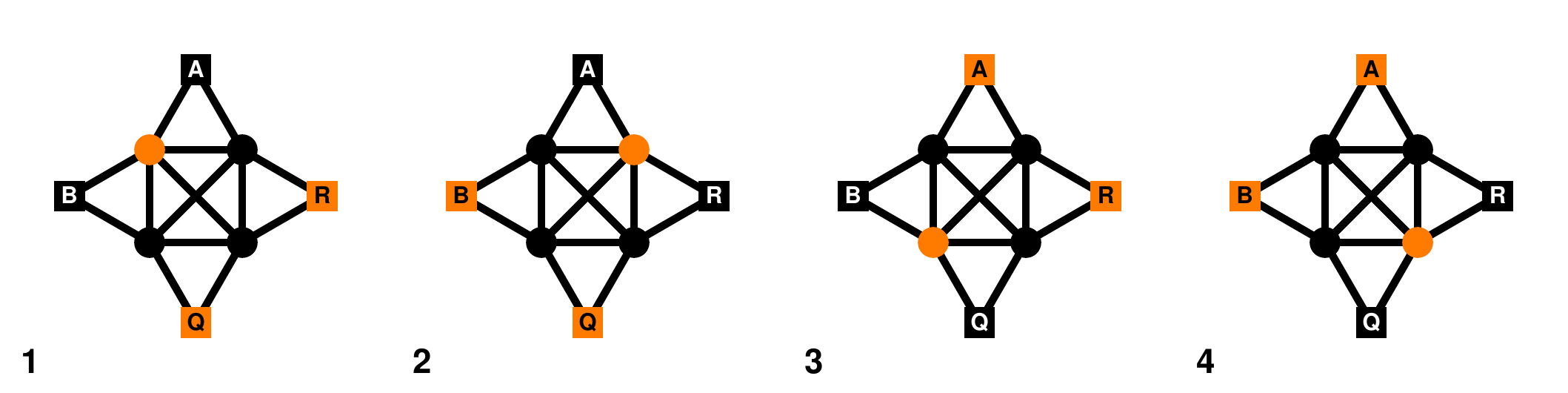}};
        \node[figure] at (15,2.3) {\includegraphics[width=1.2cm]{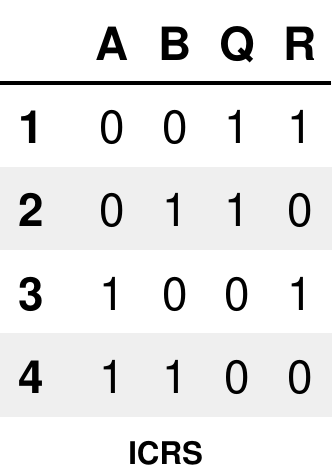}};
        \node[figure] at (2.0,0.0) {\includegraphics[width=1.5cm]{legend3.pdf}};
        %
    \end{tikzpicture}
    \caption{%
	\emph{Crossing.} 
    (a) The crossing constructed from the Boolean circuit crossing
    based on \texttt{XNOR}-gates (see Ref.~\cite{Dewdney_1979} and
    \cref{fig:logic_primitives}); it is an amalgamation of \texttt{LNK}-,
    \texttt{CPY}-, and \texttt{XNOR}-complexes. The ground state manifold
    (not shown) is 4-fold degenerate and ensures $A=Q$ and $B=R$. The complex
    requires $\sim 27$ atoms and is therefore of no practical relevance.
    (b) By contrast, the minimal crossing $\C_\texttt{CRS}$ requires only
    10 atoms; it was constructed by systematically excluding functionally
    equivalent complexes with fewer atoms. The shown data is explained in
    the caption of \cref{fig:logic_primitives}.
    (c) The minimal \emph{inverted} crossing $\C_\texttt{ICRS}$ is smaller
    than the non-inverted crossing and requires only eight atoms. To construct
    $\C_\texttt{CRS}$ from $\C_\texttt{ICRS}$, two \texttt{NOT}-complexes
    must be amalgamated to adjacent ports. This is a recurring scheme due
    to the inverting nature of the Rydberg blockade.
    }
    \label{fig:crossing}
\end{figure*}

The crossing complex realizes the somewhat surprising feature of
intersecting information channels in a strictly two-dimensional
setup of strongly interacting information carriers (recall Step~2
in \cref{sec:completeness}). The possibility to realize such a planar
crossing in a circuit with the three primitives \texttt{LNK}, \texttt{CPY}
and \texttt{NOR} was crucial for the proof of \cref{thm:1}. Note that
the \emph{existence} of such a complex followed immediately from the
existence of the three aforementioned complexes and the well-known fact
that Boolean circuits can be made planar~\cite{Dewdney_1979}. However, just
as for the Boolean gates in \cref{sec:primitives}, the \texttt{NOR}-based
implementation of the circuit crossing in Ref.~\cite{Dewdney_1979} is of low
practical value as it requires seven \texttt{NOR}-gates (if we implement
\texttt{NOT}-gates directly, \cref{fig:decomposition}d); even a simpler
crossing based on only three minimal \texttt{XNOR}-gates requires $\sim 27$
atoms, see \cref{fig:crossing}a. Thus we are again tasked with finding a
minimal complex that realizes the same function.

By systematically excluding the existence of crossing complexes for
$N=4,\dots,9$ atoms, we finally find the minimal complex $\C_\texttt{CRS}$
depicted in \cref{fig:crossing}b comprising 10 atoms. The proof for
its minimality is very technical and more complicated than for the logic
primitives because geometric constraints must be taken into account for the
crossing~\cite{Stastny2023}.
The structure with two dangling ports (Q and R) immediately suggests the
\textit{inverted} crossing $\C_\texttt{ICRS}$ in \cref{fig:crossing}c with
eight atoms, i.e., a complex that allows two signals to pass each other
while inverting both at the same time.
The minimality of the inverted crossing complex $\C_\texttt{ICRS}$ with eight
atoms follows as a corollary from the minimality of the non-inverted crossing
$\C_\texttt{CRS}$ with 10 atoms as the latter can be obtained from the former
by amalgamation of two \texttt{NOT}-complexes (thereby adding two atoms). In
line with our comment at the end of the previous \cref{sec:primitives},
the inverted variant of the crossing is smaller than its non-inverted
counterpart. We note that the inverted crossing $\C_\texttt{ICRS}$ has also
been described in Ref.~\cite{Nguyen2022} were it plays an important role in
mapping non-planar optimization problems to planar Rydberg structures.

\section{Examples: Spin Liquid Primitives}
\label{sec:examples}

In this part, we focus on our motivation outlined in the introduction,
namely the implementation of tessellated target Hilbert spaces of systems
that are characterized by local gauge constraints. We discuss two models
exemplarily: the surface code with Abelian $\mathbb{Z}_2$ topological order
and the non-Abelian Fibonacci model. For the surface code, we will be able
to utilize the Boolean primitives discussed in \cref{sec:primitives}; by
contrast, for the Fibonacci model such a reduction will not be useful.

\subsection{Surface code}
\label{subsec:toric}

The \emph{toric code}~\cite{Kitaev2003} is the prime example for a spin
liquid in two dimensions with long-range entangled ground states that do
not break any symmetries but instead feature \emph{topological order}. The
toric code is referred to as \emph{surface code} if realized on surfaces
with boundaries~\cite{Bravyi1998}; we will stick to this name in the
following. The surface code describes a gapped phase with $\mathbb{Z}_2$
topological order that is described by the mechanism of string-net
condensation~\cite{Levin2005}. It allows for localized excitations that
are Abelian anyons~\cite{Kitaev2006} which, in turn, leads to ground state
degeneracies on topologically non-trivial surfaces (including flat surfaces
with non-trivial boundaries). As a consequence, surface codes are promising
candidates for quantum memories that encode logical qubits reliably into
delocalized degrees of freedom~\cite{Dennis2002}. This makes the implementation
of systems with this kind of topological order interesting both from an
academic and an applied perspective~\cite{Barends2014,Kelly2015,Semeghini2021}.

Here we consider the surface code on a finite square lattice with ``rough''
boundaries (like the gray background lattice in \cref{fig:toric}d); ``rough''
boundaries are terminated by dangling edges that attach to quadrivalent
vertices. The Hamiltonian
\begin{equation}
    H=-J_A\sum_{\text{Sites}\,s}A_s-J_B\sum_{\text{Faces}\,p}B_p    
    \label{eq:toric_H}
\end{equation}
operates on qubits that live on the edges $e$ of the square lattice. The
operators
\begin{equation}
    A_s=\prod_{e\in s}\sigma^z_e
    \quad\text{and}\quad
    B_p=\prod_{e\in p}\sigma^x_e
    \label{eq:toric_AB}
\end{equation}
are referred to as \emph{star} and \emph{plaquette} operators, respectively.
Here, $e\in s$ denotes edges that emanate from site $s$ and $e\in p$
denotes sites that bound face $p$; $\sigma_e^\alpha$ are Pauli matrices for
$\alpha=x,y,z$ acting on the qubit on edge $e$. Since $\com{A_s}{B_p}=0$,
the Hamiltonian \eref{eq:toric_H} is frustration-free and its ground state
$\ket{G}$ is characterized by $A_s\ket{G}=B_p\ket{G}=\ket{G}$ for all sites
$s$ and faces $p$ (assuming $J_A,J_B>0$). Due to the uniform ``rough''
boundaries there is no ground state degeneracy and $\ket{G}$ is unique.

The construction of $\ket{G}$ is straightforward: To satisfy the
constraint $A_s\ket{G}=\ket{G}$ on sites $s$, one can choose the product
state $\ket{\vec{0}}$ with $\sigma_e^z\ket{\vec{0}}=\ket{\vec{0}}$ for all
edges. This state does \emph{not} satisfy the constraint $B_p\ket{G}=\ket{G}$
on faces, though. To fix this, one defines the multiplicative group
$\mathcal{B}=\langle\{B_p\,|\,\text{Faces}\,p\}\rangle$ generated by all
plaquette operators (note that $B_p^2=\id$), and constructs the superposition
\begin{align}
    \ket{G}\propto\sum_{C\in\mathcal{B}}C\ket{\vec{0}}\,.
    \label{eq:toric_gs}
\end{align}
The state $\ket{G}$ is invariant under any $B_p$ by construction since
$\mathcal{B}$ is left-invariant under any $B_p$ by definition. Furthermore,
since $\com{A_s}{B_p}=0$, the site-constraint $A_s\ket{G}=\ket{G}$ is still
satisfied. Thus \cref{eq:toric_gs} describes, up to normalization, the unique
ground state of \cref{eq:toric_H}.

The states $\ket{\vec{C}}\equiv C\ket{\vec{0}}$ have a peculiar structure:
each $C$ can be described as a collection of closed loops on the lattice where
the $\sigma_e^x$ of products of $B_p$ operators act (loops that terminate on
dangling edges at the boundary are considered closed); this loop structure
is then imprinted on $\ket{\vec{0}}$ so that $\ket{\vec{C}}$ is a product
state with a loop pattern $\vec{C}$ of flipped qubits $\ket{1}$. The ground
state \cref{eq:toric_gs} is therefore given by the equal-weight superposition
of all closed loop configurations on the square lattice---which makes it an
example of a \emph{string-net condensate}~\cite{Levin2005} with a non-trivial
pattern of long-range entanglement~\cite{Kitaev2006a,Levin2006}.

\begin{figure*}[t]
    \centering
    \begin{tikzpicture}
        \node[label] at (0,5.5) {\lbl{a}};
        \node[figure,anchor=north west] at (-0.4,5.8) {%
            \includegraphics[width=5.5cm]{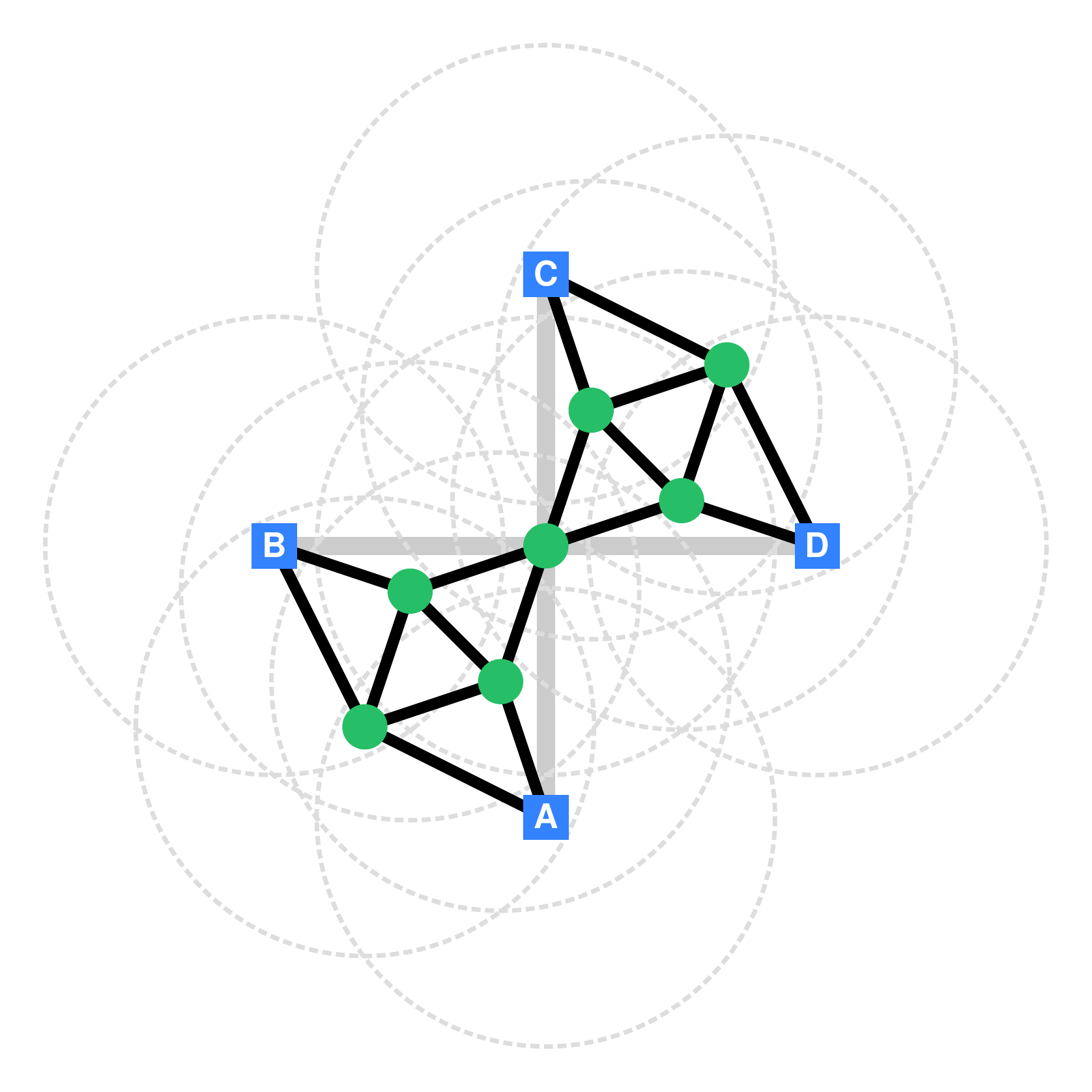}};
        \node[figtext,anchor=center] at (0.8,4.5) {$\C_\texttt{SCU}$};
        \node[label] at (0,0) {\lbl{c}};
        \node[figure,anchor=north west] at (0,0) {%
            \includegraphics[width=15cm]{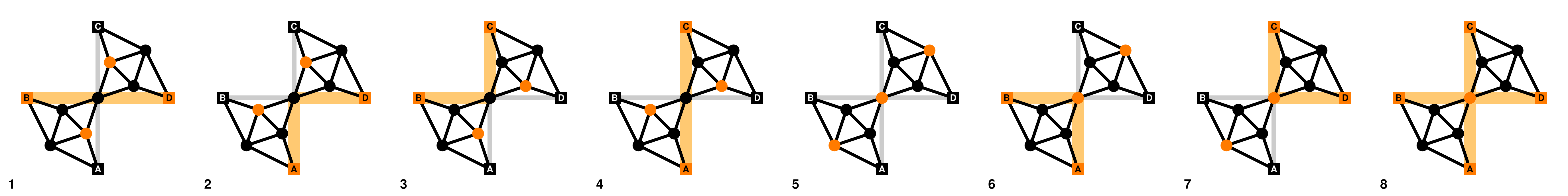}};
        \node[label] at (6.0,5.5) {\lbl{b}};
        \node[figure,anchor=north west] at (6.0,5.2) {%
            \includegraphics[width=1.6cm]{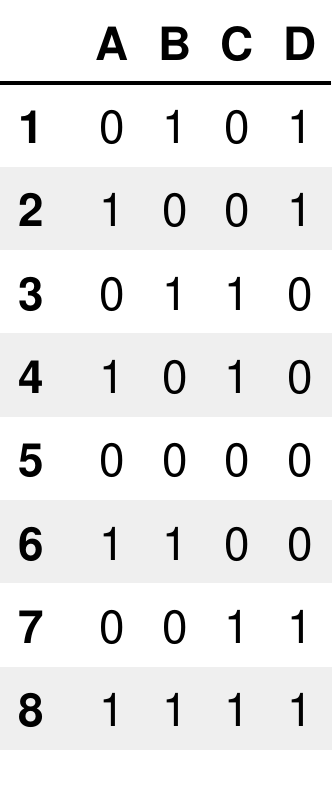}};
        \node[label] at (8.5,5.5) {\lbl{d}};
        \node[figure,anchor=north west] at (8.5,6) {%
            \includegraphics[width=6cm]{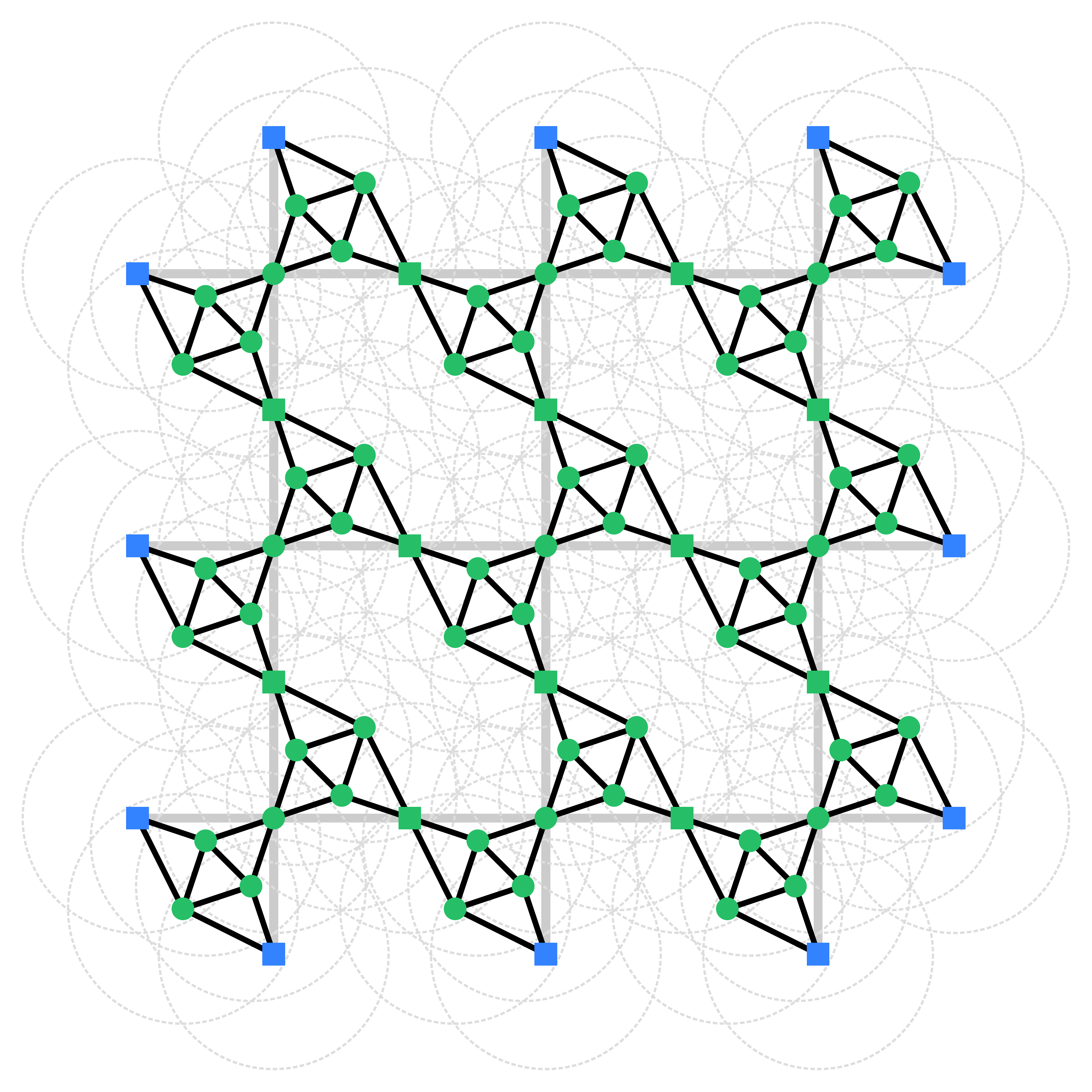}};
        \node[figtext,anchor=center] at (14.0,3.9) {$\C_\sub{\normalfont Loop}$};
        \node[figure,anchor=center] at (4.8,0.8) {\includegraphics[width=1.5cm]{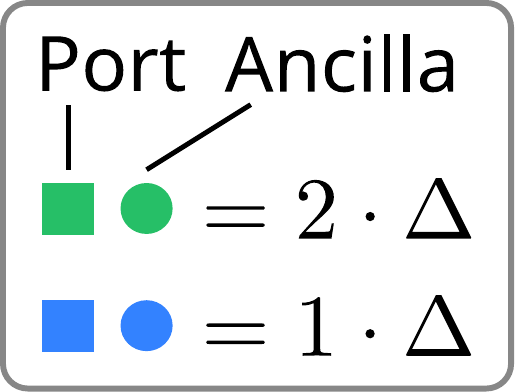}};
        %
    \end{tikzpicture}
    \caption{%
	\emph{Surface code.} 
    (a) Unit cell/vertex complex $\C_\texttt{SCU}$ for the surface code
    ($\mathbb{Z}_2$ topological order). The complex is the amalgamation and
    deformation of two \texttt{XNOR}-complexes [see \cref{fig:logic_primitives}
    and \cref{eq:toric_constraint}] and implements the check function
    constraint $f_\sub{Loop}=1$ defined in \cref{eq:toric_f}. The
    deformations are necessary to prevent an unwanted blockade of ancillas
    in the amalgamation.  Black edges denote blockades between atoms, gray
    edges illustrate the underlying square lattice.
    (b,c) Truth table and ground state manifold of the complex. The manifold
    contains all configurations with an even number of labeled atoms excited,
    thereby realizing Gauss's law on the site (colored edges). This
    provides the local isomorphism between $\H_\sub{T}=\H_\sub{Loop}$
    and $\H_0[\C_\sub{Loop}]$.
    (d) Periodic tessellation $\C_\sub{Loop}$ of the vertex complex
    $\C_\texttt{SCU}$. The copies overlap on the edges and are amalgamated
    at these ports (which makes the detunings uniform in the bulk).
    }
    \label{fig:toric}
\end{figure*}
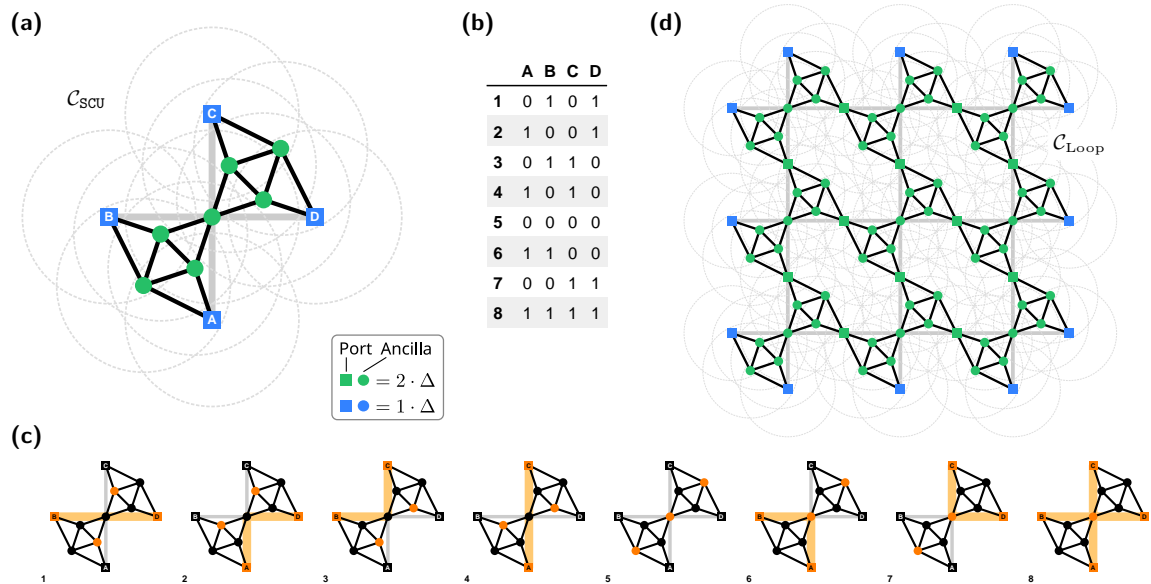

To prepare this state in a real system, one could try to implement the
Hamiltonian \eref{eq:toric_H} and cool the system into its ground state. This
is a challenging task due to the four-body interactions \eref{eq:toric_AB}
which are notoriously hard to realize.
On the Rydberg platform, an alternative and more promising approach goes as
follows: In a first step, one prepares only the subspace
\begin{align}
    \mathcal{H}_\sub{Loop}
    &:=\{\,\ket{\Psi}\,|\,\forall\,{\text{Sites}\;s}:A_s\ket{\Psi}=\ket{\Psi}\,\}
    \nonumber\\
    &\phantom{:}=\spn{\ket{\vec{C}}\,|\,C\in\mathcal{B}}
\end{align}
as the low-energy manifold of a suitably designed structure of
atoms. ($\mathcal{H}_\sub{Loop}$ is the Hilbert space of a $\mathbb{Z}_2$
lattice gauge theory with charge-free background~\cite{Kogut1979}. The local
constraint $A_s\ket{\Psi}=\ket{\Psi}$ corresponds to the gauge symmetry
of this theory and is known as \emph{Gauss's law}.) The $B_p$-terms in
\cref{eq:toric_H} induce quantum fluctuations on this subspace which give
rise to the string-net condensed ground state in \cref{eq:toric_gs}. On the
Rydberg platform, quantum fluctuations can be induced \emph{perturbatively}
by ramping up the Rabi frequency $\Omega_i$. Such fluctuations can give
rise to interesting quantum phases, as shown in Ref.~\cite{Verresen_2021}
for a different model.
This motivates the construction of a Rydberg complex $\C_\sub{Loop}$ with
\begin{align}
    \H_0[\C_\sub{Loop}]\stackrel{\text{loc}}{\simeq} \H_\sub{T}=\H_\sub{Loop}
    =\spn{\ket{\vec{C}}\,|\,C\in\mathcal{B}}\,,
    \label{eq:toric_mapping}
\end{align}
i.e., a Rydberg complex the degenerate ground states of which can
be locally mapped one-to-one to loop configurations on the square
lattice. $\H_0[\C_\sub{Loop}]$ is then a subspace with dimension
$\dim\H_0[\C_\sub{Loop}]\sim 2^M$ where $M$ denotes the number of unit cells
of the square lattice. Note that $\H_0[\C_\sub{Loop}]$ cannot be decomposed
into factors of local Hilbert spaces (like, e.g., the full Hilbert space
$\H=(\mathbb{C}^2)^{\otimes 2M}$ can).

To this end, we assign bits $x_e^1$ to the edges of the square lattice
$\mathcal{L}$ ($K=1$). Our goal is to specify the tessellated ``loop language''
$L_\mathcal{L}[f_\sub{Loop}]$---which contains all bit patterns that trace
out closed loop configurations on the lattice (closed in the sense defined
above)---in terms of a local check function $f_\sub{Loop}$ and a local
bit-projector $\mathfrak{u}_s$ on each site $s$ of the square lattice.
The bit-projector simply selects the four bits on edges adjacent to $s$,
\begin{align}
    \mathfrak{u}_s\left(
    \includegraphics[width=1.7cm,valign=c]{u_square.pdf}
    \right)
    =(x_{e_1}^1,x_{e_2}^1,x_{e_3}^1,x_{e_4}^1)
\end{align}
and the check function reads
\begin{align}
    f_\sub{Loop}(x_1,x_2,x_3,x_4)&=(x_1\odot x_2)\odot (x_3\odot x_4)
    \label{eq:toric_f}
\end{align}
with the \texttt{XNOR}-gate $\odot$ defined in \cref{eq:XNOR}, that is, $A\odot
B=1$ iff $A=B$. It is easy to verify by inspection that $f_\sub{Loop}=1$ if
and only if the number of active bits is \emph{even}, thereby enforcing Gauss's
law on every site of the lattice (because loops cannot terminate there).

We could now construct a complex as discussed in \cref{sec:complexes}, using
the minimal \texttt{XNOR}-complex depicted in \cref{fig:logic_primitives}.
For this construction, we would amalgamate three of these complexes according
to \cref{eq:toric_f} and detune the final output to enforce $f_\sub{Loop}=1$;
this would require at least 16 atoms per site. However, we can do much better
by rewriting the constraint as an equality:
\begin{equation}
    f_\sub{Loop}=1\quad\Leftrightarrow\quad
    x_1\odot x_2 = x_3\odot x_4\,.
    \label{eq:toric_constraint}
\end{equation}
Indeed, \cref{eq:toric_constraint} evaluates to true iff $x_1+x_2+x_3+x_4$
is even.
In general, an implementation of an equality constraint $f_1=f_2$ of two
functions on separate inputs is achieved by amalgamation of their complexes
$\C_{f_1}$ and $\C_{f_2}$ at their output ports, as noted at the end of
\cref{sec:completeness}.
Therefore, the vertex complex $\C_\texttt{SCU}\equiv\C_{f_\sub{Loop}=1}$
(``Surface Code Unit cell'') that realizes the constraint
\cref{eq:toric_constraint} is that of only \emph{two} \texttt{XNOR}-gates
amalgamated at their outputs (\cref{fig:toric}a) which requires only 11
atoms. Surprisingly, it turns out that this realization is also minimal,
see \cref{app:toric} for a proof. (Note that typically the construction
of larger complexes from minimal primitives does \emph{not} yield minimal
complexes.) The two \texttt{XNOR}-complexes that make up the vertex complex
are geometrically deformed variants of the \texttt{XNOR}-complex shown in
\cref{fig:logic_primitives}. This is necessary to prevent unwanted blockades
between ancillas in the amalgamation.

In \cref{fig:toric}b we show the configurations of the four labeled ports (A,
B, C, and D) of the complex in the 8-fold degenerate ground state manifold. In
\cref{fig:toric}c we illustrate the excitation patterns of these eight ground
states (atoms excited to the Rydberg state are colored orange). Highlighting
the edges of the square lattice whenever the labeled ports associated to
them are excited yields the local mapping \eref{eq:toric_mapping} to the
loop structure of states in $\H_\sub{Loop}$. Note that the ancillas do not
add additional degrees of freedom in the ground state manifold.

For the tessellation (\cref{fig:toric}d) the vertex complex is copied and
shifted periodically along the basis vectors of the square lattice. The
labeled ports are then amalgamated to the corresponding ports of complexes
on adjacent sites. Quite remarkably, due to the amalgamation, the detunings
in the bulk become uniform, which makes this tessellation interesting under
the constraints of current platforms~\cite{Semeghini2021,Ebadi2022}. (Note
that imposing periodic boundary conditions on the lattice, i.e., going back
to the \emph{toric} code, would render the detunings completely uniform.)

Let us briefly comment on the modifications of the surface code patch in
\cref{fig:toric}d that would be necessary to use it as a quantum code.  It is
well-known~\cite{Bravyi1998} that a surface code patch encodes a single logical
qubit if its four sides alternate in boundary types: top and bottom remain
``rough'' but left and right are modified to ``smooth'' boundaries by cutting
of the dangling edges of the square lattice. On these boundaries, the sites
become trivalent ``T''-shaped with the same Gauss's law (i.e., the number
of active edges must be even). On these sites, the quadrivalent complex in
\cref{fig:toric}a must be replaced by a trivalent one.  Conveniently enough,
this is just the \texttt{XOR}-complex in \cref{fig:logic_primitives} as the
truth table of \texttt{XOR} contains exactly the four assignments of three
Boolean variables such that $x_1+x_2+x_3$ is even. As a bonus, closing of
the left and right sides of the patch with \texttt{XOR}-complexes leads to
completely uniform detunings along these boundaries. The simplicity of the
vertex complex on trivalent sites suggests a definition of the surface code on
the Honeycomb lattice (which is perfectly possible~\cite{Levin2005}). However,
because of the two sites per unit cell, this does not reduce the number of
required atoms per unit cell to implement the check function. Indeed, the
realizations with minimal Rydberg complexes on both lattices are essentially
equivalent, as can be seen in \cref{fig:toric}d by rotating the tessellation
by 45\degree.

Note that the unit cell complex in \cref{fig:toric}a for the square lattice
allows for tessellations (\cref{fig:toric}d) with 9 atoms per unit cell
(four of the 11 atoms of $\C_\texttt{SCU}$ are shared between pairs of
unit cells) such that the number of atoms for an $L\times L$ lattice is
of order $9L^2$. The number of Rydberg atoms that can be prepared and
controlled in tweezer arrays has recently reached the range of several
hundreds~\cite{Ebadi2021,Scholl2021,Schymik2022}, so that lattices with
$\sim 6\times 6$ sites are already within reach of state-of-the art platforms.

\subsection{Fibonacci model}
\label{subsec:fibonacci}

\begin{figure*}[t]
    \centering
    \begin{tikzpicture}
        \node[figtext] at (0,6) {\lbl{a}};
        \node[figure,anchor=north west] at (0,6) {%
            \includegraphics[width=4.5cm]{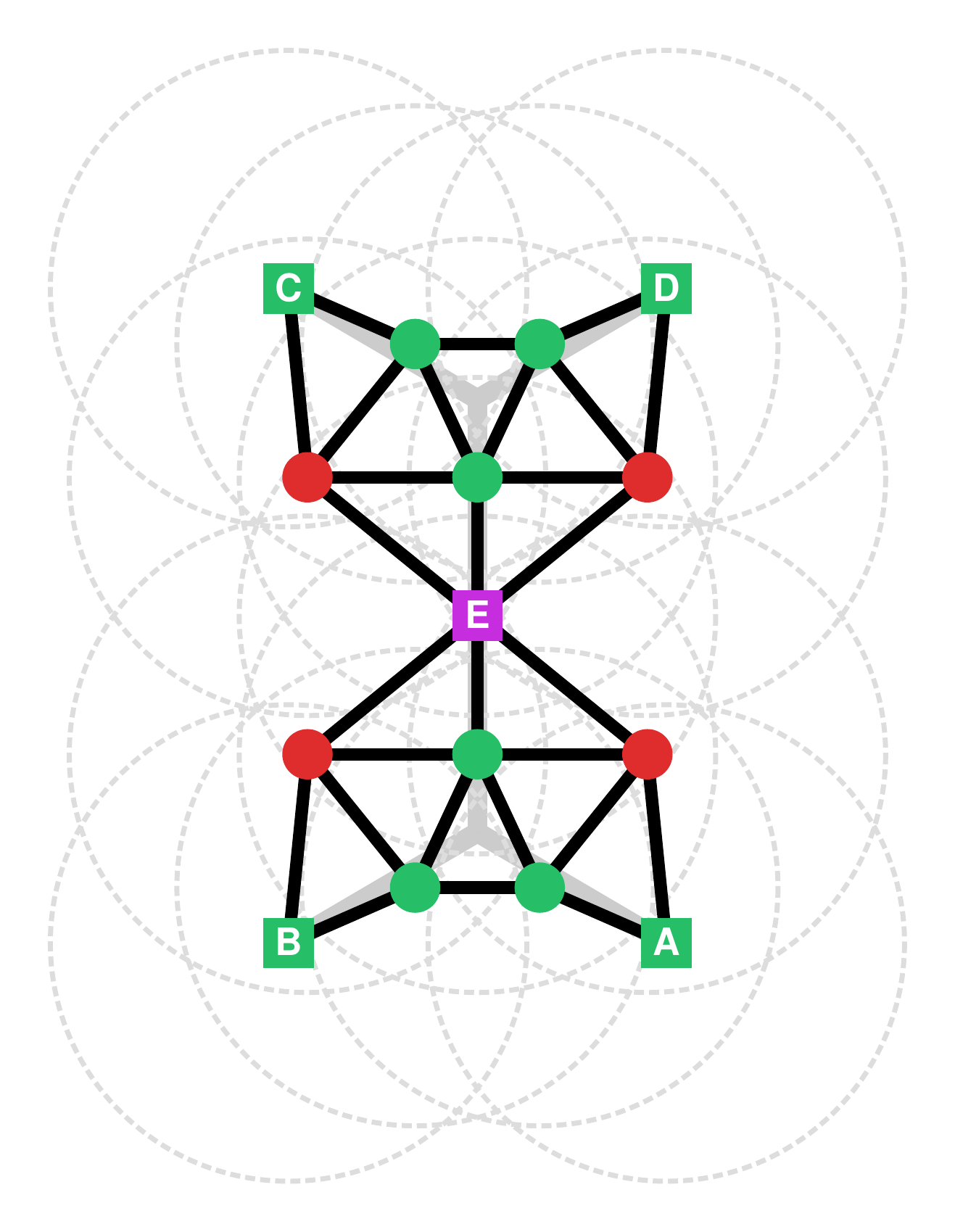}};
        \node[figtext,anchor=center] at (0.5,3.0) {$\C_\texttt{FMU}$};
        \node[figtext,anchor=center] at (2.4,5.2) {$\C_{f_\sub{Fib}=1}$};
        \draw[draw=black,line width=0.8pt,dash pattern=on 2pt off 1.5pt] (1.1,4.9) rectangle (3.6,2.8);
        \node[figtext] at (0,0.2) {\lbl{c}};
        \node[figure,anchor=north west] at (0,0) {%
            \includegraphics[width=15cm]{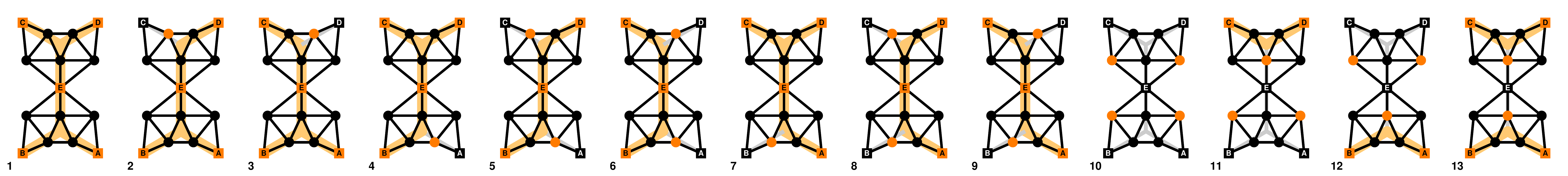}};
        \node[figtext] at (5.0,6) {\lbl{b}};
        \node[figure,anchor=north west] at (5.5,6) {%
            \includegraphics[width=1.9cm]{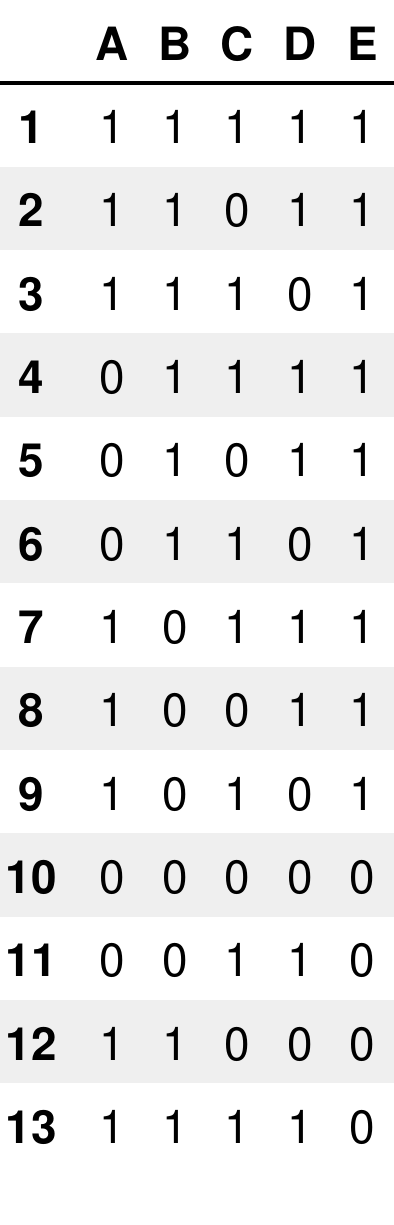}};
        \node[figtext] at (8.0,6) {\lbl{d}};
        \node[figure,anchor=north west] at (8.1,5.7) {%
            \includegraphics[width=7cm]{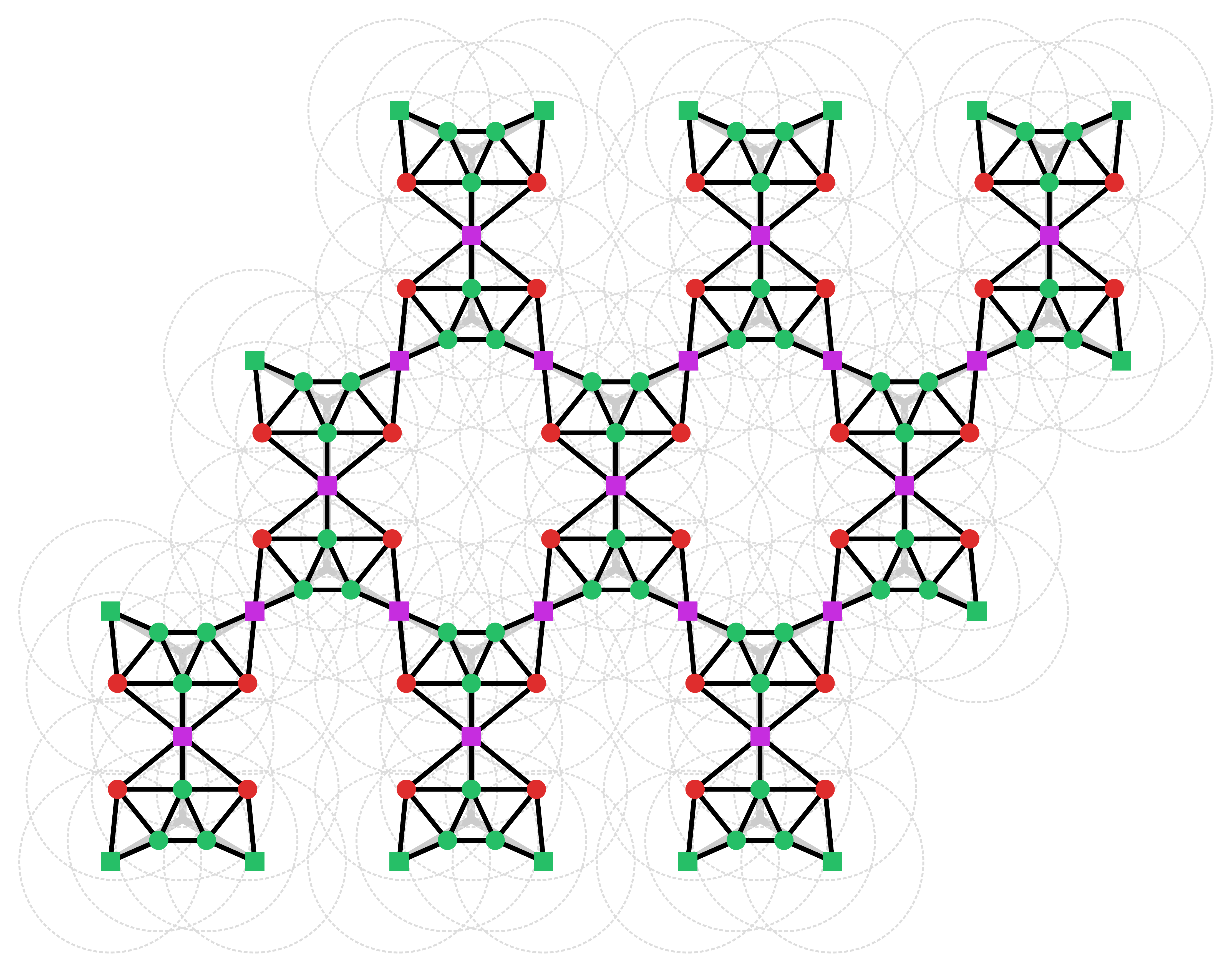}};
        \node[figtext,anchor=center] at (14.0,1.0) {$\C_\sub{Fib}$};
        \node[figure,anchor=center] at (8.9,4.7) {\includegraphics[width=1.5cm]{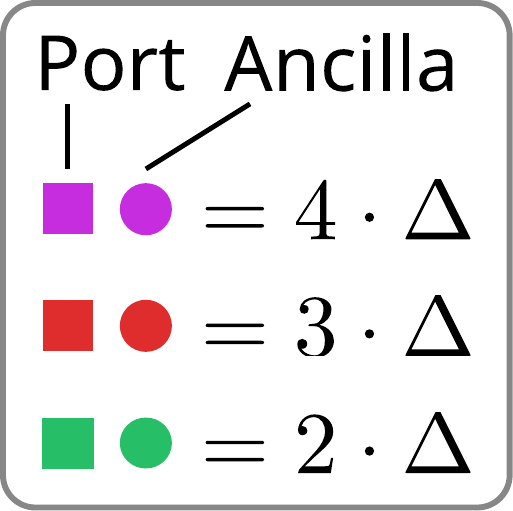}};
        %
    \end{tikzpicture}
    \caption{%
	\emph{Fibonacci model.} 
    (a) Unit cell complex $\C_\texttt{FMU}$ for the Fibonacci model that
    implements two copies of the single-site check function constraint
    $f_\sub{Fib}=1$ defined in \cref{eq:fib_f}. The complex is the
    amalgamation of two equivalent 8-atom complexes $\C_{f_\sub{Fib}=1}$ on
    the two trivalent sites that make up the basis of the honeycomb unit
    cell. Black edges denote blockades between atoms, gray edges illustrate
    the underlying Honeycomb lattice.
    (b,c) Truth table and ground state manifold of the unit cell complex. The
    manifold contains all configurations with closed strings and, in
    addition, configurations with three strings fusing on a site. This
    provides the local isomorphism between the string-net Hilbert space
    $\H_\sub{T}$ and $\H_0[\C_\sub{Fib}]$.
    (d) Periodic tessellation $\C_\sub{Fib}$ of the complex
    $\C_\texttt{FMU}$. The copies overlap on the edges and are amalgamated
    at the corresponding ports.
    }
    \label{fig:fibonacci}
\end{figure*}

The surface code only supports Abelian anyons, which are not sufficient for
universal \emph{topological quantum computation}, where gates are implemented
fault tolerantly by braiding of localized excitations and measurements
correspond to their fusion~\cite{Freedman2002,Nayak2008,Wang2010}.
The simplest anyon model that supports universal computation by braiding is
known as \emph{Fibonacci model} due to the role the Fibonacci numbers play in
the fusion rules~\cite{Freedman2002a,Preskill2004,Bonesteel2005}; it may be
realized in some fractional quantum Hall states~\cite{Read1999,Xia2004}. As
quasiparticles, the properties of Fibonacci anyons are a consequence of
and encoded in the entanglement pattern of the ground state on which they
live. The latter turns out to have a representation as a string-net condensate
with weights and ``string-net'' patterns that differ from the surface code
[cf.~\cref{eq:toric_gs}]. If we consider a Honeycomb lattice with qubits
on its edges, the fixed-point ground state of the Fibonacci model has the
form~\cite{Levin2005}
\begin{align}
    \ket{G}=\sum_{\vec S}\Phi(\vec S)\,\ket{\vec S}\,,
    \label{eq:fibonacci_gs}
\end{align}
where the sum goes over all patterns (``string-nets'') $\vec S$ of flipped
qubits $\ket{1}$ on the edges of the Honeycomb lattice where \emph{no single
string} ends on a vertex. That is, in contrast to the loop patterns $\vec{C}$
of the surface code, vertices with \emph{three} fusing strings are allowed. The
coefficients $\Phi(\vec S)$ of the superposition are non-trivial functions
of the pattern $\vec S$, so that the condensate is no longer an equal-weight
superposition~\cite{Levin2005,Fidkowski2009,Fendley2008}. It is possible to
write down a solvable, local Hamiltonian like \cref{eq:toric_H} with the exact
ground state \eref{eq:fibonacci_gs} which is, however, so complicated that it
is essentially useless for implementations~\cite{Levin2005}. This complication,
together with the potential usefulness of the model for quantum computation,
motivates again the construction of a Rydberg complex $\C_\sub{Fib}$ that
implements the tessellated target Hilbert space
\begin{align}
    \H_0[\C_\sub{Fib}]\stackrel{\text{loc}}{\simeq}\H_\sub{T}
    =\spn{\ket{\vec S}\,|\,\text{String-net}\,\vec S}
    \label{eq:fibonacci_mapping}
\end{align}
which has the dimension $\dim\H_0[\C_\sub{Fib}]\sim
(1+\varphi^2)^{M}+(1+\varphi^{-2})^{M}$ where $M$ is the number of
unit cells of the Honeycomb lattice and $\varphi=(1+\sqrt{5})/2$ is
the golden ratio~\cite{Simon2013,Schulz2013}. As for the surface code,
$\H_0[\C_\sub{Fib}]$ is a Hilbert space that cannot be decomposed into
factors of local Hilbert spaces.

Since the Honeycomb sites are trivalent, the bit-projector takes now the form
\begin{align}
    \mathfrak{u}_s\left(
    \includegraphics[width=1.5cm,valign=c]{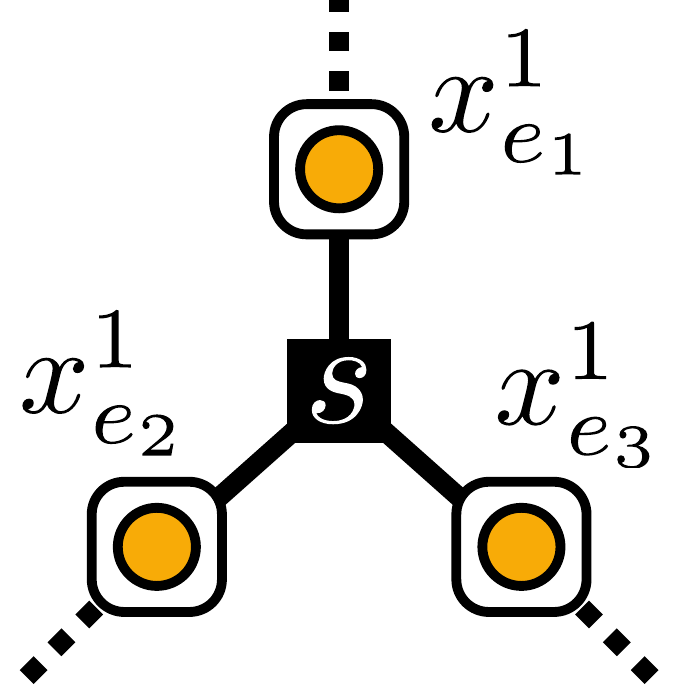}
    \right)
    =(x_{e_1}^1,x_{e_2}^1,x_{e_3}^1)
\end{align}
and the check function that specifies the allowed string-nets can be written
in the compact form
\begin{align}
    f_\sub{Fib}(x_1,x_2,x_3)=
    (x_1\oplus x_2\equiv x_3)
    \vee
    (x_1\wedge x_2 \wedge x_3)
    \label{eq:fib_f}
\end{align}
where the first clause $(x_1\oplus x_2\equiv x_3)$ realizes the loop constraint
($\equiv$ denotes the \emph{logical equivalence} which is equivalent to
the \texttt{XNOR}-gate as a connective) and the second clause $(x_1\wedge
x_2 \wedge x_3)$ allows for the fusion of three strings. Note that without
the second clause we fall back to the loop constraint of the surface code
(now on the honeycomb lattice).

Since there are \emph{five} assignments with $f_\sub{Fib}=1$, this check
function cannot be realized by a single logic gate (despite having three
ports) but must be decomposed into a circuit. Furthermore, since the
amalgamation of two logic gates always results in a complex with an even
number of ports, at least three gates would be necessary to realize the
Fibonacci constraint. This already leads into the territory of $\gtrsim 15$
atoms which we deem too much overhead for a single site. Therefore we follow
the same approach as for the logic primitives in \cref{sec:primitives}:
We systematically exclude the existence of complexes $\C_{f_\sub{Fib}=1}$
for $N=3,4,\dots,7$ atoms (\cref{app:fibonacci}). The approach fails for
$N=8$ and we find the minimal complex in \cref{fig:fibonacci}a (dashed
box). The amalgamation of two of the complexes, one mirrored horizontally,
yields the complex $\C_\texttt{FMU}$ (``Fibonacci Model Unit cell'')
for the two-site unit cell of the Honeycomb lattice, which can then be
tessellated as shown in \cref{fig:fibonacci}d. In contrast to the surface
code, the detunings are not uniform in this case. The full ground state
manifold of the unit cell is shown in \cref{fig:fibonacci}b. The colored
edges in \cref{fig:fibonacci}c for each ground state configuration establish
the local mapping in \cref{eq:fibonacci_mapping}. Note how all string-net
configurations are allowed except for single strings terminating at a site.

Note that the complex for the hexagonal unit cell with 15 atoms in
\cref{fig:fibonacci}a can be interpreted as the complex on a tilted
\emph{square} lattice (by virtually contracting the vertical edges of the
honeycomb lattice). This complex, however, is not minimal as we know of a
12 atom complex that realizes the Fibonacci check function constraint on
quadrivalent sites.

We conclude this section with a comment on the detunings $4\Delta$ of the
ports of the tessellated complex (\cref{fig:fibonacci}d). Note that this is
the first (and only) complex studied in this paper with detunings exceeding
the range $\{1\Delta,2\Delta,3\Delta\}$. This seems to be in tension with the
claim at the end of \cref{sec:completeness} (according to which this range
is sufficient to implement any Boolean constraint). The solution is simple:
Instead of amalgamating the single-site complexes $\C_{f_\sub{Fib}=1}$
directly (which yields the $4\Delta$ on the identified ports), one can use a
single 3-atom \texttt{LNK}-complex to establish the connection between two
ports. On each edge of the lattice, a single atom with detuning $4\Delta$
is then replaced by three atoms with detunings $3\Delta$, $2\Delta$, and
$3\Delta$, respectively. This modification spoils of course the minimality of
the complex, but the corollary in \cref{sec:completeness} did not come with
an assertion for minimality. This observation suggest a tradeoff between
minimizing the \emph{number} of atoms and minimizing the \emph{range}
of required detunings; a potentially interesting and useful direction for
future research (see also \cref{sec:outlook}).

\section{Geometric Optimization}
\label{sec:optimization}

So far we optimized complexes only in terms of their \emph{size} (number of
atoms) for a given language. As a result, we ended up with minimal complexes
that are defined by their blockade graph $B$, local detunings $\{\Delta_i\}$,
and an assignment of ports $\ell$, i.e., atoms that realize the desired
language in the ground state manifold.
Remember that in a blockade graph $B=(V,E)$ an edge $e=(i,j)\in E$ between
atoms $i,j\in V$ indicates that they are in blockade, i.e., cannot be excited
simultaneously. An abstract graph that can be realized in this way by placing
atoms in the plane which are in blockade if and only if their distance is
smaller than some blockade radius $\Rb$ is called a \emph{unit disk graph},
and a geometry $G_\C$ that realizes a prescribed graph as its blockade graph
is a \emph{unit disk embedding} of this graph.
So far, the actual geometry $G_\C$ of our minimal complexes was only taken
into account insofar as a unit disk embedding of the required blockade graph
$B$ must \emph{exist}. (Note that there are graphs that cannot be realized as
blockade graphs of planar geometries, so that this ``geometric realizability''
is a non-trivial condition; deciding whether a given graph can be realized
in this way is unfortunately \textsf{NP}-hard~\cite{Breu1998}.)

Whenever there exists a planar geometry $G_\C=(\vec r_i)_{i\in
V}\in\mathbb{R}^{2N}\equiv\mathfrak{C}_N$ that realizes a prescribed blockade
graph, there typically exist many such geometries: In most cases, there is a
bit of ``wiggle room'' around a given geometry without changing the blockade
graph. In addition, there can be geometrically distinct realizations of the
same blockage graph that cannot be continuously deformed into each other
without violating the blockade constraints. For example:
\begin{center}
    \includegraphics[width=0.95\linewidth]{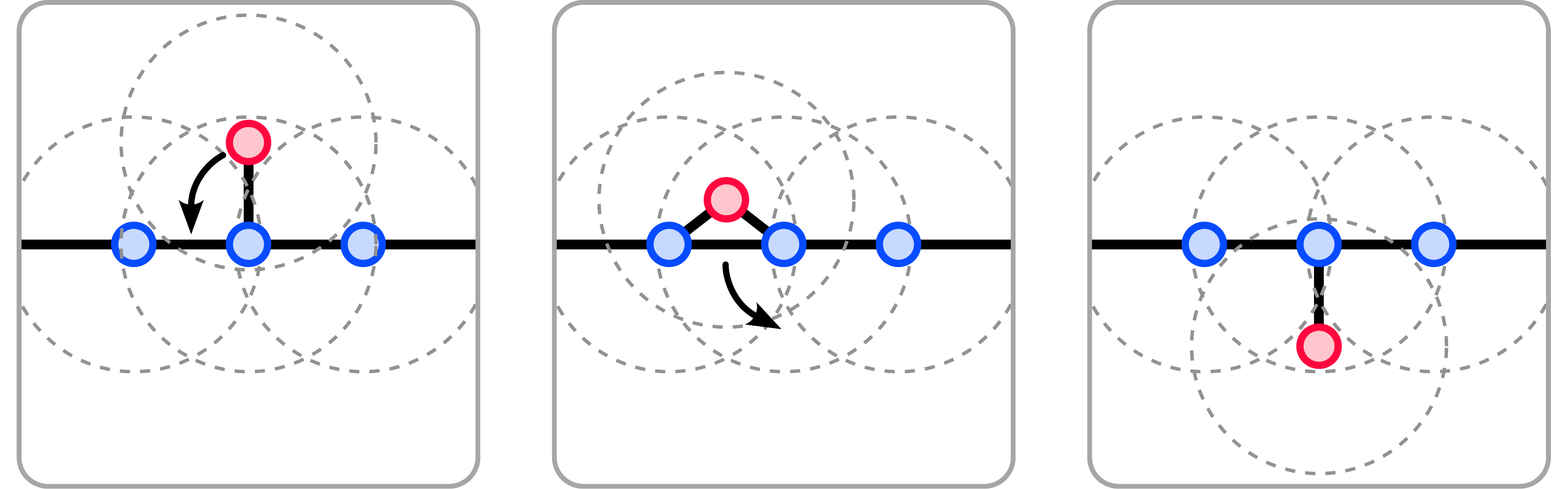}
\end{center}
This can lead to disconnected regions in the configuration space
$\mathfrak{C}_N$ that realize a given blockade graph.

To optimize the geometry of a complex in $\mathfrak{C}_N$, we have to quantify
what we mean by a ``good'' complex. To this end, we define an \emph{objective
function} $\Gamma\,:\,\mathfrak{C}_N\rightarrow\mathbb R$ that quantifies
the quality of the complex and that we seek to minimize. One example is
\begin{align}
    \tilde\Gamma(G_\C)=\frac{\delta E}{\Delta E}
    \label{eq:gamma1}
\end{align}
where $\delta E$ and $\Delta E$ are the width of the ground state manifold
and the gap (recall \cref{fig:rationale}). The problem with \cref{eq:gamma1}
is that its evaluation scales exponentially with the number of atoms $N$
because the computation of $\delta E$ and $\Delta E$ in principle requires
access to the complete spectrum of \cref{eq:H} (which is in general an
\texttt{NP}-hard problem~\cite{Pichler2018a}). While this is feasible for
small complexes, it becomes quickly a bottleneck as $\tilde\Gamma$ must be
evaluated repeatedly when iteratively optimizing a geometry. Furthermore,
in the PXP approximation, interaction energies are either infinite or zero so
that $\tilde\Gamma$ vanishes whenever the blockade constraints are satisfied.
Thus we need a simpler, heuristic quantity that can be directly computed
from the geometry of the complex.

\subsection{Geometric robustness}
\label{subsec:optimization_preliminaries}

To motivate the quantity we propose as objective function below, we first
have to review the role of the blockade radius $\Rb$ in the PXP model.
In the limit of vanishing driving, the blockade radius $\Rb$ is the distance
from an atom where the van der Waals interaction matches its detuning:
$C_6\,\Rb_i^{-6}\stackrel{!}{=}\Delta _i$. As the detunings can vary from
atom to atom in a generic structure $\C$, so does the blockade radius
$\Rb_i$ (this dependence is quite weak, though). However, as outlined in
\cref{sec:setting}, we would like to work in the approximate framework of the
PXP model with a \emph{unique} blockade radius $\Rb$, because then the effects
of interactions between atoms simplify to kinematic constraints encoded in
a blockade graph. In the following, we interpret a given blockade graph $B$
as the encoding of the constraints we would like to realize with a structure
$\C$ of yet unknown geometry $G_\C$.

We can now introduce two dimensionless quantities. First, the \emph{robustness}
of a structure with respect to a given blockade graph $B=(V,E)$ is defined as
\begin{align}
    \xi_B(\C):=\frac{\min\limits_{(i,j)\notin E}d(\vec r_i,\vec r_j)-\max\limits_{(i,j)\in E}d(\vec r_i,\vec r_j)}%
    {\min\limits_{(i,j)\notin E}d(\vec r_i,\vec r_j)+\max\limits_{(i,j)\in E}d(\vec r_i,\vec r_j)}\,,
    \label{eq:robustness}
\end{align}
where $d(\vec r_i,\vec r_j)$ denotes the Euclidean distance. The robustness
is a scale-invariant, finite number $\xi_B(\C)\in [-1,1]$ where $\xi_B(\C)>0$
indicates a valid unit disk embedding $G_\C$ that realizes the prescribed
blockade graph $B$ for blockade radii in some finite interval. Larger positive
values of $\xi_B(\C)$ indicate more \emph{robust} embeddings with more
``wiggle room'' around the positions without changing the blockade graph, or,
equivalently, a wider range of blockade radii that yield the same blockade
graph. If $\xi_B(\C)<0$, the unit disk graph induced by $G_\C$ does not
match the prescribed blockade graph $B$.

Similarly, the \emph{spread} of a structure $\C$ is defined as
\begin{align}
    s(\C)&:=
    \frac{\max_i \Rb_i-\min_i \Rb_i}%
    {\max_i \Rb_i+\min_i \Rb_i}
    \nonumber\\
    &=\frac{(\max_i \Delta_i)^{1/6}-(\min_i \Delta_i)^{1/6}}%
    {(\max_i \Delta_i)^{1/6}+(\min_i \Delta_i)^{1/6}}\,.
    \label{eq:spread}
\end{align}
The spread $s\in [0,1]$ quantifies the relative variations in blockade
radii of a structure (a system with uniform detuning $\Delta_i\equiv\Delta$
has vanishing spread). Just as \cref{eq:robustness} does not depend on the
length scale, \cref{eq:spread} is independent of the $C_6$ coefficient,
i.e., the strength of the interaction.

We can now take into account the variability of the blockade radius without
abandoning the PXP model as follows. We call a structure $\C$ a \emph{valid
implementation} of a blockade graph $B$ if
\begin{align}
    s(\C) < \xi_B(\C)\,.
    \label{eq:valid}
\end{align}
This condition ensures that the geometry $G_\C$ can be scaled such that all
distances of atoms that should (not) be in blockade according to $B$, are
smaller (larger) than the smallest (largest) blockade radius of the structure
$\C$. As this condition is scale-invariant, we do not have to specify $\Rb$
in the following. Note that all structures presented in this paper are valid
in the sense of \cref{eq:valid}.

\subsection{Numerical optimization}
\label{subsec:optimization_results}

These considerations suggest the robustness $\xi_B$ as a measure for the
quality of geometries. We therefore set $\Gamma=-\xi_B$ to maximize this
quantity by minimizing $\Gamma$. The blockade graph $B$ and the detunings
$\{\Delta_i\}$ are fixed and define the functional properties of the complex;
in particular, the spread $s(\C)$ is constant. Thus we optimize for geometries
that satisfy the validity constraint \eref{eq:valid} with a maximal margin
between robustness and spread.

\begin{figure}[tb]
    \centering
    \begin{tikzpicture}
        \node[figtext] at (-1.25,4.8) {\lbl{a}};
        \node[figtext,anchor=south east] at (1.2,0.7) {%
            \color{optimized}$\xi(\C_\texttt{NOR$\triangle$}^\mathrm{opt})=0.268$};
        \node[figtext,anchor=south east] at (1.2,2.9) {%
            \color{unoptimized}$\xi(\C_\texttt{NOR$\triangle$})=0.088$};
        \node[figtext,anchor=south east] at (1.2,-0.3) {%
            $s(\C_\texttt{NOR$\triangle$}^\mathrm{(opt)})=0.058$};
        \node[figure,anchor=center] at (0,4.2) {%
        \includegraphics[width=1.6cm]{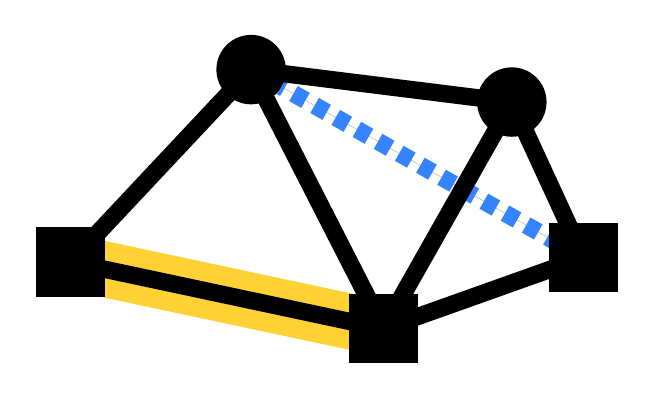}};
        \node[figure,anchor=center] at (0,2.0) {%
        \includegraphics[width=1.6cm]{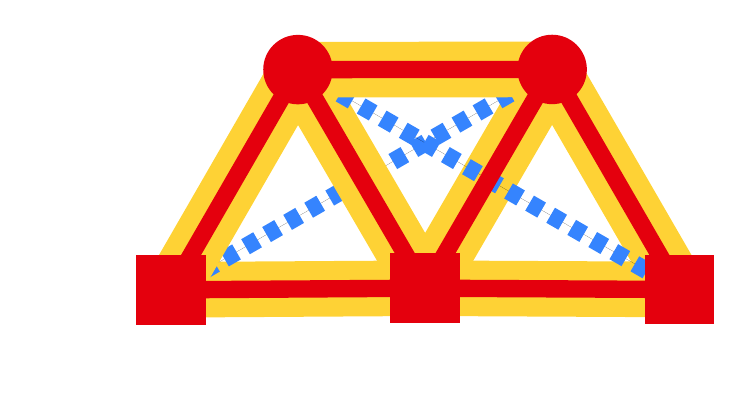}};
        \node[figtext,anchor=south east] at (3.8,0.7) {%
            \color{optimized}$\xi(\C_\texttt{NOR$\circ$}^\mathrm{opt})=0.236$};
        \node[figtext,anchor=south east] at (3.8,2.9) {%
            \color{unoptimized}$\xi(\C_\texttt{NOR$\circ$})=0.111$};
        \node[figtext,anchor=south east] at (3.85,-0.3) {%
            $s(\C_\texttt{NOR$\circ$}^\mathrm{(opt)})=0.058$};
        \node[figure,anchor=center] at (2.6,4.2) {%
        \includegraphics[width=1.6cm]{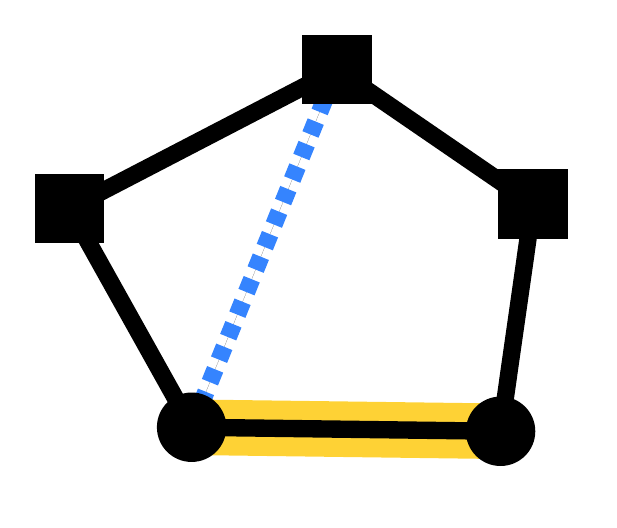}};
        \node[figure,anchor=center] at (2.6,2.0) {%
        \includegraphics[width=1.6cm]{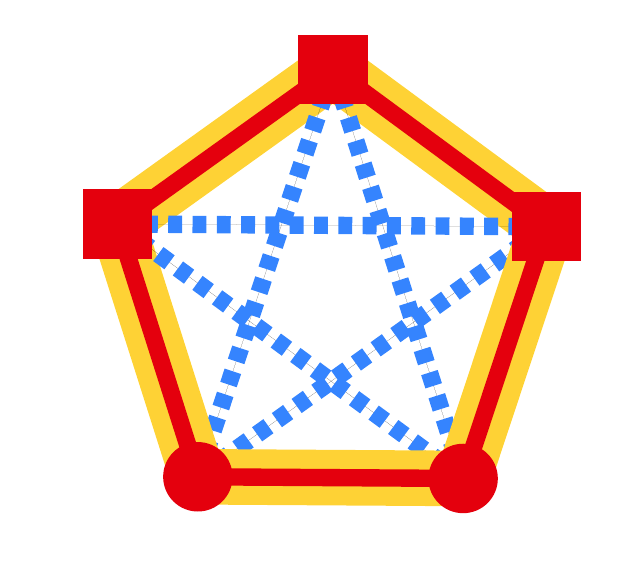}};
        \node[figtext] at (3.8,4.8) {\lbl{b}};
        \node[figtext,anchor=south east] at (6.6,0.7) {%
            \color{optimized}$\xi(\C_\texttt{SCU}^\mathrm{opt})=0.133$};
        \node[figtext,anchor=south east] at (6.6,2.9) {%
            \color{unoptimized}$\xi(\C_\texttt{SCU})=0.117$};
        \node[figtext,anchor=south east] at (6.7,-0.3) {%
            $s(\C_\texttt{SCU}^\mathrm{(opt)})=0.058$};
        \node[figure,anchor=center] at (5.4,4.2) {%
        \includegraphics[width=2.0cm]{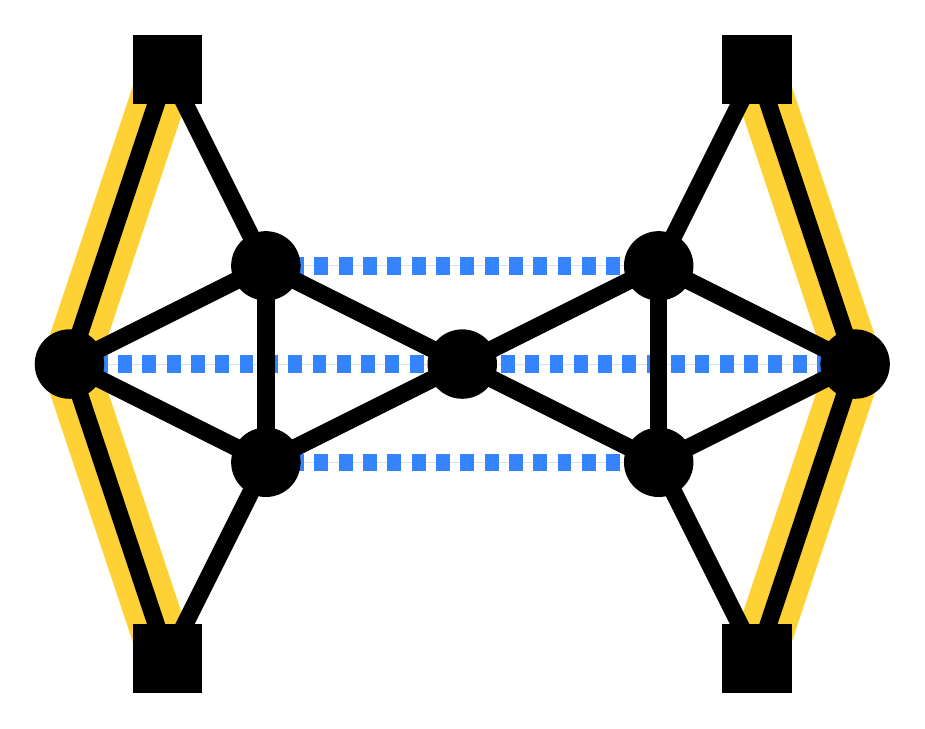}};
        \node[figure,anchor=center] at (5.4,2.0) {%
        \includegraphics[width=2.0cm]{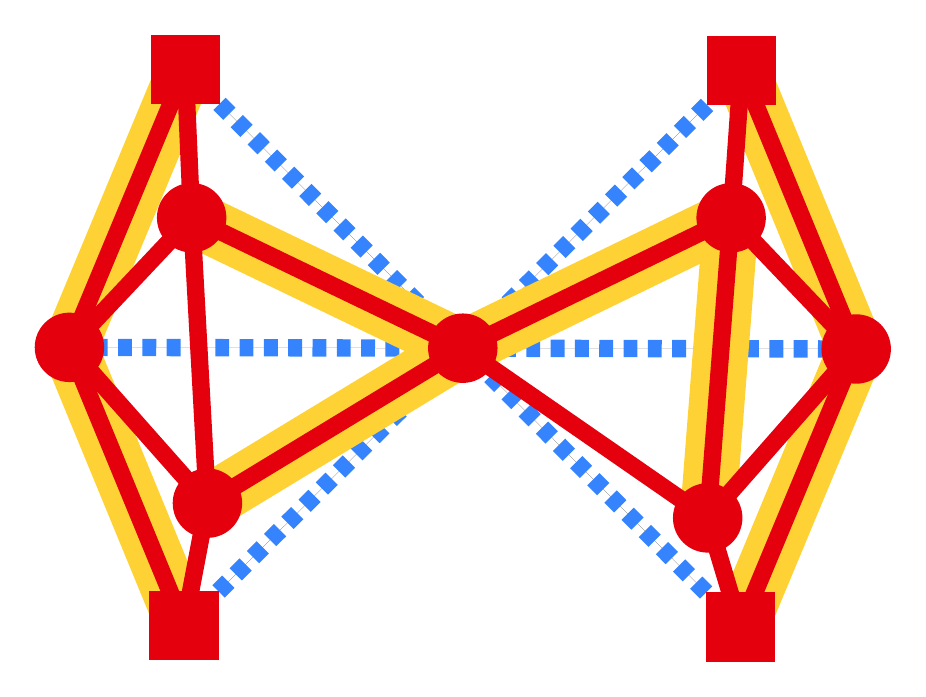}};
        \draw[draw=black,line width=0.5pt,dash pattern=on 1pt off 1.5pt] (-1.2,0.4) -- (6.7,0.4);
        %
    \end{tikzpicture}
    \caption{%
	\emph{Optimization (Examples).} 
    (a) Comparison of perturbed (black) and optimized (red) geometries
    for the two minimal \texttt{NOR}-complexes. Maximum distance
    blockades are highlighted yellow, minimum distances of unblocked
    atoms are indicated by dashed blue edges. The optimal geometries
    are highly symmetric and match the manually constructed ones in
    \cref{fig:logic_primitives} and \cref{fig:pxp_primitives}c. The
    robustness for each complex is printed below the geometries and
    the spread on the bottom of each column (we omit blockade graph
    indices). Note that $\xi(\C_\texttt{NOR$\triangle$}^\mathrm{opt})
    >\xi(\C_\texttt{NOR$\circ$}^\mathrm{opt})$ which makes the triangular
    version \texttt{NOR$\triangle$} potentially more robust than the
    ring-shaped \texttt{NOR$\circ$}. For all geometries the validity constraint
    $s(\C)<\xi(\C)$ is satisfied.
    (b) Comparison of the optimized geometry for the vertex complex
    $\C_\texttt{SCU}$ of the surface code (red) and the manually constructed
    geometry (black) from \cref{fig:toric}a; the robustness increases by
    $\Delta\xi=0.016$. Due to unconstrained atoms, the optimization can
    break the symmetry and produce slightly skewed geometries.
    }
    \label{fig:optimization}
\end{figure}
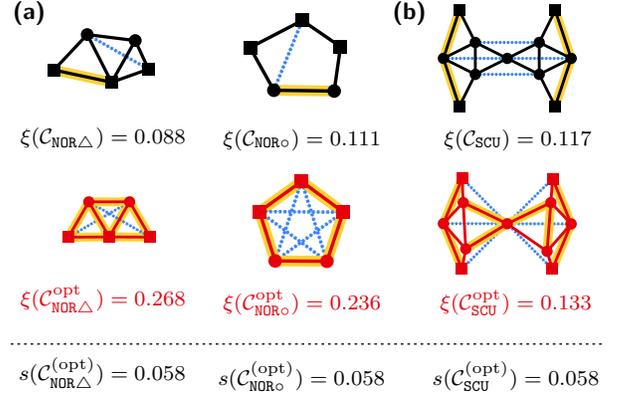

We call a complex $\C$ \emph{globally (locally) optimal} if $\xi_B(\C)>0$ and
its geometry is a global (local) minimum of $\Gamma$ in $\mathfrak{C}_N$. To
minimize $\Gamma$ on the high-dimensional space $\mathfrak{C}_N$, we employ
the \texttt{SciPy} implementation~\cite{Virtanen2020} of generalized simulated
annealing~\cite{Tsallis1996,Andricioaei1996} in combination with a local
optimization based on the Nelder-Mead algorithm~\cite{Nelder1965,Gao2010},
see \cref{app:optimization} for details.
Remember that the robustness is a scale-invariant quantity, so that the
scale of the optimized geometry is arbitrary. For normalization, we rescale
the geometries by setting the blockade radius
\begin{align}
    \Rb:=\frac{1}{2}\left[\max_{(i,j)\in E}d(\vec r_i,\vec r_j)
    +\min_{(i,j)\notin E}d(\vec r_i,\vec r_j)\right]
    \stackrel{!}{=}1\,.
    \label{eq:normalization}
\end{align}

First, we initialized the algorithm with the hand-crafted geometries
of all primitives in \cref{sec:primitives,sec:crossing} and the vertex
complexes in \cref{sec:examples} to optimize their robustness (we believe
the results to be globally optimal but we did not prove this). With these
initial configurations, the optimizer already started with a valid unit disk
embedding of $B$ ($\xi_B>0$) and tried to maximize the robustness further.
The results were typically only slightly deformed versions of the manually
constructed complexes, confirming our intuition. Some of the primitives (in
particular the ring-shaped \texttt{NOR}-complex in \cref{fig:pxp_primitives}c)
were already optimal due of their high symmetry.
In \cref{fig:optimization}a we demonstrate this by comparing slightly perturbed
geometries (black) to the subsequently optimized versions (red) for both
minimal realizations of the \texttt{NOR}-complex. In particular, we find
\begin{align}
    \xi_{B_{\texttt{NOR$\triangle$}}}(\C_\texttt{NOR$\triangle$}^\mathrm{opt})=0.268
    >
    0.236=\xi_{B_{\texttt{NOR$\circ$}}}(\C_\texttt{NOR$\circ$}^\mathrm{opt})
\end{align}
and conclude that the triangular version \texttt{NOR$\triangle$}
(\cref{fig:logic_primitives}) is potentially more robust than the ring-shaped
\texttt{NOR$\circ$} (\cref{fig:pxp_primitives}c). For both, the validity
constraint \eref{eq:valid} is safely satisfied ($x\in\{\circ,\triangle\}$):
\begin{align}
    s(\C_\texttt{NOR$x$}^\mathrm{(opt)})
    =0.058<
    \xi_{B_{\texttt{NOR$x$}}}(\C_\texttt{NOR$x$}^\mathrm{opt})\,.
\end{align}
Since the robustness depends only on the maximum (minimum) distance of
atoms that are (not) in blockade, there can be atoms with positions that are
unconstrained in small regions of the plane. These positions can be chosen
by the optimization algorithm at will, leading to slightly skewed geometries
that break the natural symmetry of the complex; an example is given by the
optimized surface code unit cell complex in \cref{fig:optimization}b. This
is an artifact of our particular objective function that can be eliminated
by more sophisticated choices for $\Gamma$ (e.g.\ motivated by specific
experimental requirements).
All optimized complexes are accessible online~\cite{data}, normalized
according to \cref{eq:normalization}.

In a second run, we went one step further and initialized the optimization with
geometries that \emph{violated} the prescribed blockade graphs (by placing
the atoms randomly). In this case, the algorithm started with $\xi_B<0$
and first had to identify valid unit disk embeddings by stochastic jumps in
the configuration space. These runs typically rediscovered the geometries
we already knew. In some cases, alternative geometries were found (which
turned out to be local maxima of robustness, though). We conclude that it
is not only possible to optimize given geometries but also to \emph{find}
them (if they exist), at least for small complexes.

As a final remark, we stress that geometric optimization is in general
not \emph{reducible}, i.e., optimizing the primitives of a larger circuit
does not necessarily optimize the whole circuit as constraints between
primitives are not taken into account by this approach. This is particularly
important for tessellated complexes of quantum phases like the spin liquids
in \cref{sec:examples}, where one should optimize the complete tessellation
to minimize unwanted residual interactions that are not present in the
optimization of a single-site or unit cell complex.

\section{Outlook \& Comments}
\label{sec:outlook}

We conclude with a few comments on open questions and directions for future
research.

\subparagraph{Minimality.}
To find and prove the minimality of complexes we systematically excluded
realizations with fewer atoms. While this approach is more efficient than a
brute force search (by exploiting constraints from the language, the detunings,
and the planar geometry), it is still far from trivial and cannot be easily
automated. It would be both interesting and useful to develop an algorithm
that, given a uniform language, constructs a \emph{minimal} graph with
weighted nodes, and a labeled node for each letter position of the language,
such that each \emph{maximum-weight independent set}~\cite{gross2013handbook}
is in one-to-one correspondence with a word of the language. We are neither
aware of such an algorithm nor of statements on the complexity to find minimal
solutions. (Note that a solution of this problem might not even be a unit disk
graph, i.e., realizable by the blockade graph of a planar Rydberg complex.)

\subparagraph{Optimization.}
It is clear that our treatment of optimization in \cref{sec:optimization}
only scratches the surface. First, our choice of the objective function
$\Gamma$ is heuristic and other functions may be more appropriate for
specific experimental settings. This would change the ``optimal'' geometries
of complexes, of course. Second, there is a plethora of alternative numerical
algorithms available that could be used to minimize the objective function
more efficiently. In particular the existence of distinct geometries that
are separated by complexes that violate the blockade graph may require more
sophisticated algorithms to escape locally optimal configurations and find
the global optimum. The algorithms also should scale well with the size
of the complex because, as mentioned previously, tessellations should be
optimized as a whole to take into account constraints between its primitives.

If we go one step further and ask for an algorithm that \emph{constructs}
geometries from a given blockade graph, we quickly enter complexity hell:
Deciding whether a given blockade graph can be realized as a unit disk graph is
known to be \textsf{NP}-hard~\cite{Breu1998}. Even if we are \emph{promised}
to be given a unit disk graph as blockade graph, there is \emph{no} efficient
algorithm that outputs the geometry of a complex that realizes it. This
is so because there are unit disk graphs that require exponentially many
bits to specify the positions of the nodes~\cite{McDiarmid2013}. To add
insult to injury, even finding certain \emph{approximations} of unit disk
graph embeddings are known to be \texttt{NP}-hard~\cite{Kuhn2004}. None
of these statements prevent us from looking for heuristic algorithms to
solve these problems for specific cases, of course (as we demonstrated in
\cref{sec:optimization}).

\subparagraph{Uniformity.}
Most of the complexes discussed in this paper make use of atom-specific
detunings (e.g.\ \cref{fig:logic_primitives} and \cref{fig:fibonacci}d). Only
the surface code tessellation in \cref{fig:toric}d is uniform in detunings,
at least in the bulk. While it is possible to realize atom-specific
detunings~\cite{Labuhn_2014,Omran2019}, single-site addressability
adds significant experimental overhead. Thus it is reasonable to
ask whether complexes with non-uniform detunings can be replaced by
(potentially larger) complexes with uniform detunings (without adding
additional degrees of freedom). For instance, there is a third minimal
\texttt{NOR}-complex with uniform detuning $\Delta_i\equiv \Delta$. However,
in amalgamated circuits this uniformity is often destroyed---on the
contrary, it is the \emph{non}-uniformity of the \texttt{XNOR}-complex
(\cref{fig:logic_primitives}) that made the bulk of the surface code uniform
(\cref{fig:toric}d). The quest for uniformity is therefore best formulated
on the level of complete circuits or tessellations.

\subparagraph{Beyond planarity.}
We focused completely on \emph{planar} Rydberg complexes to comply
with the restrictions of current experimental platforms: For the
addressability of single atoms it is simply convenient to have a
dimension of unimpeded access. However, technologically, three-dimensional
structures of Rydberg atoms are possible and have been experimentally
demonstrated~\cite{Barredo2018,Kim2022}. Releasing the planarity
constraint drastically changes the rules for the construction of Rydberg
complexes. For instance, ports that are located \emph{inside} a 2D complex
(and would require expensive crossings to be routed to the perimeter) can
be directly accessed from the third dimension, possibly simplifying certain
functional primitives. Note, however, that at least the logic primitives
in \cref{fig:logic_primitives} do \emph{not} profit from a third dimension.
(This follows from the proofs in \cref{app:pxp}.)

\subparagraph{Beyond the PXP approximation.}
Our construction of Rydberg complexes was based on the assumption that atoms
within the blockade radius can never be simultaneously excited, while atoms
separated by more than the blockade radius do not interact at all; this
``PXP approximation'' implements the dynamical effect of the interactions
as a kinematic constraint. In reality, however, the atoms interact via
the van der Waals interaction $U_\sub{vdW}=C_6\,r^{-6}$ which contributes
also beyond the blockade radius, can lift the degeneracy $\delta E$ of
the ground state manifold, and reduce the gap $\Delta E$ that separates
it from excited states. One therefore expects that complexes with $\delta
E\approx 0$ in the vdW model are geometrically more constrained than in the
PXP model. This has an effect on the geometrical optimization of complexes
(see above) and the appropriate choice of the objective function: To take
into account residual interactions properly, heuristic functions like the
robustness should be replaced by realistic functions like \cref{eq:gamma1},
at least for small complexes where they can be computed exactly.

We checked that the three primitives in \cref{fig:pxp_primitives} can
be realized with perfect degeneracy $\delta E=0$ and gap $\Delta E>0$
in the vdW model by small adjustments of the detunings to balance residual
interactions. In principle, a \texttt{NOR}-complex can even be realized with
only three atoms, arranged in a triangle with precisely defined shape. This
is possible, because the two ancillas in \cref{fig:pxp_primitives}c were
only necessary to balance the energies of states with one and two input
ports active; in the vdW model, the same can be achieved by exploiting
the residual interaction between the two input ports. Which version of the
\texttt{NOR}-complex is more useful for implementations is an open question.

\subparagraph{Quantum phase diagrams.}
In this paper, we only studied the ground state manifold of the Hamiltonian
\eref{eq:H} without quantum fluctuations ($\Omega_i=0$). As has been
demonstrated in Refs.~\cite{Verresen_2021,Samajdar_2021}, the interplay of
quantum fluctuations ($\Omega_i>0$) and the strong blockade interactions
can give rise to interesting many-body quantum phases at zero temperature.
Thus it seems natural to explore the quantum phase diagrams of the proposed
spin-liquid tessellations in \cref{sec:examples}, for example numerically
using density matrix renormalization group (DMRG) techniques. Analytically,
one could derive the effective Hamiltonians on the constructed low-energy
manifolds for finite but small Rabi frequencies $\Omega_i\ll\Delta E$ in
perturbation theory~\cite{Bravyi_2011}. Note that in general one expects
the relative strengths of the effective terms to depend on the specific
complex used to implement the local constraints. This raises the subsequent
question whether these couplings can be \emph{tuned} by modifications of
the used complexes.

\subparagraph{Dynamical preparation.}
In recent experiments~\cite{Semeghini2021}, dynamical preparation
schemes have been used to prepare long-range entangled many-body states
out-of-equilibrium~\cite{Giudici2022,Sahay2022}. The idea is to use
``quasiadiabatic'' protocols $\Omega_i(t)$ and $\Delta_i(t)$ where the
detuning increases continuously to its target value while a finite Rabi
frequency ensures the coupling of different excitation patterns.
This allows for the preparation of non-trivial superpositions of states
in the low-energy subspace of the classical Hamiltonian \eref{eq:H}. It
would be interesting to explore the states of the proposed tessellations
that can be prepared by such dynamical protocols numerically, and
study the effects of defects in the intended logic of the complexes
due to local excitations. Similar questions arise for the primitives
in \cref{sec:primitives,sec:crossing} and circuits built from these by
amalgamation.

\subparagraph{Alternative platforms.}
Rydberg atoms in optical tweezer arrays are the most prominent and advanced
platform with a high level of coherent control that features a blockade
mechanism. Our paper originated in this context and is therefore phrased
in its terminology. It is important to keep in mind, however, that our main
results only require some sort of blockade mechanism, fine-grained control
over the geometric structure of the system, and locally tunable energy shifts
(like chemical potentials or magnetic fields). A natural follow-up question is
then whether there are alternative physical systems with these features. Both
our abstract framework and the introduced complexes could be applied to and
realized by such systems.

\section{Summary}
\label{sec:summary}

In this paper, we developed a framework to design planar structures of
atoms which can be excited into Rydberg states under the constraint of the
Rydberg blockade mechanism (``Rydberg complexes''). Our framework targets
the preparation of degenerate ground state manifolds that are characterized
locally by arbitrary Boolean constraints. We proved that the truth table of
an arbitrary Boolean function can be realized as ground state manifold by
decomposing its circuit representation into three primitives that leverage
the Rydberg blockade. Motivated by this existence claim, we then presented
provably minimal complexes that realize the most important primitives of
Boolean circuits, including a crossing complex that is needed to embed
non-planar circuits into the plane. As an application of our framework, we
constructed periodic Rydberg complexes with degenerate ground state manifolds
that map locally on the non-factorizable string-net Hilbert spaces of the
surface code (with Abelian topological order) and the Fibonacci model (with
non-Abelian topological order). In combination with quantum fluctuations,
these structures may be the starting point to prepare topologically ordered
states in upcoming quantum simulators. We concluded the paper with a discussion
of the geometric optimization of Rydberg complexes using numerical algorithms
to increase their robustness against geometric imperfections and the effects
of long-range van der Waals interactions.

Our results highlight the versatility of planar structures of atoms
that interact via the Rydberg blockade mechanism. We provide a conceptual
foundation for the rationales of \emph{geometric programming}, the encoding
and solution of problems by tailoring the geometry of atomic systems, and
\emph{synthetic quantum matter}, the goal-driven design of quantum materials
on the atomic level. Due to the noisiness of near-term experimental platforms,
the latter seems particularly promising because quantum phases come with an
inherent robustness against a finite density of excitations. This robustness
is less clear in the geometric programming paradigm were the search for
(near-)optimal solutions can be severely impeded by defects in the prepared
states, especially at scale.


\begin{acknowledgments}
    We thank Sebastian Weber for comments on the manuscript.
    This project has received funding from the French-German collaboration
    for joint projects in Natural, Life and Engineering (NLE) Sciences funded
    by the Deutsche Forschungsgemeinschaft (DFG) and the Agence National de
    la Recherche (ANR, project RYBOTIN).
\end{acknowledgments}


\bibliographystyle{bib/bibstyle.bst}
\bibliography{bib/bibliography}

\begin{thebibliography}{10}
\providecommand{\url}[1]{\texttt{#1}}
\providecommand{\urlprefix}{URL }
\expandafter\ifx\csname urlstyle\endcsname\relax
  \providecommand{\doi}[1]{doi:\discretionary{}{}{}#1}\else
  \providecommand{\doi}{doi:\discretionary{}{}{}\begingroup
  \urlstyle{rm}\Url}\fi
\providecommand{\eprint}[2][]{\url{#2}}

\bibitem{Schlosser2001}
N.~Schlosser, G.~Reymond, I.~Protsenko and P.~Grangier,
\newblock \emph{Sub-poissonian loading of single atoms in a microscopic dipole
  trap},
\newblock Nature \textbf{411}(6841), 1024 (2001),
\newblock \doi{10.1038/35082512}.

\bibitem{Saffman_2010}
M.~Saffman, T.~G. Walker and K.~M{\o}lmer,
\newblock \emph{Quantum information with {Rydberg} atoms},
\newblock Reviews of Modern Physics \textbf{82}(3), 2313 (2010),
\newblock \doi{10.1103/RevModPhys.82.2313}.

\bibitem{Nogrette2014}
F.~Nogrette, H.~Labuhn, S.~Ravets, D.~Barredo, L.~B{\'{e}}guin, A.~Vernier,
  T.~Lahaye and A.~Browaeys,
\newblock \emph{Single-atom trapping in holographic 2d arrays of microtraps
  with arbitrary geometries},
\newblock Physical Review X \textbf{4}(2), 021034 (2014),
\newblock \doi{10.1103/physrevx.4.021034}.

\bibitem{Barredo_2016}
D.~Barredo, S.~de~L{\'{e}}s{\'{e}}leuc, V.~Lienhard, T.~Lahaye and A.~Browaeys,
\newblock \emph{An atom-by-atom assembler of defect-free arbitrary
  two-dimensional atomic arrays},
\newblock Science \textbf{354}(6315), 1021 (2016),
\newblock \doi{10.1126/science.aah3778}.

\bibitem{Barredo2018}
D.~Barredo, V.~Lienhard, S.~de~L{\'{e}}s{\'{e}}leuc, T.~Lahaye and A.~Browaeys,
\newblock \emph{Synthetic three-dimensional atomic structures assembled atom by
  atom},
\newblock Nature \textbf{561}(7721), 79 (2018),
\newblock \doi{10.1038/s41586-018-0450-2}.

\bibitem{Weimer2010}
H.~Weimer, M.~Müller, I.~Lesanovsky, P.~Zoller and H.~P. Büchler,
\newblock \emph{A {Rydberg} quantum simulator},
\newblock Nature Physics \textbf{6}(5), 382 (2010),
\newblock \doi{10.1038/nphys1614}.

\bibitem{Georgescu2014}
I.~Georgescu, S.~Ashhab and F.~Nori,
\newblock \emph{Quantum simulation},
\newblock Reviews of Modern Physics \textbf{86}(1), 153 (2014),
\newblock \doi{10.1103/revmodphys.86.153}.

\bibitem{Gross2017}
C.~Gross and I.~Bloch,
\newblock \emph{Quantum simulations with ultracold atoms in optical lattices},
\newblock Science \textbf{357}(6355), 995 (2017),
\newblock \doi{10.1126/science.aal3837}.

\bibitem{Altman2021}
E.~Altman, K.~R. Brown, G.~Carleo, L.~D. Carr, E.~Demler, C.~Chin, B.~DeMarco,
  S.~E. Economou, M.~A. Eriksson, K.-M.~C. Fu, M.~Greiner, K.~R. Hazzard
  \emph{et~al.},
\newblock \emph{Quantum simulators: Architectures and opportunities},
\newblock {PRX} Quantum \textbf{2}(1), 017003 (2021),
\newblock \doi{10.1103/prxquantum.2.017003}.

\bibitem{Degen2017}
C.~Degen, F.~Reinhard and P.~Cappellaro,
\newblock \emph{Quantum sensing},
\newblock Reviews of Modern Physics \textbf{89}(3), 035002 (2017),
\newblock \doi{10.1103/revmodphys.89.035002}.

\bibitem{Pezze2018}
L.~Pezz{\`{e}}, A.~Smerzi, M.~K. Oberthaler, R.~Schmied and P.~Treutlein,
\newblock \emph{Quantum metrology with nonclassical states of atomic
  ensembles},
\newblock Reviews of Modern Physics \textbf{90}(3), 035005 (2018),
\newblock \doi{10.1103/revmodphys.90.035005}.

\bibitem{Ladd2010}
T.~D. Ladd, F.~Jelezko, R.~Laflamme, Y.~Nakamura, C.~Monroe and J.~L. O'Brien,
\newblock \emph{Quantum computers},
\newblock Nature \textbf{464}(7285), 45 (2010),
\newblock \doi{10.1038/nature08812}.

\bibitem{Henriet2020}
L.~Henriet, L.~Beguin, A.~Signoles, T.~Lahaye, A.~Browaeys, G.-O. Reymond and
  C.~Jurczak,
\newblock \emph{Quantum computing with neutral atoms},
\newblock Quantum \textbf{4}, 327 (2020),
\newblock \doi{10.22331/q-2020-09-21-327}.

\bibitem{Graham2022}
T.~M. Graham, Y.~Song, J.~Scott, C.~Poole, L.~Phuttitarn, K.~Jooya, P.~Eichler,
  X.~Jiang, A.~Marra, B.~Grinkemeyer, M.~Kwon, M.~Ebert \emph{et~al.},
\newblock \emph{Multi-qubit entanglement and algorithms on a neutral-atom
  quantum computer},
\newblock Nature \textbf{604}(7906), 457 (2022),
\newblock \doi{10.1038/s41586-022-04603-6}.

\bibitem{Bluvstein2022}
D.~Bluvstein, H.~Levine, G.~Semeghini, T.~T. Wang, S.~Ebadi, M.~Kalinowski,
  A.~Keesling, N.~Maskara, H.~Pichler, M.~Greiner, V.~Vuleti{\'{c}} and M.~D.
  Lukin,
\newblock \emph{A quantum processor based on coherent transport of entangled
  atom arrays},
\newblock Nature \textbf{604}(7906), 451 (2022),
\newblock \doi{10.1038/s41586-022-04592-6}.

\bibitem{Gallagher2006}
T.~Gallagher,
\newblock \emph{{Rydberg} atoms},
\newblock In \emph{Springer Handbook of Atomic, Molecular, and Optical
  Physics}, pp. 235--245. Springer New York,
\newblock \doi{10.1007/978-0-387-26308-3_14} (2006).

\bibitem{Sibalic2018}
N.~Sibalic and C.~S. Adams,
\newblock \emph{{Rydberg} Physics},
\newblock {IOP} Publishing,
\newblock \doi{10.1088/978-0-7503-1635-4} (2018).

\bibitem{Jaksch2000}
D.~Jaksch, J.~I. Cirac, P.~Zoller, S.~L. Rolston, R.~C{\^{o}}t{\'{e}} and M.~D.
  Lukin,
\newblock \emph{Fast quantum gates for neutral atoms},
\newblock Physical Review Letters \textbf{85}(10), 2208 (2000),
\newblock \doi{10.1103/physrevlett.85.2208}.

\bibitem{Tong2004}
D.~Tong, S.~M. Farooqi, J.~Stanojevic, S.~Krishnan, Y.~P. Zhang,
  R.~C{\^{o}}t{\'{e}}, E.~E. Eyler and P.~L. Gould,
\newblock \emph{Local blockade of {Rydberg} excitation in an ultracold gas},
\newblock Physical Review Letters \textbf{93}(6), 063001 (2004),
\newblock \doi{10.1103/physrevlett.93.063001}.

\bibitem{Singer2004}
K.~Singer, M.~Reetz-Lamour, T.~Amthor, L.~G. Marcassa and M.~Weidemüller,
\newblock \emph{Suppression of excitation and spectral broadening induced by
  interactions in a cold gas of {Rydberg} atoms},
\newblock Physical Review Letters \textbf{93}(16), 163001 (2004),
\newblock \doi{10.1103/physrevlett.93.163001}.

\bibitem{Gaetan2009}
A.~Gaëtan, Y.~Miroshnychenko, T.~Wilk, A.~Chotia, M.~Viteau, D.~Comparat,
  P.~Pillet, A.~Browaeys and P.~Grangier,
\newblock \emph{Observation of collective excitation of two individual atoms in
  the {Rydberg} blockade regime},
\newblock Nature Physics \textbf{5}(2), 115 (2009),
\newblock \doi{10.1038/nphys1183}.

\bibitem{Urban2009}
E.~Urban, T.~A. Johnson, T.~Henage, L.~Isenhower, D.~D. Yavuz, T.~G. Walker and
  M.~Saffman,
\newblock \emph{Observation of {Rydberg} blockade between two~atoms},
\newblock Nature Physics \textbf{5}(2), 110 (2009),
\newblock \doi{10.1038/nphys1178}.

\bibitem{Ebadi2021}
S.~Ebadi, T.~T. Wang, H.~Levine, A.~Keesling, G.~Semeghini, A.~Omran,
  D.~Bluvstein, R.~Samajdar, H.~Pichler, W.~W. Ho, S.~Choi, S.~Sachdev
  \emph{et~al.},
\newblock \emph{Quantum phases of matter on a 256-atom programmable quantum
  simulator},
\newblock Nature \textbf{595}(7866), 227 (2021),
\newblock \doi{10.1038/s41586-021-03582-4}.

\bibitem{Scholl2021}
P.~Scholl, M.~Schuler, H.~J. Williams, A.~A. Eberharter, D.~Barredo, K.-N.
  Schymik, V.~Lienhard, L.-P. Henry, T.~C. Lang, T.~Lahaye, A.~M. Läuchli and
  A.~Browaeys,
\newblock \emph{Quantum simulation of {2D} antiferromagnets with hundreds of
  {Rydberg} atoms},
\newblock Nature \textbf{595}(7866), 233 (2021),
\newblock \doi{10.1038/s41586-021-03585-1}.

\bibitem{Schymik2022}
K.-N. Schymik, B.~Ximenez, E.~Bloch, D.~Dreon, A.~Signoles, F.~Nogrette,
  D.~Barredo, A.~Browaeys and T.~Lahaye,
\newblock \emph{In situ equalization of single-atom loading in large-scale
  optical tweezer arrays},
\newblock Physical Review A \textbf{106}(2), 022611 (2022),
\newblock \doi{10.1103/physreva.106.022611}.

\bibitem{Note1}
NISQ = Noisy Intermediate-Scale Quantum technology, i.e., near-term quantum
  technology without full-fledged quantum error correction, see Ref.~\cite
  {Preskill2018}.

\bibitem{Pichler2018b}
H.~Pichler, S.-T. Wang, L.~Zhou, S.~Choi and M.~D. Lukin,
\newblock \emph{Quantum optimization for maximum independent set using
  {Rydberg} atom arrays}  (2018),
\newblock \doi{10.48550/arxiv.1808.10816}.

\bibitem{Clark1990}
B.~N. Clark, C.~J. Colbourn and D.~S. Johnson,
\newblock \emph{Unit disk graphs},
\newblock Discrete Mathematics \textbf{86}(1-3), 165 (1990),
\newblock \doi{10.1016/0012-365x(90)90358-o}.

\bibitem{Pichler2018a}
H.~Pichler, S.-T. Wang, L.~Zhou, S.~Choi and M.~D. Lukin,
\newblock \emph{Computational complexity of the {Rydberg} blockade in two
  dimensions}  (2018),
\newblock \doi{10.48550/arxiv.1809.04954}.

\bibitem{Serret_2020}
M.~F. Serret, B.~Marchand and T.~Ayral,
\newblock \emph{Solving optimization problems with {Rydberg} analog quantum
  computers: Realistic requirements for quantum advantage using noisy
  simulation and classical benchmarks},
\newblock Physical Review A \textbf{102}(5), 052617 (2020),
\newblock \doi{10.1103/PhysRevA.102.052617}.

\bibitem{Wurtz2022}
J.~Wurtz, P.~L.~S. Lopes, N.~Gemelke, A.~Keesling and S.~Wang,
\newblock \emph{Industry applications of neutral-atom quantum computing solving
  independent set problems}  (2022),
\newblock \doi{10.48550/arxiv.2205.08500}.

\bibitem{Dalyac2022}
C.~Dalyac and L.~Henriet,
\newblock \emph{Embedding the {MIS} problem for non-local graphs with bounded
  degree using {3D} arrays of atoms}  (2022),
\newblock \doi{10.48550/arxiv.2209.05164}.

\bibitem{Nguyen2022}
M.-T. Nguyen, J.-G. Liu, J.~Wurtz, M.~D. Lukin, S.-T. Wang and H.~Pichler,
\newblock \emph{Quantum optimization with arbitrary connectivity using
  {Rydberg} atom arrays},
\newblock PRX Quantum \textbf{4}, 010316 (2023),
\newblock \doi{10.1103/PRXQuantum.4.010316}.

\bibitem{Lanthaler2023}
M.~Lanthaler, C.~Dlaska, K.~Ender and W.~Lechner,
\newblock \emph{Rydberg-blockade-based parity quantum optimization},
\newblock Phys. Rev. Lett. \textbf{130}, 220601 (2023),
\newblock \doi{10.1103/PhysRevLett.130.220601}.

\bibitem{Jeong2023}
S.~Jeong, M.~Kim, M.~Hhan and J.~Ahn,
\newblock \emph{Quantum programming of the satisfiability problem with
  {Rydberg} atom graphs}  (2023),
\newblock \doi{10.48550/arxiv.2302.14369}.

\bibitem{Note2}
Of course one should not expect an exponential speedup by these mappings as it
  is widely believed~\cite {Bennett1997} that $\protect \text {\protect \textsf
  {NP}}\nsubseteq \protect \text {\protect \textsf {BQP}}$.

\bibitem{Byun2022}
A.~Byun, M.~Kim and J.~Ahn,
\newblock \emph{Finding the maximum independent sets of platonic graphs using
  {Rydberg} atoms},
\newblock {PRX} Quantum \textbf{3}(3), 030305 (2022),
\newblock \doi{10.1103/prxquantum.3.030305}.

\bibitem{Kim2022}
M.~Kim, K.~Kim, J.~Hwang, E.-G. Moon and J.~Ahn,
\newblock \emph{{Rydberg} quantum wires for maximum independent set problems},
\newblock Nature Physics \textbf{18}(7), 755 (2022),
\newblock \doi{10.1038/s41567-022-01629-5}.

\bibitem{Ebadi2022}
S.~Ebadi, A.~Keesling, M.~Cain, T.~T. Wang, H.~Levine, D.~Bluvstein,
  G.~Semeghini, A.~Omran, J.-G. Liu, R.~Samajdar, X.-Z. Luo, B.~Nash
  \emph{et~al.},
\newblock \emph{Quantum optimization of maximum independent set using {Rydberg}
  atom arrays},
\newblock Science \textbf{376}(6598), 1209 (2022),
\newblock \doi{10.1126/science.abo6587}.

\bibitem{Celi2020}
A.~Celi, B.~Vermersch, O.~Viyuela, H.~Pichler, M.~D. Lukin and P.~Zoller,
\newblock \emph{Emerging two-dimensional gauge theories in {Rydberg}
  configurable arrays},
\newblock Physical Review X \textbf{10}(2), 021057 (2020),
\newblock \doi{10.1103/physrevx.10.021057}.

\bibitem{Verresen_2021}
R.~Verresen, M.~D. Lukin and A.~Vishwanath,
\newblock \emph{Prediction of toric code topological order from {Rydberg}
  blockade},
\newblock Phys. Rev. X \textbf{11}, 031005 (2021),
\newblock \doi{10.1103/PhysRevX.11.031005}.

\bibitem{Samajdar_2021}
R.~Samajdar, W.~W. Ho, H.~Pichler, M.~D. Lukin and S.~Sachdev,
\newblock \emph{Quantum phases of {Rydberg} atoms on a kagome lattice},
\newblock Proceedings of the National Academy of Sciences \textbf{118}(4),
  e2015785118 (2021),
\newblock \doi{10.1073/pnas.2015785118}.

\bibitem{Moessner2001}
R.~Moessner, S.~L. Sondhi and E.~Fradkin,
\newblock \emph{Short-ranged resonating valence bond physics, quantum dimer
  models, and ising gauge theories},
\newblock Physical Review B \textbf{65}(2), 024504 (2001),
\newblock \doi{10.1103/physrevb.65.024504}.

\bibitem{Semeghini2021}
G.~Semeghini, H.~Levine, A.~Keesling, S.~Ebadi, T.~T. Wang, D.~Bluvstein,
  R.~Verresen, H.~Pichler, M.~Kalinowski, R.~Samajdar, A.~Omran, S.~Sachdev
  \emph{et~al.},
\newblock \emph{Probing topological spin liquids on a programmable quantum
  simulator},
\newblock Science \textbf{374}(6572), 1242 (2021),
\newblock \doi{10.1126/science.abi8794}.

\bibitem{Giudici2022}
G.~Giudici, M.~D. Lukin and H.~Pichler,
\newblock \emph{Dynamical preparation of quantum spin liquids in {Rydberg} atom
  arrays},
\newblock Physical Review Letters \textbf{129}(9), 090401 (2022),
\newblock \doi{10.1103/physrevlett.129.090401}.

\bibitem{Sahay2022}
R.~Sahay, A.~Vishwanath and R.~Verresen,
\newblock \emph{Quantum spin puddles and lakes: {NISQ}-era spin liquids from
  non-equilibrium dynamics}  (2022),
\newblock \doi{10.48550/arxiv.2211.01381}.

\bibitem{Kitaev2003}
A.~Kitaev,
\newblock \emph{Fault-tolerant quantum computation by anyons},
\newblock Annals of Physics \textbf{303}(1), 2 (2003),
\newblock \doi{10.1016/s0003-4916(02)00018-0}.

\bibitem{Levin2005}
M.~A. Levin and X.-G. Wen,
\newblock \emph{String-net condensation: {A} physical mechanism for topological
  phases},
\newblock Physical Review B \textbf{71}, 045110 (2005),
\newblock \doi{10.1103/physrevb.71.045110}.

\bibitem{Note3}
To prevent misconceptions, we stress that the term ``complex'' in ``Rydberg
  complex'' refers to a \protect \emph {spatial arrangement} of Rydberg atoms
  (with additional data) and is \protect \emph {not} related to the
  mathematical concept of an \protect \emph {independence complex}, i.e., the
  family of independent sets of a graph.

\bibitem{Lukin_2001}
M.~D. Lukin, M.~Fleischhauer, R.~Cote, L.~M. Duan, D.~Jaksch, J.~I. Cirac and
  P.~Zoller,
\newblock \emph{Dipole blockade and quantum information processing in
  mesoscopic atomic ensembles},
\newblock Physical Review Letters \textbf{87}(3), 037901 (2001),
\newblock \doi{10.1103/PhysRevLett.87.037901}.

\bibitem{Schauss2015}
P.~Schau{\ss}, J.~Zeiher, T.~Fukuhara, S.~Hild, M.~Cheneau, T.~Macr{\`{\i}},
  T.~Pohl, I.~Bloch and C.~Gross,
\newblock \emph{Crystallization in {I}sing quantum magnets},
\newblock Science \textbf{347}(6229), 1455 (2015),
\newblock \doi{10.1126/science.1258351}.

\bibitem{Labuhn2016}
H.~Labuhn, D.~Barredo, S.~Ravets, S.~de~L{\'{e}}s{\'{e}}leuc, T.~Macr{\`{\i}},
  T.~Lahaye and A.~Browaeys,
\newblock \emph{Tunable two-dimensional arrays of single {R}ydberg atoms for
  realizing quantum {I}sing models},
\newblock Nature \textbf{534}(7609), 667 (2016),
\newblock \doi{10.1038/nature18274}.

\bibitem{Bernien_2017}
H.~Bernien, S.~Schwartz, A.~Keesling, H.~Levine, A.~Omran, H.~Pichler, S.~Choi,
  A.~S. Zibrov, M.~Endres, M.~Greiner, V.~Vuleti{\'{c}} and M.~D. Lukin,
\newblock \emph{Probing many-body dynamics on a 51-atom quantum simulator},
\newblock Nature \textbf{551}(7682), 579 (2017),
\newblock \doi{10.1038/nature24622}.

\bibitem{Labuhn_2014}
H.~Labuhn, S.~Ravets, D.~Barredo, L.~B{\'{e}}guin, F.~Nogrette, T.~Lahaye and
  A.~Browaeys,
\newblock \emph{Single-atom addressing in microtraps for quantum-state
  engineering using {Rydberg} atoms},
\newblock Physical Review A \textbf{90}(2), 023415 (2014),
\newblock \doi{10.1103/PhysRevA.90.023415}.

\bibitem{Omran2019}
A.~Omran, H.~Levine, A.~Keesling, G.~Semeghini, T.~T. Wang, S.~Ebadi,
  H.~Bernien, A.~S. Zibrov, H.~Pichler, S.~Choi, J.~Cui, M.~Rossignolo
  \emph{et~al.},
\newblock \emph{Generation and manipulation of {Schrödinger} cat states in
  {Rydberg} atom arrays},
\newblock Science \textbf{365}(6453), 570 (2019),
\newblock \doi{10.1126/science.aax9743}.

\bibitem{Lesanovsky2011}
I.~Lesanovsky,
\newblock \emph{Many-body spin interactions and the ground state of a dense
  {Rydberg} lattice gas},
\newblock Physical Review Letters \textbf{106}(2), 025301 (2011),
\newblock \doi{10.1103/physrevlett.106.025301}.

\bibitem{Browaeys_2020}
A.~Browaeys and T.~Lahaye,
\newblock \emph{Many-body physics with individually controlled {Rydberg}
  atoms},
\newblock Nature Physics \textbf{16}(2), 132 (2020),
\newblock \doi{10.1038/s41567-019-0733-z}.

\bibitem{Davis1994}
M.~Davis, R.~Sigal and E.~Weyuker,
\newblock \emph{Computability, Complexity, and Languages: Fundamentals of
  Theoretical Computer Science},
\newblock Computer Science and Scientific Computing. Elsevier Science,
\newblock ISBN 9780080502465 (1994).

\bibitem{Kogut1979}
J.~B. Kogut,
\newblock \emph{An introduction to lattice gauge theory and spin systems},
\newblock Reviews of Modern Physics \textbf{51}, 659 (1979),
\newblock \doi{10.1103/revmodphys.51.659}.

\bibitem{Note4}
Here, \protect \emph {disjoint} means that $(x,y)\protect \neq (x',y')$ implies
  $x\protect \neq x'$ \protect \emph {and} $y\protect \neq y'$; thus $\gamma $
  can be interpreted as a \protect \emph {partial bijection} between character
  positions of the two languages $L_1$ and $L_2$.

\bibitem{Wernick1942}
W.~Wernick,
\newblock \emph{Complete sets of logical functions},
\newblock Transactions of the American Mathematical Society \textbf{51}(0), 117
  (1942),
\newblock \doi{10.1090/s0002-9947-1942-0005281-2}.

\bibitem{Sheffer1913}
H.~M. Sheffer,
\newblock \emph{A set of five independent postulates for {B}oolean algebras,
  with application to logical constants},
\newblock Transactions of the American Mathematical Society \textbf{14}(4), 481
  (1913),
\newblock \doi{10.1090/s0002-9947-1913-1500960-1}.

\bibitem{Dewdney_1979}
A.~K. Dewdney,
\newblock \emph{Logic circuits in the plane},
\newblock {ACM} {SIGACT} News \textbf{10}(3), 38 (1979),
\newblock \doi{10.1145/1113654.1113655}.

\bibitem{Stastny2023}
S.~Stastny,
\newblock \emph{Functional {Rydberg} Complexes in the {PXP}-Model},
\newblock Master's thesis, University of Stuttgart, Stuttgart (2023).

\bibitem{Bravyi1998}
S.~B. Bravyi and A.~Y. Kitaev,
\newblock \emph{Quantum codes on a lattice with boundary},
\newblock arXiv e-prints quant-ph/9811052 (1998),
\newblock \eprint{quant-ph/9811052}.

\bibitem{Kitaev2006}
A.~Kitaev,
\newblock \emph{Anyons in an exactly solved model and beyond},
\newblock Annals of Physics \textbf{321}(1), 2 (2006),
\newblock \doi{10.1016/j.aop.2005.10.005}.

\bibitem{Dennis2002}
E.~Dennis, A.~Kitaev, A.~Landahl and J.~Preskill,
\newblock \emph{Topological quantum memory},
\newblock Journal of Mathematical Physics \textbf{43}(9), 4452 (2002),
\newblock \doi{10.1063/1.1499754}.

\bibitem{Barends2014}
R.~Barends, J.~Kelly, A.~Megrant, A.~Veitia, D.~Sank, E.~Jeffrey, T.~C. White,
  J.~Mutus, A.~G. Fowler, B.~Campbell, Y.~Chen, Z.~Chen \emph{et~al.},
\newblock \emph{Superconducting quantum circuits at the surface code threshold
  for fault tolerance},
\newblock Nature \textbf{508}, 500 (2014),
\newblock \doi{10.1038/nature13171}.

\bibitem{Kelly2015}
J.~Kelly, R.~Barends, A.~G. Fowler, A.~Megrant, E.~Jeffrey, T.~C. White,
  D.~Sank, J.~Y. Mutus, B.~Campbell, Y.~Chen, Z.~Chen, B.~Chiaro \emph{et~al.},
\newblock \emph{State preservation by repetitive error detection in a
  superconducting quantum circuit},
\newblock Nature \textbf{519}, 66 (2015),
\newblock \doi{10.1038/nature14270}.

\bibitem{Kitaev2006a}
A.~Kitaev and J.~Preskill,
\newblock \emph{Topological entanglement entropy},
\newblock Physical Review Letters \textbf{96}, 110404 (2006),
\newblock \doi{10.1103/physrevlett.96.110404}.

\bibitem{Levin2006}
M.~Levin and X.-G. Wen,
\newblock \emph{Detecting topological order in a ground state wave function},
\newblock Physical Review Letters \textbf{96}, 110405 (2006),
\newblock \doi{10.1103/physrevlett.96.110405}.

\bibitem{Freedman2002}
M.~H. Freedman, A.~Kitaev, M.~J. Larsen and Z.~Wang,
\newblock \emph{Topological quantum computation},
\newblock Bulletin of the American Mathematical Society \textbf{40}(1), 31
  (2002),
\newblock \doi{10.1090/s0273-0979-02-00964-3}.

\bibitem{Nayak2008}
C.~Nayak, S.~H. Simon, A.~Stern, M.~Freedman and S.~Das~Sarma,
\newblock \emph{Non-abelian anyons and topological quantum computation},
\newblock Reviews of Modern Physics \textbf{80}, 1083 (2008),
\newblock \doi{10.1103/revmodphys.80.1083}.

\bibitem{Wang2010}
Z.~Wang,
\newblock \emph{Topological Quantum Computation},
\newblock No. 112 in Regional Conference Series in Mathematics / Conference
  Board of the Mathematical Sciences. American Mathematical Society,
  Providence, Rhode Island,
\newblock ISBN 9780821849309 (2010).

\bibitem{Freedman2002a}
M.~H. Freedman, M.~Larsen and Z.~Wang,
\newblock \emph{A modular functor which is universal for quantum computation},
\newblock Communications in Mathematical Physics \textbf{227}(3), 605 (2002),
\newblock \doi{10.1007/s002200200645}.

\bibitem{Preskill2004}
J.~Preskill,
\newblock \emph{Lecture notes for physics 219: Quantumcomputation} (2004).

\bibitem{Bonesteel2005}
N.~E. Bonesteel, L.~Hormozi, G.~Zikos and S.~H. Simon,
\newblock \emph{Braid topologies for quantum computation},
\newblock Physical Review Letters \textbf{95}(14), 140503 (2005),
\newblock \doi{10.1103/physrevlett.95.140503}.

\bibitem{Read1999}
N.~Read and E.~Rezayi,
\newblock \emph{Beyond paired quantum {Hall} states: Parafermions and
  incompressible states in the first excited {Landau} level},
\newblock Physical Review B \textbf{59}(12), 8084 (1999),
\newblock \doi{10.1103/physrevb.59.8084}.

\bibitem{Xia2004}
J.~S. Xia, W.~Pan, C.~L. Vicente, E.~D. Adams, N.~S. Sullivan, H.~L. Stormer,
  D.~C. Tsui, L.~N. Pfeiffer, K.~W. Baldwin and K.~W. West,
\newblock \emph{Electron correlation in the second {Landau} level: A
  competition between many nearly degenerate quantum phases},
\newblock Physical Review Letters \textbf{93}(17), 176809 (2004),
\newblock \doi{10.1103/physrevlett.93.176809}.

\bibitem{Fidkowski2009}
L.~Fidkowski, M.~Freedman, C.~Nayak, K.~Walker and Z.~Wang,
\newblock \emph{From string nets to nonabelions},
\newblock Communications in Mathematical Physics \textbf{287}(3), 805 (2009),
\newblock \doi{10.1007/s00220-009-0757-9}.

\bibitem{Fendley2008}
P.~Fendley,
\newblock \emph{Topological order from quantum loops and nets},
\newblock Annals of Physics \textbf{323}(12), 3113 (2008),
\newblock \doi{10.1016/j.aop.2008.04.011}.

\bibitem{Simon2013}
S.~H. Simon and P.~Fendley,
\newblock \emph{Exactly solvable lattice models with crossing symmetry},
\newblock Journal of Physics A: Mathematical and Theoretical \textbf{46}(10),
  105002 (2013),
\newblock \doi{10.1088/1751-8113/46/10/105002}.

\bibitem{Schulz2013}
M.~D. Schulz, S.~Dusuel, K.~P. Schmidt and J.~Vidal,
\newblock \emph{Topological phase transitions in the golden string-net model},
\newblock Physical Review Letters \textbf{110}(14), 147203 (2013),
\newblock \doi{10.1103/physrevlett.110.147203}.

\bibitem{Breu1998}
H.~Breu and D.~G. Kirkpatrick,
\newblock \emph{Unit disk graph recognition is {NP}-hard},
\newblock Computational Geometry \textbf{9}(1-2), 3 (1998),
\newblock \doi{10.1016/s0925-7721(97)00014-x}.

\bibitem{Virtanen2020}
P.~Virtanen, R.~Gommers, T.~E. Oliphant, M.~Haberland, T.~Reddy, D.~Cournapeau,
  E.~Burovski, P.~Peterson, W.~Weckesser, J.~Bright, S.~J. van~der Walt,
  M.~Brett \emph{et~al.},
\newblock \emph{{SciPy} 1.0: fundamental algorithms for scientific computing in
  python},
\newblock Nature Methods \textbf{17}(3), 261 (2020),
\newblock \doi{10.1038/s41592-019-0686-2}.

\bibitem{Tsallis1996}
C.~Tsallis and D.~A. Stariolo,
\newblock \emph{Generalized simulated annealing},
\newblock Physica A: Statistical Mechanics and its Applications
  \textbf{233}(1-2), 395 (1996),
\newblock \doi{10.1016/s0378-4371(96)00271-3}.

\bibitem{Andricioaei1996}
I.~Andricioaei and J.~E. Straub,
\newblock \emph{Generalized simulated annealing algorithms using tsallis
  statistics: Application to conformational optimization of a tetrapeptide},
\newblock Physical Review E \textbf{53}(4), R3055 (1996),
\newblock \doi{10.1103/physreve.53.r3055}.

\bibitem{Nelder1965}
J.~A. Nelder and R.~Mead,
\newblock \emph{A simplex method for function minimization},
\newblock The Computer Journal \textbf{7}(4), 308 (1965),
\newblock \doi{10.1093/comjnl/7.4.308}.

\bibitem{Gao2010}
F.~Gao and L.~Han,
\newblock \emph{Implementing the {Nelder-Mead} simplex algorithm with~adaptive
  parameters},
\newblock Computational Optimization and Applications \textbf{51}(1), 259
  (2010),
\newblock \doi{10.1007/s10589-010-9329-3}.

\bibitem{data}
S.~Stastny, H.~P. Büchler and N.~Lang,
\newblock \emph{{Data for ``Functional completeness of planar Rydberg blockade
  structures''}},
\newblock \doi{10.18419/darus-3307} (2023).

\bibitem{gross2013handbook}
J.~Gross, J.~Yellen and P.~Zhang,
\newblock \emph{Handbook of Graph Theory, Second Edition},
\newblock Discrete Mathematics and Its Applications. Taylor \& Francis,
\newblock ISBN 9781439880180 (2013).

\bibitem{McDiarmid2013}
C.~McDiarmid and T.~Müller,
\newblock \emph{Integer realizations of disk and segment graphs},
\newblock Journal of Combinatorial Theory, Series B \textbf{103}(1), 114
  (2013),
\newblock \doi{10.1016/j.jctb.2012.09.004}.

\bibitem{Kuhn2004}
F.~Kuhn, T.~Moscibroda and R.~Wattenhofer,
\newblock \emph{Unit disk graph approximation},
\newblock In \emph{Proceedings of the 2004 joint workshop on Foundations of
  mobile computing - {DIALM}-{POMC} {\textquotesingle}04}. {ACM} Press,
\newblock \doi{10.1145/1022630.1022634} (2004).

\bibitem{Bravyi_2011}
S.~Bravyi, D.~P. DiVincenzo and D.~Loss,
\newblock \emph{{S}chrieffer-{W}olff transformation for quantum many-body
  systems},
\newblock Annals of Physics \textbf{326}(10), 2793 (2011),
\newblock \doi{10.1016/j.aop.2011.06.004}.

\bibitem{Preskill2018}
J.~Preskill,
\newblock \emph{Quantum computing in the {NISQ} era and beyond},
\newblock Quantum \textbf{2}, 79 (2018),
\newblock \doi{10.22331/q-2018-08-06-79}.

\bibitem{Bennett1997}
C.~H. Bennett, E.~Bernstein, G.~Brassard and U.~Vazirani,
\newblock \emph{Strengths and weaknesses of quantum computing},
\newblock {SIAM} Journal on Computing \textbf{26}(5), 1510 (1997),
\newblock \doi{10.1137/s0097539796300933}.

\bibitem{Xiang1997}
Y.~Xiang, D.~Sun, W.~Fan and X.~Gong,
\newblock \emph{Generalized simulated annealing algorithm and its application
  to the thomson model},
\newblock Physics Letters A \textbf{233}(3), 216 (1997),
\newblock \doi{10.1016/s0375-9601(97)00474-x}.

\bibitem{Xiang2000}
Y.~Xiang and X.~G. Gong,
\newblock \emph{Efficiency of generalized simulated annealing},
\newblock Physical Review E \textbf{62}(3), 4473 (2000),
\newblock \doi{10.1103/physreve.62.4473}.

\end{thebibliography}

\clearpage


\onecolumngrid
\appendix

\section{Minimality of logic primitives}
\label{app:pxp}

Here we prove the claims in \cref{sec:completeness} and \cref{sec:primitives}
about the minimality of the logic primitives. The proofs in this section do
\emph{not} require geometric arguments (i.e.\ whether a given blockade graph is
a unit disk graph or not). This makes the claims independent of the embedding
dimension; in particular, they remain valid for three-dimensional complexes.

We start with a few general remarks.
First, the languages we seek to implement as ground state manifolds (GSM)
are \emph{irreducible} in the sense that they cannot be written as a product
of two smaller languages.  [The product of two formal languages is simply
the set of all words from the first concatenated with all words from the
second.] This is easy to check for all Boolean gates by inspecting their
truth tables. The crucial point is that irreducible languages can only be
implemented by complexes with \emph{connected} blockade graphs.

Second, because we are only interested in GSM of PXP models, all detunings
can be assumed to be strictly positive, $\Delta_i>0$.
Indeed, atoms with \emph{negative} detuning cannot be excited in the GSM
so that they can be deleted from the complex without changing the GSM (and
without closing the gap).
The argument against atoms with \emph{vanishing} detuning is more subtle.
If such an atom is not excited in any of the GSM states, it can be deleted
without changing the GSM. If it \emph{is} excited in some of the GSM states,
there is always an otherwise identical state in the GSM where it is not
excited. Such an atom therefore must be a port because as an ancilla it would
add internal degrees of freedom that are not accessible via the ports (this
follows from our definition of a complex). The language that corresponds
to a complex with a zero-detuning port therefore has the property that
for every word with a ``1'' at the corresponding position, there must be
an otherwise identical word with a ``0''. (This does not imply that the
language is reducible; for example, $L=\{111,011,000\}$ has this property
for the first letter but is irreducible.) While such languages do exist,
they cannot be truth tables of Boolean functions because such a port cannot
be used as an input or an output (assuming we forbid ``dummy'' inputs that
have no effect on the output). All languages discussed and implemented in this
paper (also the ones for the vertex complexes of spin liquids) do \emph{not}
have this property, hence we can assume non-vanishing detunings.

Because of the positivity of all detunings, ground states are always given by
\emph{maximal independent sets} (MIS*) of the blockade graph. [A \emph{maximal}
independent set is a subset of vertices such that (1) no two vertices of the
set are connected by an edge of the graph and (2) no vertex can be added to
the set without violating (1). \emph{Maximum} independent sets (MIS) are the
largest \emph{maximal} independent sets.] The inverse is not necessarily
true: Depending on the detunings, not every MIS* describes a ground state
configuration (an example is the ring-like \texttt{NOR}-complex).

\subsection{\texttt{CPY}-complex}
\label{app:cpy}

\begin{lemma}
    A \texttt{CPY}-complex cannot be realized with less than 4 atoms
    (1 ancilla).
\end{lemma}

\begin{proof}
    Assume there is a complex without ancillas described by
    \begin{align}
        H=-\Delta_1 n_1-\Delta_2 n_2-\Delta_3 n_3=: E_{n_1n_2n_3}\,.
    \end{align}
    Since $(n_1n_2n_3)=(111)$ must be a ground state of the complex,
    none of the pairs of the atoms can be in blockade so that there is
    no kinematic constraint on the configurations $(n_1n_2n_3)$. To be a
    \texttt{CPY}-complex, it must be
    \begin{align}
        -(\Delta_1+\Delta_2+\Delta_3)=E_{111}\stackrel{!}{=}E_{000}=0
        \quad\text{and}\quad
        E_{n_1n_2n_3}> 0
        \quad\text{for all}\quad(n_1n_2n_3)\neq (000),(111)\,.
    \end{align}
    The finite-gap condition requires in particular $\Delta_i>0$ for all
    $i=1,2,3$ which leads to $-(\Delta_1+\Delta_2+\Delta_3)<0$ and thereby
    contradicts the degeneracy condition.

    \emph{Alternative argument:} The copy language $L_\texttt{CPY}=\{000,111\}$
    is irreducible. Since $(111)$ must be in the GSM, the only admissible
    blockade graph is the trivial graph on three vertices without edges:
    $B=(V=\{1,2,3\},E=\emptyset)$. But a disconnected blockade graph cannot
    implement an irreducible language.
\end{proof}

\subsection{\texttt{NOR}-complex}
\label{app:nor}

\begin{lemma}
    A \texttt{NOR}-complex cannot be realized with less than 5 atoms
    (2 ancillas).
\end{lemma}

\begin{proof}
    We show that a \texttt{NOR}-complex cannot be realized with
    one ancilla or less. First, assume there is no ancilla so that the
    Hamiltonian is again
    \begin{align}
        H=-\Delta_1 n_1-\Delta_2 n_2-\Delta_3 n_3=:E_{n_1n_2n_3}\,,
    \end{align}
    now with potential kinematic constraints due to the Rydberg blockade. The
    conditions for a \texttt{NOR}-complex demand the equality of the following
    energies:
    \begin{subequations}
        \begin{align}
            E_{001}&=-\Delta_3\\
            E_{010}&=-\Delta_2\\
            E_{100}&=-\Delta_1\\
            E_{110}&=-\Delta_1-\Delta_2\,.
        \end{align}
    \end{subequations}
    It follows immediately $\Delta_1=\Delta_2=\Delta_3$ and $\Delta_1=0$
    so that all detunings must vanish. But then $(n_1n_2n_3)=(000)$
    is---independent of the configuration and its implied kinematic
    constraints---degenerate with the four states that belong to the
    \texttt{NOR}-manifold (which it must not be).

    \emph{Alternative argument:} The \texttt{NOR}-language
    $L_\texttt{NOR}=\{001,010,100,110\}$ is irreducible and forbids a blockade
    between the two input ports [because of $(110)$]. The only consistent
    blockade graph $B$ is therefore the line graph of three vertices.
    But this graph has only \emph{two} maximal independent sets, whereas we
    need at least \emph{four} to realize~$L_\texttt{NOR}$.

    So let us assume a system with one additional ancilla,
    \begin{align}
        H=-\Delta_1 n_1-\Delta_2 n_2-\Delta_3 n_3-\Delta_4 \tilde n_4\,,
    \end{align}
    and an arbitrary geometry that may lead to kinematic constraints on the
    allowed configurations.
    Let now $\vep(n_1n_2n_3)$ denote the \emph{minimal} energy of the system
    \emph{without} the contribution from the ports under the ``boundary
    condition'' that these are in the state $(n_1n_2n_3)$ and under the
    kinematic constraints imposed by the Rydberg blockade; furthermore, set
    $E_{n_1n_2n_3}:=-\Delta_1 n_1-\Delta_2 n_2-\Delta_3 n_3+\vep(n_1n_2n_3)$.
    In the current situation with only one ancilla, it is either
    $\vep(n_1n_2n_3)=0$ if the minimum is obtained by $\tilde n_4=0$,
    or $\vep(n_1n_2n_3)=-\Delta_4$ if $\tilde n_4=1$ minimizes the energy
    (and this is consistent with the configuration $(n_1n_2n_3)$).
    With this notation, the conditions to be a \texttt{NOR}-complex take
    the following form. First, the degeneracy of the \texttt{NOR}-manifold
    demands the equivalence of the following expressions:
    \begin{subequations}
        \begin{align}
            E_{001}&=-\Delta_3+\vep(001)\\
            E_{010}&=-\Delta_2+\vep(010)\\
            E_{100}&=-\Delta_1+\vep(100)\\
            E_{110}&=-\Delta_1-\Delta_2+\vep(110)\,,
        \end{align}
    \end{subequations}
    which immediately implies
    \begin{subequations}
        \begin{align}
            \Delta_1&=\vep(110)-\vep(010)\\
            \Delta_2&=\vep(110)-\vep(100)\\
            \Delta_3&=\vep(110)+\vep(001)-\vep(100)-\vep(010)\,.
        \end{align}
    \end{subequations}
    Second, the gap condition requires (among other conditions)
    \begin{subequations}
        \begin{align}
            \vep(000)=E_{000}&\stackrel{!}{>}E_{100}=-\Delta_1+\vep(100)=\vep(010)-\vep(110)+\vep(100)\\
            \Leftrightarrow\quad
            \vep(000)+\vep(110)&>\vep(010)+\vep(100)
            \label{eq:cond1}
        \end{align}
    \end{subequations}
    because a state with $(n_1n_2n_3)=(000)$ is not allowed in the
    \texttt{NOR}-manifold. Note that the only kinematic constraints on the
    ancilla in \cref{eq:cond1} can come from the two input vertices since
    $n_3=0$ for all four terms. We show now that \cref{eq:cond1} cannot be
    satisfied with a single ancilla.

    Consider first the case where $\Delta_4\leq 0$. Then the minimal energy
    under any condition $(n_1n_2n_3)$ is reached by switching the ancilla
    \emph{off}, $\tilde n_4=0$ (this is possible for all kinematic constraints),
    so that $0+0>0+0$ leads to a contradiction.
    Thus we have to assume $\Delta_4>0$ (this we could have anticipated from
    the arguments above). Now the energy can be lowered by switching the
    ancilla \emph{on}, but this might be forbidden by the kinematic constraints
    for certain boundary conditions $(n_1n_2n_3)$. We consider three cases:
    \begin{enumerate}[(i)]
        \item No blockade between the two inputs and the ancilla.
            In this case, the ancilla will be switched on in all four terms
            of \cref{eq:cond1} so that $-\Delta_4-\Delta_4>-\Delta_4-\Delta_4$
            violates the gap condition.
        \item The ancilla is in blockade with one of the inputs.
            W.l.o.g.\ let $n_1$ be in blockade with $\tilde n_4$. Then
            \cref{eq:cond1} reads $-\Delta_4+0>-\Delta_4+0$ which again
            violates the gap condition.
        \item The ancilla is in blockade with both inputs. Now
            \cref{eq:cond1} reads $-\Delta_4+0 > 0 + 0$ which is in
            contradiction with the assumption $\Delta_4>0$.
    \end{enumerate}
    In conclusion, we showed that it is impossible to satisfy the gap condition
    with a single ancilla. 

    \emph{Alternative argument:} Of the six connected graphs on four vertices,
    only the ``tetrahedron graph'' has four maximal independent sets (the
    others have at most three), which is necessary to realize the four words
    in $L_\texttt{NOR}$. But none of these four maximal independent sets
    contain more than one vertex [which would be necessary for $(110)$].
\end{proof}

\subsection{\texttt{AND}-complex, \texttt{OR}-complex and \texttt{XNOR}-complex}
\label{app:andorxnor}
 
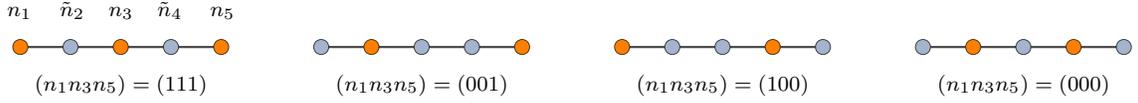
\begin{figure*}[t]
    \centering
    \begin{tikzpicture}
        \node[figtext,anchor=south] at (-1.35,0.3) {$n_1$};
        \node[figtext,anchor=south] at (-0.65,0.3) {$\tilde n_2$};
        \node[figtext,anchor=south] at (0,0.3) {$n_3$};
        \node[figtext,anchor=south] at (0.65,0.3) {$\tilde n_4$};
        \node[figtext,anchor=south] at (1.35,0.3) {$n_5$};
        \node[figure] at (0,0) {\includegraphics[width=3.0cm]{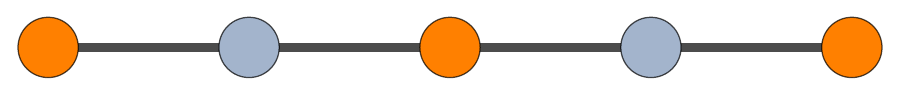}};
        \node[figtext,anchor=center] at (0,-0.5) {$(n_1n_3n_5)=(111)$};
        \node[figure] at (4,0) {\includegraphics[width=3.0cm]{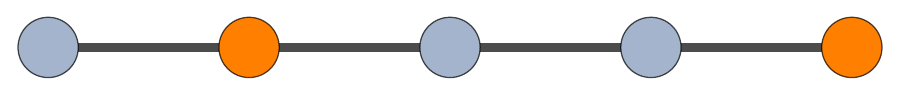}};
        \node[figtext,anchor=center] at (4,-0.5) {$(n_1n_3n_5)=(001)$};
        \node[figure] at (8,0) {\includegraphics[width=3.0cm]{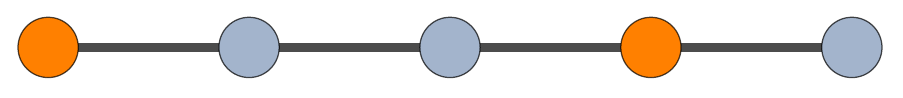}};
        \node[figtext,anchor=center] at (8,-0.5) {$(n_1n_3n_5)=(100)$};
        \node[figure] at (12,0) {\includegraphics[width=3.0cm]{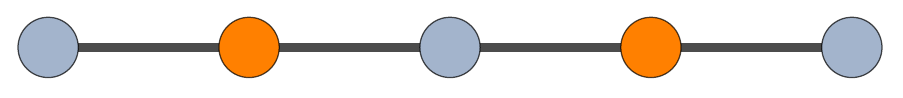}};
        \node[figtext,anchor=center] at (12,-0.5) {$(n_1n_3n_5)=(000)$};
        %
    \end{tikzpicture}
    \caption{%
        The line graph is the only connected blockade graph on five vertices
        with (at least) four maximal independent sets (orange vertices),
        (at least) one of which has (at least) three vertices. To realize
        the state $(111)$, the vertices $\{1,3,5\}$ must be chosen as ports,
        with $3$ as output; the four maximal independent sets then realize
        the truth table of \texttt{AND} (these four states cannot be made
        degenerate while maintaining a gap, see text).
    }
    \label{fig:and_or_xnor}
\end{figure*}

\begin{lemma}
    \texttt{AND}-, \texttt{OR}- and \texttt{XNOR}-complexes cannot be realized
    with less then 6 atoms (3 ancillas).
\end{lemma}

\begin{proof}
All these complexes contain the state $(111)$ such that no two ports can be
in blockade with each other. This implies that no realization of these gates
is possible with four or less atoms as the only connected blockade graph
which fulfills this constraint is the star graph of the \texttt{CPY}-complex
(which has only two MIS*).

The number of vertices is still small enough to systematically screen the
21 connected graphs on five vertices and select the 11 relevant ones with at
least four maximal independent sets. One can check that only the chain graph
has a MIS* with (at least) three vertices, which is needed to realize the
port configuration $(111)$ (\cref{fig:and_or_xnor}). This MIS* contains the
vertices $\{1,3,5\}$ of the chain, which we therefore must choose as ports:
$(n_1n_3n_5)=(111)$. With these ports, the set of four MIS* then realizes
the language $L=\{111,100,001,000\}$ which we identify as the truth table
of the \texttt{AND}-gate if we choose the port on the central atom 3 as
output. This proves that the \texttt{OR}- and \texttt{XNOR}-complex cannot
be realized with five atoms (even if another port is declared as output).

So far the arguments were purely \emph{kinematic} insofar as only the blockade
constraints and the knowledge that the GSM is generate by maximal independent
sets were used. To exclude the \texttt{AND}-gate, this is not enough, and
we have to use \emph{energetic} arguments by studying possible choices for
detunings. The degeneracy of the GSM requires the following four expressions
to be equal:
\begin{subequations}
    \begin{align}
        E_{111}&=-\Delta_1-\Delta_3-\Delta_5\\
        E_{100}&=-\Delta_1-\Delta_4\\
        E_{001}&=-\Delta_2-\Delta_5\\
        E_{000}&=-\Delta_2-\Delta_4\,,
    \end{align}
\end{subequations}
which immediately implies $\Delta_4=\Delta_5$ and therefore $\Delta_3=0$,
which is not allowed (remember that vanishing detunings are forbidden). This
proves that also the \texttt{AND}-complex cannot be realized with five atoms.
\end{proof}

\subsection{\texttt{NAND}-complex and \texttt{XOR}-complex}
\label{app:nandxor}
 
\begin{lemma}
    \texttt{NAND}- and \texttt{XOR}-complexes cannot be realized with less
    than 7 atoms (4 ancillas).
\end{lemma}

\begin{proof}
The truth tables of both \texttt{NAND} and \texttt{XOR} contain the states
$(110),(101)$ and $(011)$ so that no two ports can be in blockade with
each other. This excludes a realization with less than four atoms (see
\cref{app:andorxnor}). If two ancillas are available, we can switch one of the
input ports on; this switches (at least) one ancilla off. The remaining two
ports and (at most) one ancilla then must realize the \texttt{NOT}-language
$L_\neg=\{01,10\}$. This is impossible since the two ports cannot be directly
connected and the only blockade graph with a single ancilla realizes the
\texttt{LNK}-language $L_\texttt{LNK}=\{00,11\}$.
So let us assume that the complexes can be realized with three ports and
three ancillas. For the following arguments, only the edges between ports
and ancillas are of importance; potential blockades between ancillas can be
ignored. We consider three cases:
\begin{enumerate}[(i)]
\item There is at least one port that connects to all three ancillas. If this
port is on, all ancillas are off, hence the two remaining ports must be on
as well; but then at least two of the three states $(110),(101)$ and $(011)$
cannot be realized in the GSM.
\item There is at least one port that connects to a single ancilla. This
edge can be interpreted as an amalgamated \texttt{NOT}-complex. If we delete
the port, subtract its detuning from the connected ancilla, and declare the
latter as a new port, the new complex of five atoms realizes the truth table of
the original complex with one column inverted (w.l.o.g.\ the first one). For
both gates, this new manifold contains the states $(010),(001)$ and $(111)$
(plus another one that depends on the gate). The only blockade graph on
five vertices with at least four MIS*, one of which contains at least three
vertices [needed for $(111)$], has been identified in \cref{app:andorxnor}
as the line graph. There it has also been shown that there is no assignment
of detunings that realizes a four-fold degenerate GSM.
\item All inputs are connected with exactly two ancillas. There
are three possibilities to connect three ports with two ancillas each
(\cref{fig:nand_xor}). By inspection one shows that in all three cases there
is a pair of ports that, when activated, forces all ancillas connected to
the third port to be off; as this forces the third port to be on, at least
one of the states $(110),(101)$ and $(011)$ cannot be realized in the GSM.
\end{enumerate}
This proves that the \texttt{NAND}- and \texttt{XOR}-complex cannot be
realized with six atoms.
\end{proof}

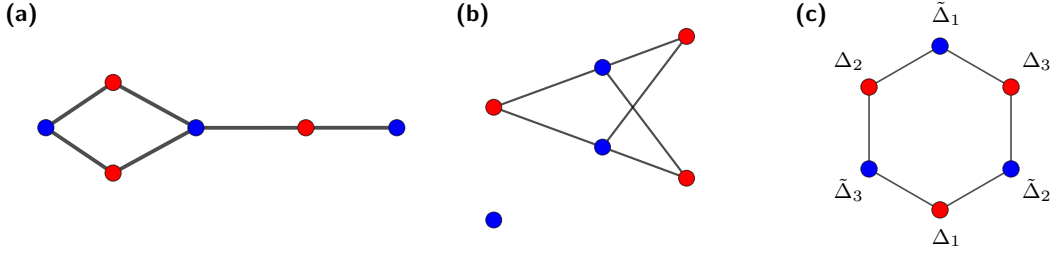
\begin{figure*}[t]
    \centering
    \begin{tikzpicture}
        \node[figtext,anchor=south] at (12.3,-1.65) {$\Delta_1$};
        \node[figtext,anchor=south] at (12.3,1.3) {$\tilde\Delta_1$};
        \node[figtext,anchor=south] at (11.0,0.7) {$\Delta_2$};
        \node[figtext,anchor=south] at (11.0,-1.05) {$\tilde\Delta_3$};
        \node[figtext,anchor=south] at (13.5,0.7) {$\Delta_3$};
        \node[figtext,anchor=south] at (13.5,-1.05) {$\tilde\Delta_2$};
        \node[label] at (0,1.5) {\lbl{a}};
        \node[figure,anchor=west] at (0,0) {\includegraphics[width=5.1cm]{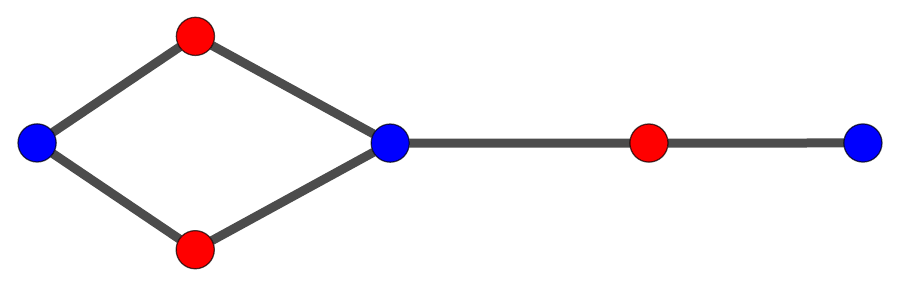}};
        \node[label] at (6,1.5) {\lbl{b}};
        \node[figure,anchor=west] at (6,0) {\includegraphics[width=2.9cm]{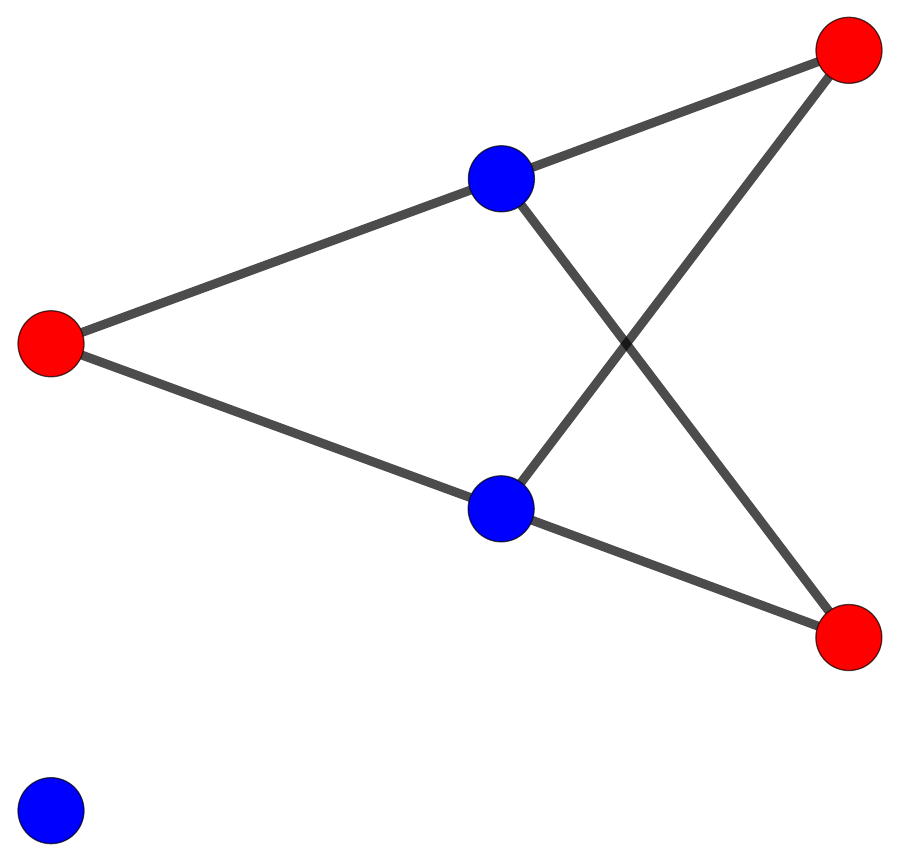}};
        \node[label] at (10.5,1.5) {\lbl{c}};
        \node[figure,anchor=west] at (11,0) {\includegraphics[width=2.2cm]{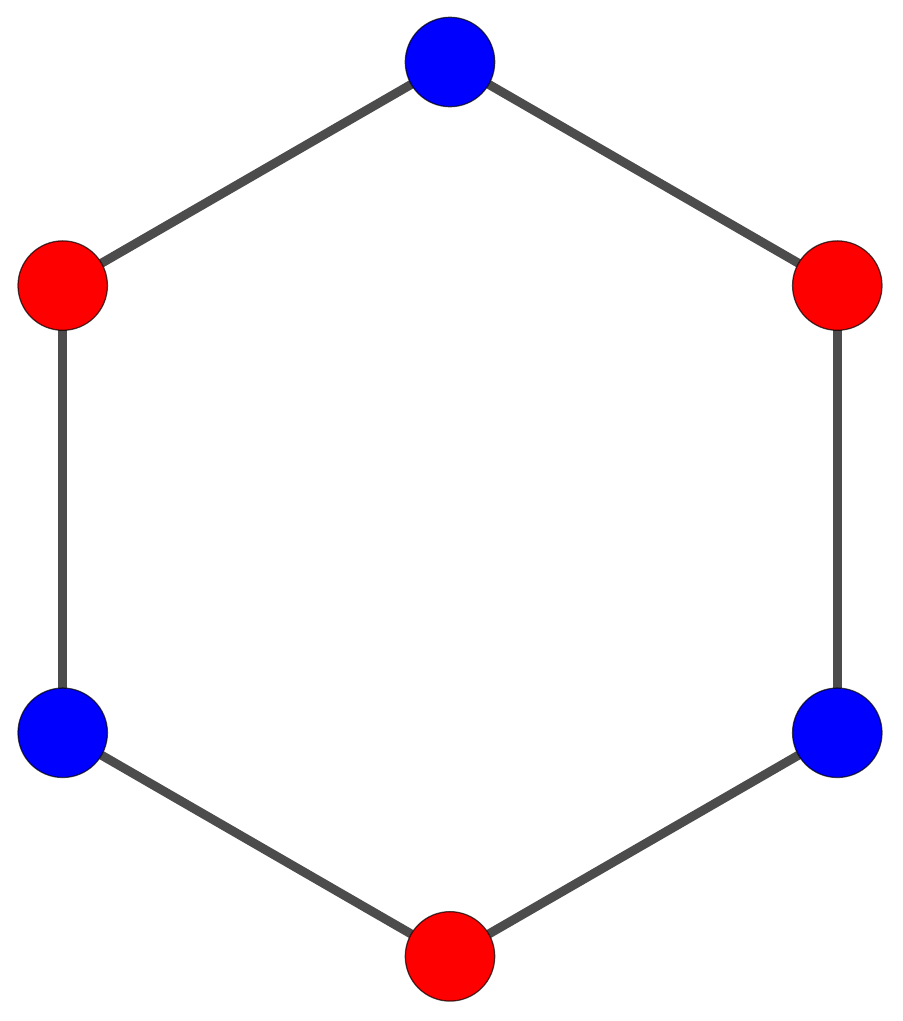}};
        %
    \end{tikzpicture}
    \caption{%
        The three bipartite graphs between three ports (red) and three
        ancillas (blue) where all ports have degree two. Note that these
        do not represent complete blockade graphs as we omit blockades
        between ancillas. The detunings in (c) are used in \cref{app:xnor}.
    }
    \label{fig:nand_xor}
\end{figure*}

\emph{Note:} Removing a \texttt{NOT}-complex by deleting the port, subtracting
its detuning from its ancilla, and declaring the ancilla as new port, is the
inverse of amalgamation; let us call it \emph{amputation}. One has to make sure
that the subtraction of the detuning of the port $\Delta_p$ from the detuning
of its adjacent ancilla $\Delta_a$ does not lead to negative (or vanishing)
detunings on the ancilla (= new port). Indeed, if $\Delta_p>\Delta_a$, the
port would be always on in all ground state configurations; this makes the port
superfluous and the language of the GSM reducible. If $\Delta_p=\Delta_a$, the
language of the original complex would have the property that for every word
with a ``0'' at the corresponding position, there is a otherwise identical word
with a ``1''. This is the dual property of the one discussed at the beginning
of \cref{app:pxp} and no language discussed in this paper has this property.

\subsection{Uniqueness of the blockade graph of the minimal \texttt{XNOR}-complex}
\label{app:xnor}

In contrast to the minimal \texttt{NOR}-complexes (for which there are
different blockade graph realizations), there is only one realization of the
minimal \texttt{XNOR}-complex. This will be useful in \cref{app:fibonacci}
to prove the minimality of the vertex complex of the Fibonacci model.
\begin{lemma}
    The blockade graph of the minimal \texttt{XNOR}-complex with 6 atoms
    (\cref{fig:logic_primitives}) is unique.
\end{lemma}
\begin{proof}
We showed in \cref{app:andorxnor} that a \texttt{XNOR}-complex needs at least
six atoms; so let us assume we have six atoms at our disposal. We now try to
contrive a complex that realizes the language $L_\odot=\{001,010,100,111\}$
systematically:
\begin{enumerate}[(i)]

\item Assume there exists such a complex with at least one port that connect
to only one ancilla. If this port is amputated, the remaining $5$ atoms
realize the \texttt{XOR}-language $\overline{L}_\odot=\{101,110,000,011\}$,
which is impossible as shown in \cref{app:nandxor}.

\item Assume at least one port connects to all three ancillas. If this port
is switched on, all ancillas are switched off and therefore the other two
ports must be active. This is inconsistent with one of the states $(001)$,
$(010)$ and $(100)$.

\item Because of (i) and (ii), only the case where all ports connect to two
ancillas remains. There are three classes of blockade graphs that satisfy
this, \cref{fig:nand_xor}. The first two graphs in \cref{fig:nand_xor} can
be immediately excluded as they are inconsistent with the states $(001)$,
$(010)$ and $(100)$ (= only \emph{one} port activated). Only the ``hexagon
graph'' in \cref{fig:nand_xor} remains as a possible blockade structure
between ports and ancillas. Without additional blockades between the
ancillas, the maximal independent sets of this graph allow for the states
$\{000,001,010,100,111\}\supset L_\odot$.

Let $\Delta_{1,2,3}$ denote the detunings of the three ports
and $\tilde\Delta_{1,2,3}$ the detunings of the three ancillas
(where $\tilde\Delta_i$ describes the ancilla opposite of port $i$,
\cref{fig:nand_xor}). In the state $(100)$, only the first port is excited. So
the opposite ancilla must be excited as well to block the two other ports
(if this ancilla were off, one could lower the energy by switching the other
two ports on). To balance this state energetically with the state $(111)$, the
detuning of the ancilla must equal the sum of the detunings of its two adjacent
ports. Due to the permutation symmetry of $L_\odot$ and the rotation symmetry
of the ``hexagon graph'', this argument is valid for all three ancillas:
\begin{align}
    \label{eq:ancilla_det}
    \tilde\Delta_1=\Delta_2+\Delta_3\,,\quad
    \tilde\Delta_2=\Delta_1+\Delta_3\,,\quad\text{and}\quad
    \tilde\Delta_3=\Delta_1+\Delta_2\,.
\end{align}
Because all detunings must be positive, this implies for any pair of ancillas
\begin{align}
    -\tilde\Delta_i-\tilde\Delta_j < -\Delta_1-\Delta_2-\Delta_3 = E_{111}\,.
\end{align}
Since $(111)$ must be in the GSM (i.e., $E_{111}$ must be the lowest allowed
energy), there must be an additional blockade between all pairs of ancillas
to prevent them from being excited simultaneously. This yields the blockade
graph of the \texttt{XNOR}-complex depicted in \cref{fig:logic_primitives}.
It has only four maximal independent sets that realize the language
$L_\odot=\{001,010,100,111\}$. The choices of the port detunings $\Delta_i>0$
are arbitrary; the ancilla detunings are then given by \cref{eq:ancilla_det}.
\end{enumerate}
We conclude that the blockade graph of the minimal realization of a
\texttt{XNOR}-complex with six atoms is unique (there is only freedom
in choosing the detunings). In addition, we proved that no \emph{strict
superset} of $L_\odot$ can be realized by a complex with six atoms or less
(this is used in \cref{app:fibonacci}).
\end{proof}

\section{Constructing subcomplexes}
\label{app:subcomplex}

Here we discuss a method to construct subcomplexes by fixing a port in
the active state and deleting its adjacent ancillas in the blockade graph.
This method is used in the proofs of \cref{app:examples} and the final remark
of \cref{sec:completeness}.
Consider a complex $\C$ that realizes a language $L$ with ports that are not
in blockade with each other. We can select one of the ports $p$ and define
the sublanguage $L_p\subset L$ of words $\vec x\in L$ with $x_p=1$. Our goal
is to construct a $L_p'$-complex $\C'_p$ where $L_p'$ is obtained from $L_p$
by deleting the constant letter at position $p$ that corresponds to the
fixed port.
The simplest solution is to keep the geometry of the complex $\C$ and increase
the detuning of the fixed port $\Delta_p$, thereby creating a gap between
states of the original GSM where the port is on and states where it is off;
the port can then be downgraded to an ancilla. In all states of the new
GSM this ancilla is active, while its adjacent ancillas are inactive. This
suggests that one can delete these atoms to obtain a smaller complex $\C_p'$
that realizes the same language $L_p'$:

\begin{lemma}
    Let the finite complex $\C$ realize the irreducible language $L$ with
    ports that are not in blockade with each other (with $\delta E=0$ and
    $\Delta E>0$).  Consider one of the ports $p$ with detuning $\Delta_p>0$
    and let the languages $L_p$ and $L_p'$ be defined as above. Then the
    structure $\C_p'$ obtained from $\C$ by deleting the port $p$ and all
    its adjacent ancillas is a $L_p'$-complex if the ports of $\C_p'$ are
    inherited from $\C$ in the natural way.
\end{lemma}
\begin{proof}

    First, note that since $L$ is irreducible, it is $L_p\neq\emptyset$, i.e.,
    there are configurations in the GSM of $\C$ where the port $p$ is active.
    We have to show two things: (a) the structure $\C_p'$ together with the
    inherited ports is a complex (i.e., its ground states can be labeled by
    the configurations of the ports), and (b) the language that describes
    this GSM is $L_p'$. 
    
    Let the GSM of the new structure $\C_p'$ be defined by $\delta
    E=0$ (since the structure is finite, it is automatically $\Delta E
    >0$). Every kinetically allowed (= admissible) configuration in this
    GSM can be \emph{extended} to an admissible configuration of $\C$
    by setting the deleted ancillas to \emph{off} and the port $p$ to
    \emph{on}. If $E_0(\C_p')$ denotes the ground state energy of $\C_p'$
    and $E_0(\C)$ the same for $\C$, this implies that $E_0(\C)\leq
    E_0(\C_p')-\Delta_p$. Conversely, because $L_p\neq\emptyset$, there
    are admissible configurations in the GSM of $\C$ where the port $p$ is
    \emph{on} and, consequently, all adjacent ancillas are \emph{off}. By
    \emph{truncating} the configurations of the adjacent ancillas and
    the port $p$, this yields a admissible configuration for $\C_p'$ with
    energy $E_0(\C)+\Delta_p$ so that $E_0(\C_p')\leq E_0(\C)+\Delta_p$. In
    combination, we have
    \begin{align}
        E_0(\C_p')=E_0(\C)+\Delta_p
    \end{align}
    for the ground state energy of the new structure $\C_p'$.
    Using this result and the mappings of \emph{extension} and
    \emph{truncation}, we can draw two conclusions:
    \begin{enumerate}[(1)]

        \item Every configuration in the GSM of $\C_p'$ can be \emph{extended}
        to a configuration in the GSM of $\C$ which corresponds to a word
        in $L_p$. We can immediately conclude two things:
        \begin{enumerate}[(i)]

            \item Since the extended configurations must be distinguishable
            by the ports of the \emph{complex} $\C$ ignoring port $p$
            (this port is always \emph{on} for configurations in $L_p$),
            and because these ports are inherited by the structure $\C_p'$,
            we can conclude that the configurations of the GSM of $\C_p'$
            can also be distinguished by these ports. This makes $\C_p'$
            a \emph{complex} that realizes some language $L^{?}$.

            \item Every word in $L^{?}$ is mapped by the extension to a word
            in $L_p$ which implies $L^{?}\subseteq L_p'$.

        \end{enumerate}
        
        \item Conversely, every configuration in the GSM of $\C$ which
        corresponds to a word in $L_p$ can be \emph{truncated} to an admissible
        configuration of $\C_p'$ with energy $E_0(\C)+\Delta_p=E_0(\C_p')$,
        which implies $L_p'\subseteq L^{?}$.

    \end{enumerate}
    In conclusion, we showed that $L^{?}=L_p'$ and therefore that $\C_p'$
    is indeed a $L_p'$-complex.
\end{proof}

\section{Minimality of spin liquid primitives}
\label{app:examples}

\subsection{Vertex/Unit cell complex for the surface code ($\C_\texttt{SCU}$)}
\label{app:toric}

\begin{lemma}
    The vertex complex (unit cell complex) $\C_\texttt{SCU}$ of the surface
    code on the square lattice cannot be realized with less than 11 atoms.
\end{lemma}

\begin{proof}

Here we show that the vertex complex of the surface code on the square
lattice requires at least 11 atoms; to this end, we use and expand on the
tricks introduced in \cref{app:nandxor}.
First, note that the GSM is symmetric under the permutation of ports
(\cref{fig:toric}b) and includes the state $(1111)$, i.e., no two ports can
be in blockade with each other. In addition, the GSM is symmetric under the
simultaneous inversion of an even number of letters in all words (= columns):
\begin{align}
    L_\texttt{SCU}\equiv\begin{Bmatrix}1111\\1100\\0011\\1001\\0110\\0101\\1010\\0000\\\end{Bmatrix}
    \;\xrightarrow{\text{inv.\ 4.\ letter}}\;
    \overline{L}_\texttt{SCU}\equiv\begin{Bmatrix}1110\\1101\\0010\\1000\\0111\\0100\\1011\\0001\\\end{Bmatrix}
    \;\xrightarrow{\text{inv.\ 3.\ letter}}\;
    L_\texttt{SCU}=\begin{Bmatrix}1100\\1111\\0000\\1010\\0101\\0110\\1001\\0011\\\end{Bmatrix}
    \;\xrightarrow{\text{inv.\ 2.\ letter}}\;
    \overline{L}_\texttt{SCU}=\begin{Bmatrix}1000\\1011\\0100\\1110\\0001\\0010\\1101\\0111\\\end{Bmatrix}
    \;\cdots
    \label{eq:symmetry}
\end{align}

Let us now systematically exclude the existence of surface code complexes
with $N\leq 10$ atoms:
\begin{itemize}

\item $N<8$: If one fixes one port of a surface code complex as active, the
remaining three ports realize a \texttt{XNOR}-complex with at least two atoms
less than the surface code complex (because the active port deactivates at
least one ancilla permanently). Since we proved in \cref{app:andorxnor} that
\texttt{XNOR}-complexes require at least six atoms, this implies immediately
that the surface code complex cannot be realized with $N<8$ atoms.

\item $N=8$: If there are at least two ports that are connected to only one
ancilla each, we can consider these as amalgamated \texttt{NOT}-complexes and
amputate two of them (see the note in \cref{app:nandxor}), thereby creating a
complex with only six atoms that realizes the same GSM due to the inversion
symmetry detailed in \cref{eq:symmetry}. Since this is not possible, there
can be at most one port that connects to only one ancilla. Choose one of the
other ports that connect to at least two ancillas and again fix it in the
active state (here we use the permutation symmetry of $L_\texttt{SCU}$). This
produces a complex with at most five atoms (the fixed port plus at least two
ancillas are removed from the surface code complex) that realizes again the
\texttt{XNOR}-manifold, which is impossible. Hence the surface code complex
cannot be realized with eight atoms.

\item $N=9$: To show that the complex cannot be realized with nine atoms we
consider three cases:
\begin{enumerate}[(i)]
\item Assume there is at least one port connected to three or more ancillas. If
this port is fixed as active, it blocks at least three ancillas. The resulting
\texttt{XNOR}-complex on the three remaining ports has at most five atoms,
which is impossible.
\item Assume at least one port connects to a single ancilla. Amputating this
port yields a $\overline{L}_\texttt{SCU}$-complex with eight atoms. If an
arbitrary port of this complex is fixed as active, the resulting complex has
at most six atoms. Inspection of the language $\overline{L}_\texttt{SCU}$
[\cref{eq:symmetry}] shows that this complex realizes the truth table of a
\texttt{XOR}-gate, which, however, requires at least seven atoms (as shown
in \cref{app:nandxor}).

\begin{figure*}[t]
    \centering
    \begin{tikzpicture}
        \node[label] at (0,3.2) {\lbl{a}};
        \node[figure,anchor=north west] at (0,3.2) {\includegraphics[width=2.0cm]{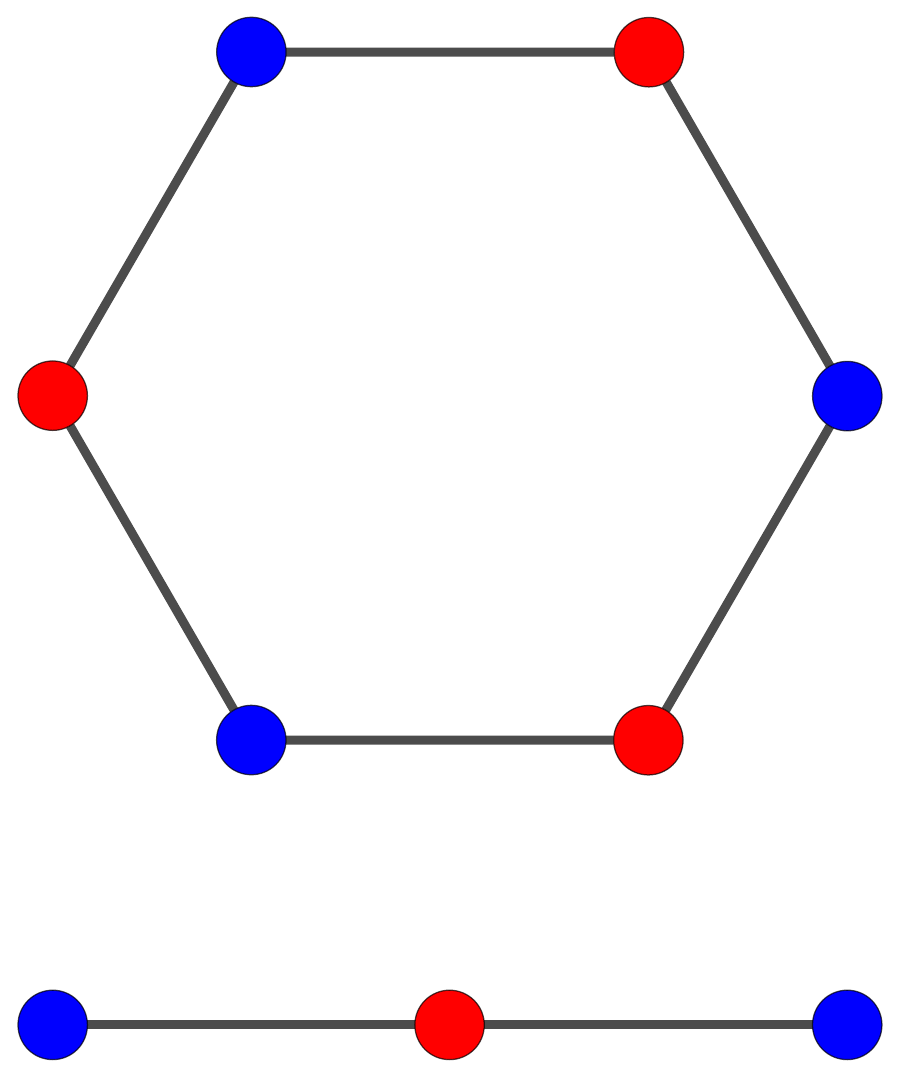}};
        \node[label] at (5,3.2) {\lbl{b}};
        \node[figure,anchor=north west] at (5,3.2) {\includegraphics[width=3.7cm]{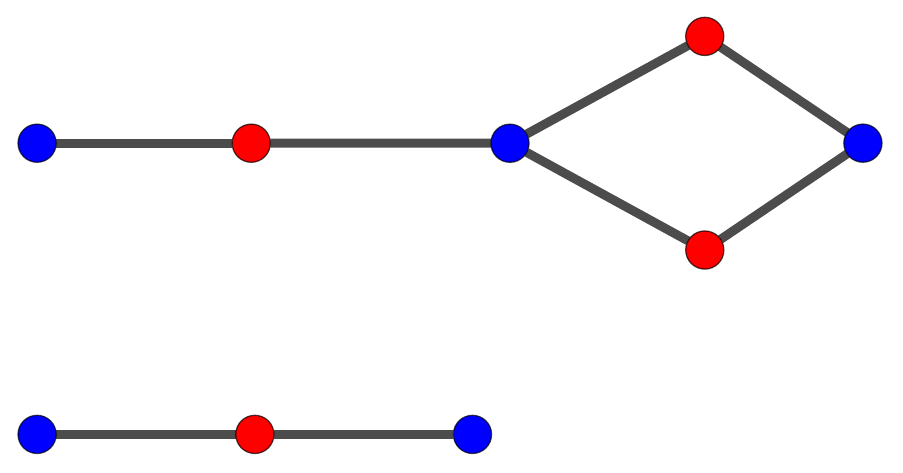}};
        \node[label] at (10,3.2) {\lbl{c}};
        \node[figure,anchor=north west] at (10.4,3.2) {\includegraphics[width=2.5cm]{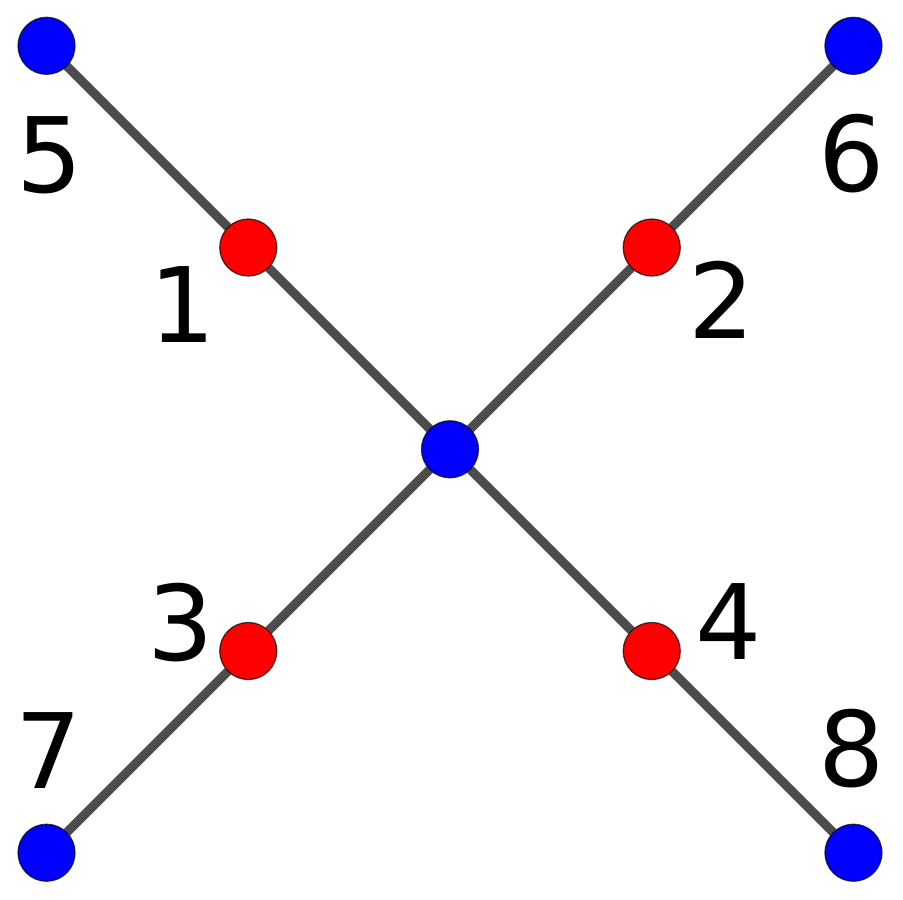}};
        \node[label] at (0,0) {\lbl{d}};
        \node[figure,anchor=north east,rotate=90] at (0,0.0) {\includegraphics[width=1.9cm]{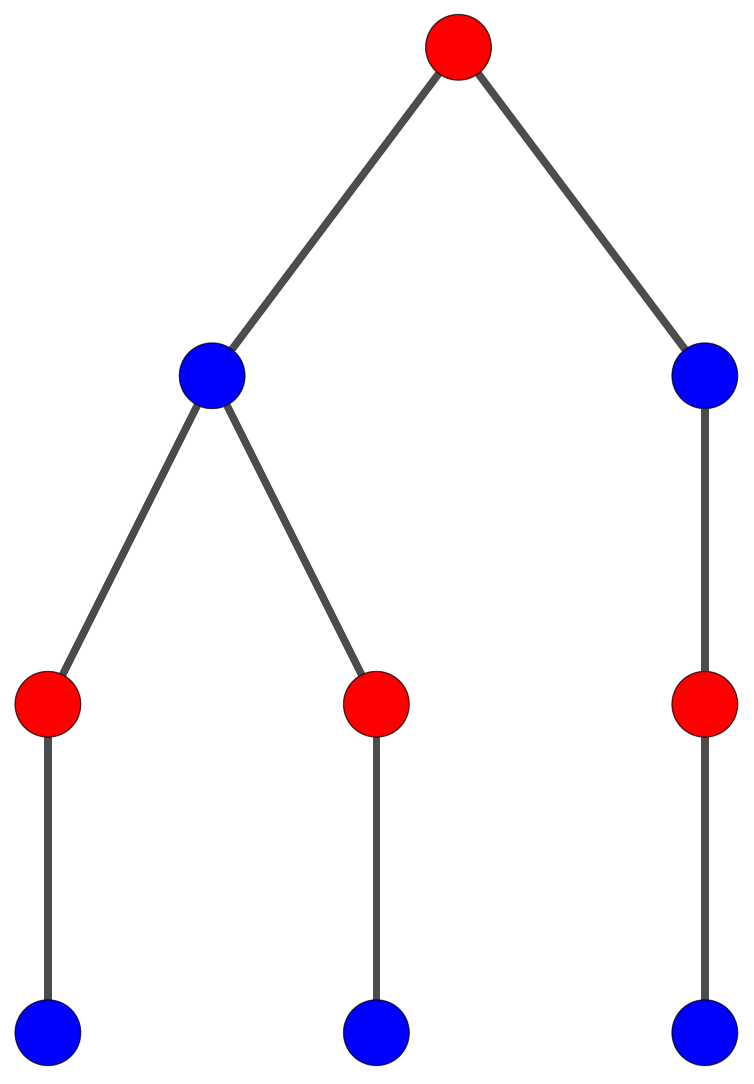}};
        \node[label] at (5,0) {\lbl{e}};
        \node[figure,anchor=north west] at (5,-0.8) {\includegraphics[width=3.7cm]{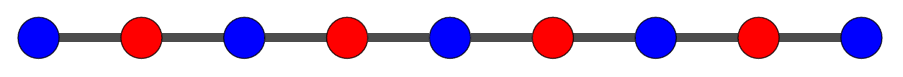}};
        \node[label] at (10,0) {\lbl{f}};
        \node[figure,anchor=north west] at (10.4,0.0) {\includegraphics[width=4.7cm]{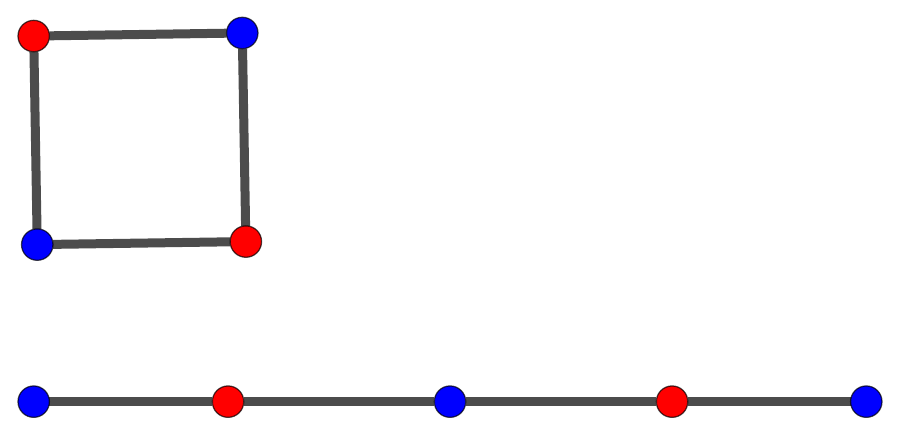}};
        %
    \end{tikzpicture}
    \caption{The remaining six classes of blockade graphs for the surface
    code vertex complex on nine atoms with ports of degree 2 and without
    disconnected ancillas. Ports (ancillas) are colored red (blue) and
    connections between ancillas are omitted. The only class that cannot
    be excluded kinematically is the ``cross'' graph (c) with atom labels
    $\{1,\dots,8\}$ and ports $\{1,2,3,4\}$, see text.}
    \label{fig:graphsTOR92}
\end{figure*}
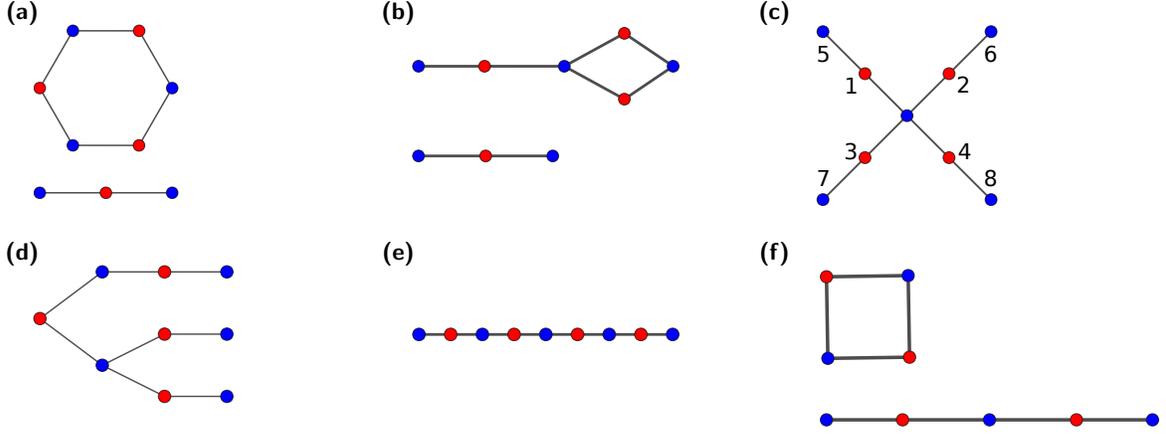

\item Because of (i) and (ii), only the case that all ports connect to exactly
two ancillas remains. There are six non-isomorphic bipartite graphs that
connect sets of four (ports) and five vertices (ancillas), where all ports
have degree 2, \cref{fig:graphsTOR92}. We exclude graphs with disconnected
ancillas because these are typically covered by the analogous step for $N=8$.
(Above we omitted this step to simplify the prove, so in principle here one
has to check the graphs with disconnected ancillas too. The result is the
same, though.) By inspection, one shows that all these graphs (except for the
``cross'' in \cref{fig:graphsTOR92}c) allow for a pair of ports that, when
activated, block all ancillas of a third port (which then must be switched on
as well).  This, however, is inconsistent with the language $L_\texttt{SCU}$
which includes for all triples of ports states where two are on and one is off.

The ``cross'' graph in \cref{fig:graphsTOR92}c cannot be excluded with this
type of kinematic reasoning because the set of maximal independent sets (with
the convention of ports shown in \cref{fig:graphsTOR92}c) induces a superset
of $L_\texttt{SCU}$. Therefore we have to use energetic arguments instead. With
the atom indices shown in \cref{fig:graphsTOR92}c, the gap condition requires
\begin{subequations}
    \begin{align}
        E_{1111}&=-\Delta_1-\Delta_2-\Delta_3 - \Delta_4 
        \stackrel{!}{<}
        -\Delta_1-\Delta_2-\Delta_3-\Delta_8=E_{1110}
        \quad\Rightarrow\quad\Delta_8 < \Delta_4\,,
        \\
        E_{1100}&=-\Delta_1-\Delta_2-\Delta_7 - \Delta_8 
        \stackrel{!}{<}
        -\Delta_1-\Delta_2-\Delta_7-\Delta_4=E_{1101}
        \quad\Rightarrow\quad\Delta_8 > \Delta_4\,.
    \end{align}  
\end{subequations}
Hence this graph cannot realize the $L_\texttt{SCU}$-manifold.
\end{enumerate}

\item $N=10$: To show that the surface code complex cannot be realized with
$N=10$ atoms, one follows the same procedure as detailed above for the case
of $N=9$ atoms (here we only briefly summarize the necessary steps): First,
one excludes the case with ports that connect to a single ancilla (where one
has to use that a $N=9$ realization of a $\overline{L}_\texttt{SCU}$-complex
can have only ports that connect to at least two ancillas). Then, one excludes
the existence of ports that connect to at least four ancillas by using that
\texttt{XNOR}-complexes cannot be realized with five atoms or less. Finally,
one must exclude blockade graphs with ports of degree three or two by the
same procedure as in Step (iii) above. In this case, there are 20 graph
classes to cover of which 15 can be kinematically excluded and 5 can be
energetically ruled out. These arguments show that the surface code vertex
complex cannot be realized with $10$ atoms.
\end{itemize}
\end{proof}

\subsection{Vertex complex for the Fibonacci model ($\C_{f_\sub{Fib}}$)}
\label{app:fibonacci}

\begin{figure*}[t]
    \centering
    \begin{tikzpicture}
        \node[label] at (0.0,1.2) {\lbl{a}};
        \node[figure,anchor=west] at (0,0) {\includegraphics[width=4.0cm]{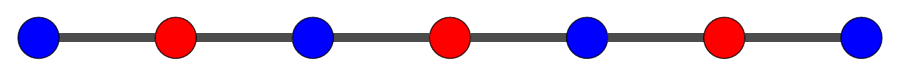}};
        \node[label] at (5.5,1.2) {\lbl{b}};
        \node[figure,anchor=west] at (6,0) {\includegraphics[width=2.3cm]{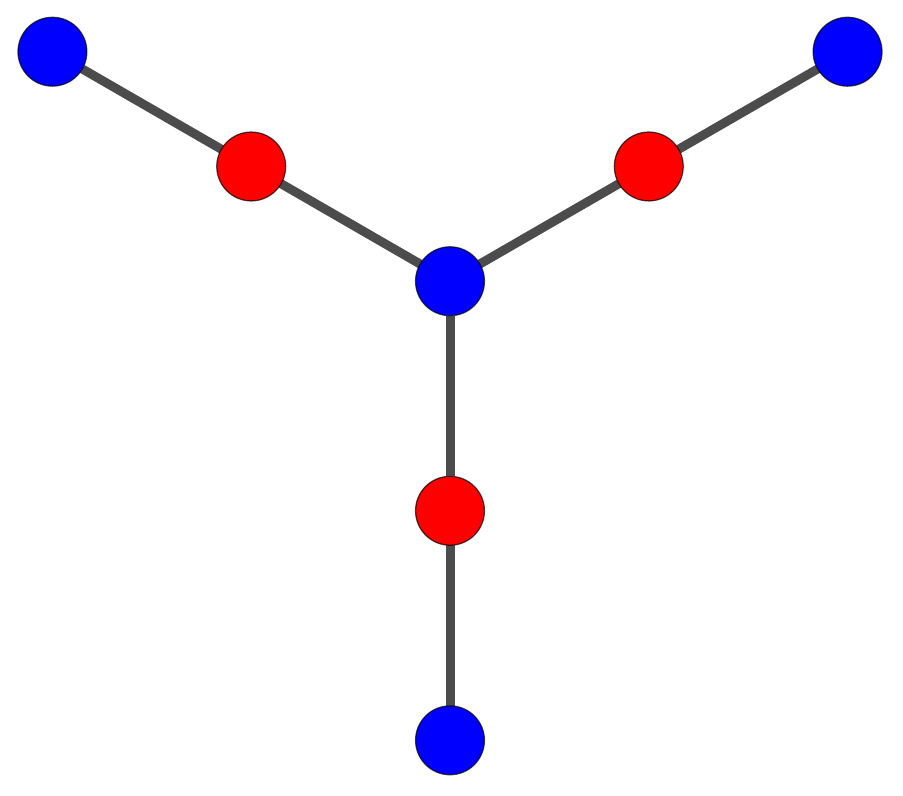}};
        \node[label] at (10.5,1.2) {\lbl{c}};
        \node[figure,anchor=west] at (11.0,0) {\includegraphics[width=2.6cm]{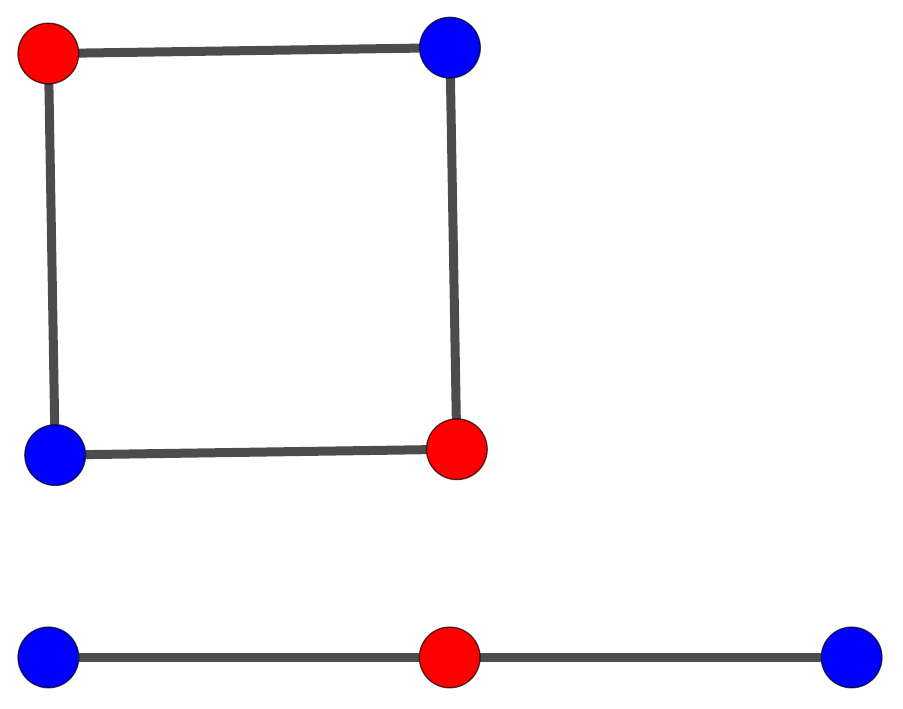}};
        %
    \end{tikzpicture}
    \caption{The remaining three classes of blockade graphs for the Fibonacci
    vertex complex on seven atoms with ports of degree 2. Ports (ancillas)
    are colored red (blue) and connections between ancillas are omitted. Only
    the graph in (b) must be energetically excluded.}
    \label{fig:graphsFIB72}
\end{figure*}
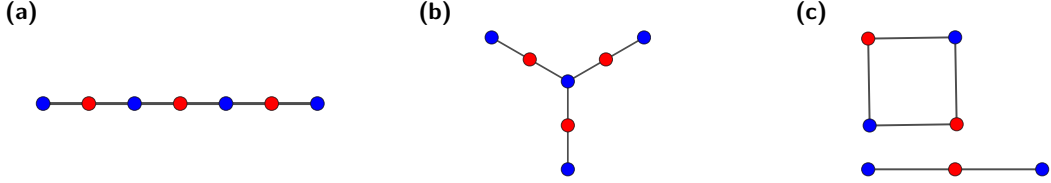

\begin{lemma}
    The vertex complex $\C_{f_\sub{Fib}}$ of the Fibonacci model on the
    Honeycomb lattice cannot be realized with less than 8 atoms.
\end{lemma}

\begin{proof}

The Fibonacci language $L_\sub{Fib}=\{000,011,110,101,111\}$ contains the three
states $(011)$, $(110)$ and $(101)$ which we used in \cref{app:nandxor}
to show (with purely kinematic arguments) that \texttt{XOR}- and
\texttt{NAND}-complexes cannot be realized with less than seven atoms. As
the argument only relied on these three states (and the existence of at least
one other state), it extends to the Fibonacci complex, which therefore also
requires at least seven atoms. Furthermore, the three states forbid blockades
between any two ports.
So assume a realization with seven atoms exists. We distinguish three cases:
\begin{enumerate}[(i)]

\item At least one port connects to a single ancilla. Amputation (see the
note in \cref{app:nandxor}) of this port yields a complex with six atoms
that realizes the manifold $\overline{L}_\sub{Fib}=\{100,111,010,001,011\}$
obtained from $L_\sub{Fib}$ by inverting the first letter (note that
$L_\sub{Fib}$ is symmetric under permutations of ports). This language
contains all \texttt{XNOR} states: $L_\odot\subset\overline{L}_\sub{Fib}$. In
\cref{app:xnor} we showed that there is only one blockade graph on 6
vertices that can realize these states; this graph has only four maximal
independent sets and therefore cannot realize the additional state $(011)$
in $\overline{L}_\sub{Fib}$.
\item At least one port connects to at least three ancillas. First, a port
that connects to all \emph{four} ancillas is inconsistent with two of the
three states $(011)$, $(110)$ and $(101)$. So assume there is a port that
connects to \emph{three} of the four ancillas. If this port is activated,
it deactivates three ancillas and the remaining two ports (together with one
ancilla) realize the irreducible language $L=\{01,10,11\}$ (the blockade
graph of these three atoms must therefore be connected). But there is no
graph on three vertices with (at least) three maximal independent sets of
which (at least) one has (at least) two vertices.

\item Because of (i) and (ii), only the case where all ports connect to two
ancillas remains. The possible classes of bipartite graphs are shown in
\cref{fig:graphsFIB72}. With a similar line of arguments as used for the
surface code [Case (iii) for $N=9$ in \cref{app:toric}], one can exclude
two of the three graphs (a and c) with kinematic arguments [using the states
$(011)$, $(110)$ and $(101)$]. The set of maximal independent sets for the
graph in \cref{fig:graphsFIB72}b includes $L_\sub{Fib}$ as a subset and can
again be excluded by energetic arguments.
\end{enumerate}
\end{proof}

\section{Numerical approach for geometric optimization}
\label{app:optimization}

To minimize the objective function $\Gamma$ on the high-dimensional
configuration space $\mathfrak{C}_N$, we used the \texttt{SciPy} method
\texttt{scipy.optimize.dual\_annealing}~\cite{Xiang1997,Xiang2000,Virtanen2020}
that implements generalized simulated
annealing~\cite{Tsallis1996,Andricioaei1996} in combination with a local
optimization based on the Nelder-Mead algorithm~\cite{Nelder1965,Gao2010}.
The stochastic algorithm starts from an initial geometry (which can be
chosen randomly), followed by iterations of jumps in $\mathfrak{C}_N$ with
probabilities that depend on the distance of the jump and the variation of
the objective function $\Gamma$; following each random jump, the Nelder-Mead
algorithm optimizes the new configuration locally. After $\lesssim 2000$
iterations we stop the algorithm and compute the robustness of the final
geometry. More technical details and all obtained optimal complexes can be
found in Ref.~\cite{Stastny2023}; the data of the optimized complexes can
also be accessed online~\cite{data}.

\end{document}